\documentclass[prx,aps,superscriptaddress,footinbib,twocolumn]{revtex4-2}

\usepackage{mathrsfs}
\usepackage{amsfonts}
\usepackage{amssymb}
\usepackage{amsmath}
\usepackage{graphicx}
\usepackage[usenames,dvipsnames]{xcolor}
\usepackage[colorlinks=true,citecolor=blue,linkcolor=magenta]{hyperref}
\usepackage{cleveref}
\usepackage{ulem}
\usepackage{lmodern}
\usepackage{amsthm}
\usepackage{dsfont}
\usepackage{natbib}
\usepackage{cancel}

\usepackage{mathtools}
\usepackage[noend]{algpseudocode}
\usepackage{algorithm}

\usepackage{stackengine}
\usepackage{scalerel}
\def\rlwd{.4pt}
\def\rlht{1.1pt}
\def\shatvrule{\rule{\rlwd}{\rlht}}
\def\shat#1{%
 \ThisStyle{%
  \setbox0=\hbox{$\SavedStyle#1$}%
  \stackon[0pt]{\stackon[1pt]{\ensuremath{\SavedStyle#1}}{%
    \shatvrule\kern\wd0\kern-\rlwd\kern-\rlwd\shatvrule}}%
    {\rule{\wd0}{\rlwd}}%
 }%
}

\newcommand{\Prob}[1]{\operatorname{Pr}(#1)}
\newcommand{\Var}[1]{\operatorname{Var}\left[{ #1 }\right]}

\newtheorem{theorem}{Theorem}

\newtheorem{proposition}{Proposition}

\newcommand{\Tr}{\mathrm{Tr}}
\newcommand{\norm}[1]{\Vert #1 \Vert}

\newcommand{\absLR}[1]{\left\vert #1 \right\vert}

\newcommand{\ket}[1]{\vert{ #1 }\rangle}

\newcommand{\ketbra}[2]{\vert #1 \rangle \langle #2 \vert}

\newlength{\ketketwidth}
\newlength{\ketwidth}

\newcommand{\kettstylesep}[3]{
    \settowidth{\ketwidth}{$#2\left|#1\right\rangle$}
    \settowidth{\ketketwidth}{$#2\left.\left|#1\right\rangle\right\rangle$}
    \left|#1\right\rangle#3\hspace{\ketwidth}\hspace{-\ketketwidth}
}
\newcommand{\dket}[1]{
    \left.\mathchoice
        {\kettstylesep{#1}{\displaystyle}{\hspace{0.3em}}}
        {\kettstylesep{#1}{\textstyle}{\hspace{0.3em}}}
        {\kettstylesep{#1}{\scriptstyle}{\hspace{0.3em}}}
        {\kettstylesep{#1}{\scriptscriptstyle}{\hspace{0.25em}}}
    \right\rangle
}

\newcommand{\bbrastylesep}[3]{
    \settowidth{\ketwidth}{$#2\left\langle#1\right|$}
    \settowidth{\ketketwidth}{$#2\left\langle\left\langle#1\right|\right.$}
    #3\hspace{\ketwidth}\hspace{-\ketketwidth}\left\langle#1\right|
}
\newcommand{\dbra}[1]{
    \left\langle\mathchoice
        {\bbrastylesep{#1}{\displaystyle}{\hspace{0.3em}}}
        {\bbrastylesep{#1}{\textstyle}{\hspace{0.3em}}}
        {\bbrastylesep{#1}{\scriptstyle}{\hspace{0.3em}}}
        {\bbrastylesep{#1}{\scriptscriptstyle}{\hspace{0.25em}}}
    \right.
}

\newcommand{\dketbra}[2]{\dket{#1}\hspace{-0.2em}\dbra{#2}}

\newcommand{\bbrakettstylesep}[4]{
    \settowidth{\ketwidth}{$#3\left\langle#1\middle|#2\right\rangle$}
    \settowidth{\ketketwidth}{$#3\left\langle\left\langle#1\middle|#2\right\rangle\right.$}
    #4\hspace{\ketwidth}\hspace{-\ketketwidth}\left\langle#1\middle|#2\right\rangle#4\hspace{\ketwidth}\hspace{-\ketketwidth}
}
\newcommand{\dbraket}[2]{
    \left\langle\mathchoice
        {\bbrakettstylesep{#1}{#2}{\displaystyle}{\hspace{0.3em}}}
        {\bbrakettstylesep{#1}{#2}{\textstyle}{\hspace{0.3em}}}
        {\bbrakettstylesep{#1}{#2}{\scriptstyle}{\hspace{0.3em}}}
        {\bbrakettstylesep{#1}{#2}{\scriptscriptstyle}{\hspace{0.25em}}}
    \right\rangle
}

\newcommand{\id}{\mathbb{I}}
\newcommand{\bbI}{\mathbb{I}}

\newcommand{\bbU}{\mathbb{U}}
\newcommand{\bbZ}{\mathbb{Z}}

\newcommand{\bbP}{\mathbb{P}}
\newcommand{\bbE}{\mathbb{E}}

\newcommand{\bbB}{\mathbb{B}}

\newcommand{\bbL}{\mathbb{L}}

\newcommand{\calP}{\mathcal{P}}
\newcommand{\calB}{\mathcal{B}}
\newcommand{\calE}{\mathcal{E}}
\newcommand{\calEeff}{\calE'}
\newcommand{\calI}{\mathcal{I}}
\newcommand{\calX}{\mathcal{X}}
\newcommand{\calY}{\mathcal{Y}}
\newcommand{\calZ}{\mathcal{Z}}
\newcommand{\calU}{\mathcal{U}}
\newcommand{\calC}{\mathcal{C}}

\newcommand{\calH}{\mathcal{H}}
\newcommand{\calM}{\mathcal{M}}
\newcommand{\calS}{\mathcal{S}}
\newcommand{\calD}{\mathcal{D}}
\newcommand{\calN}{\mathcal{N}}

\newcommand{\calA}{\mathcal{A}}

\newcommand{\calG}{\mathcal{G}}

\newcommand{\BcalE}{\bar{\calE}}

\newcommand{\MZ}{\calM_{\bbZ}}

\newcommand{\bsa}{\boldsymbol{a}}

\newcommand{\comments}[1]{}

\begin{document}

\title{Classical Noise Inversion for Error-Propagatable Circuits with Minimal Overhead under General Gate-Dependent Noise} 
\author{Dayue Qin}
\email{dyqin@fudan.edu.cn}
\affiliation{Key Laboratory for Information Science of Electromagnetic Waves (Ministry of Education), Fudan University, Shanghai 200433, China}

\author{Ying Li}
\email{yli@gscaep.ac.cn}
\affiliation{Graduate School of China Academy of Engineering Physics, Beijing 100193, China}

\author{You Zhou}
\email{you\_zhou@fudan.edu.cn}
\affiliation{Key Laboratory for Information Science of Electromagnetic Waves (Ministry of Education), Fudan University, Shanghai 200433, China}

\begin{abstract}

Quantum error mitigation (QEM) is critical for extracting reliable computations from noisy quantum processors, proving itself essential not only in the near term but also as a valuable supplement to fully fault-tolerant systems in the future. 
Despite the necessity, practical QEM deployment still confronts challenges, including the excessive cost of sampling quantum circuits and reliance on unrealistic assumptions such as gate-independent noise.
{In this paper, we propose Classical Noise Inversion (CNI), which shifts the QEM from sampling diverse quantum circuits to repeated single-circuit measurement, a drastic cost reduction given that the former incurs far heavier time overhead than the latter on realistic quantum hardware. We target general noise models and establish critical conditions for their classical tractability, under which CNI remains effective. For noise models that violate the condition, we propose partial CNI, which mitigates the classically tractable factor of noise via CNI and mitigates the other part via probabilistic error cancellation tailored to gate-dependent noise. }
To minimize the sampling overhead, we introduce noise compression, which groups noise components with equivalent effects on measurement outcomes, thereby attaining theoretically optimal error-mitigation overhead. 
{The proposed protocols are particularly efficient for error-propagatable circuits, including adaptive universal quantum computing represented by Clifford+T fault-tolerant circuits, and generalized measurement represented by classical shadows. }
To demonstrate the practical merits of CNI, we integrate it with thrifty classical shadow, and use analysis and numerical simulations to show its advantages in both efficiency and accuracy over existing approaches.
By shifting the dominant overhead of QEM from costly quantum sampling to lightweight classical post-processing while achieving theoretically minimal overhead, this work substantially elevates the practical usability of noisy quantum devices. 

\end{abstract}

\maketitle

\section{Introduction}
\label{sec: intro}

Addressing noise is essential for the utility of quantum computers. Quantum error mitigation (QEM)~\cite{Temme2017,Li2017,Endo2018,Cai2023a}, which tackles noise with minimal qubit overhead, is necessary for the reliability of near-term quantum devices~\cite{Kim2023a}. In the meanwhile, QEM complements fault-tolerant architectures~\cite{Preskill2025,Aharonov2025} by suppressing the non-correctable logical errors~\cite{Suzuki2022,Wahl2023,Zhang2025a}, enhancing imperfect magic distillations~\cite{Lostaglio2021a, Piveteau2021b}, circumventing decoding bottlenecks~\cite{Smith2024} and mitigating circuit compilation imperfections~\cite{Suzuki2022}. 

\begin{figure*}
  \begin{center}
    \includegraphics[width=\linewidth]{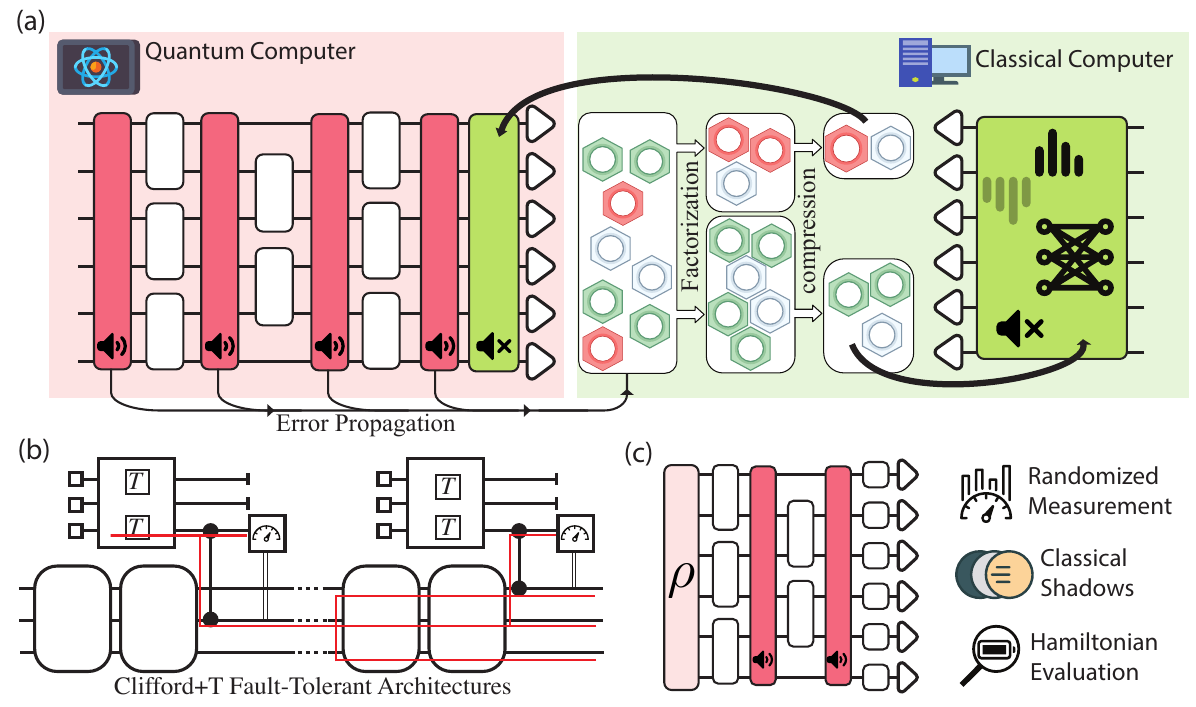}
    \caption{Illustration of the partial classical noise inversion framework and its applications. (a) The complete workflow for mitigating general noise models. The noise is propagated through the circuit and the dominant noise is canceled by classically simulating the noise inversion. Noise compression groups equivalent error components to minimize the sampling overhead. Noise factorization isolates the classically intractable components, which are mitigated by a refined probabilistic error cancellation tailored to gate-dependent noise. (b) and (c) Representative error-propagatable circuits amenable to the CNI approach. (b) As indicated by the red wire, in a Clifford+T fault-tolerant quantum circuit, any noise only passes through a limited number of non-Clifford gates, therefore noise propagation can be efficiently simulated regardless of the circuit size. (c) Quantum circuits having a classically-simulable part before measurement and the corresponding applications.
    }
    \label{fig: intro}
  \end{center}
\end{figure*} 

Despite the necessity and utility of QEM, its practical application faces {two major challenges}. {The first challenge is the excessive overhead of executing varied quantum circuits.} This overhead originates from two distinct sources. One source is a fundamental, worst-case exponential scaling of the required number of circuit samples with circuit size~\cite{Takagi2023,Tsubouchi2023,Takagi2022a,Quek2024}. 
{The other source is a practical wall-clock bottleneck: sampling a large ensemble of distinct circuits dominates the total runtime, far exceeding the quantum execution time itself. Concretely, recompiling each digital circuit into control pulses and transferring them to the device incurs a cost of hundreds of milliseconds~\cite{Karalekas2020, Fruitwala}, which is one to two orders of magnitude higher than the execution time of quantum circuits consisting of thousands of gate layers. This gap is consistent across state-of-the-art platforms: a full surface-code cycle takes only $\sim1~\mu$s on superconducting processors~\cite{GoogleQuantumAIAndCollaborators2025}, two-qubit gates are realized in $\sim270$~ns on neutral-atom processors~\cite{Bluvstein2024} and $\sim 70~\mu$s on trapped-ion systems~\cite{Ransford2025}.} 

The second challenge is that the efficacy of existing QEM techniques relies on restrictive assumptions regarding the noise properties. These methods often require the noise model to remain invariant under specific circuit modifications, a requirement that is difficult to meet under realistic, gate-dependent noise. 
For instance, probabilistic error cancellation~\cite{Temme2017,Endo2018} is derived under the assumption that the auxiliary gates inserted for mitigation do not alter the underlying noise model. Similarly, error extrapolation~\cite{Temme2017,Li2017} requires that the noise model in intentionally error-boosted circuits differs from the original only in severity, not in the fundamental structure of the error model itself. Virtual distillation, meanwhile, demands that the additional entangling gates it introduces are effectively error-free~\cite{Koczor2021,Huggins2021a,Huo2022}.

{To navigate these practical obstacles, we propose a framework that minimizes both the theoretical and practical quantum error mitigation cost. As illustrated in Fig.~\ref{fig: intro}(a), we first introduce \textit{classical noise inversion} (CNI), which mitigates propagated noise entirely in classical post-processing, thereby eliminating the massive quantum circuit-sampling cost and ensuring resilience against gate-dependent noise. Second, to address noise components that defy pure classical inversion, we establish conditions for classical tractability and develop a \textit{partial CNI} scheme based on noise factorization. Third, we introduce a \textit{noise compression} technique that groups equivalent error components to attain theoretically optimal sampling costs. This powerful framework has a wide range of applications. In the remainder of this section, we elaborate on these proposed protocols, then introduce the concept of \textit{error-propagatable circuits} where our methods can be efficiently applied, and finally demonstrate their integration with classical shadow estimation.}

CNI is a QEM framework that fundamentally shifts the mitigation overhead from quantum to classical resources. It works by classically simulating the propagation of a noise model through the gates and projective measurement, then inverting its effect entirely in post-processing. 
This approach offers two immediate advantages. First, it completely eliminates the cost of sampling distinct quantum circuits, lowering the overhead to merely sampling measurement shots from a fixed circuit, where the latter procedure is substantially faster than the former. Second, {as it operates without modifying the circuits, it remains effective under gate-dependent noise as long as the noise model is well-characterized.}

{We prove that a classical linear operation that inverts the noise effect does not exist in general---a problem that previous works have left unaddressed. Measurement error mitigation~\cite{Chen2019, Maciejewski2020, Bravyi2021b} suppresses readout errors using purely classical post-processing, by assuming that the noise is stochastic bit flips, which is not true for all quantum noise channels. While the IC-POVM incorporated with tensor-network post-processing~\cite{Filippov,Fischer2026, Farias2024} circumvents this problem, its effectiveness is limited to shallow quantum circuits. Additionally, as POVMs are implemented by inserting random gates ahead of projective measurements, mitigating general noise intrinsic to such projective measurements remains an open problem. 

For noise that cannot be completely eliminated classically, we propose partial CNI, a method that factorizes the whole propagated noise into two components: the classically tractable and intractable components. We refer to this procedure as quantum-classical factorization. The classically tractable component is suppressed via CNI in the classical post-processing stage.  Meanwhile, the classically intractable component is suppressed using a refined probabilistic error cancellation protocol. This refined scheme incorporates gate-dependent noise characteristics into the noise quasi-probability distribution, preserving robust performance even in the presence of arbitrary gate-dependent noise. }

While (partial) CNI has alleviated the costly circuit sampling, we further propose \textit{noise compression} to reach the theoretical optimality of practical quantum error mitigation. Noise compression essentially produces a QEM formula by grouping distinct error components that exhibit identical effects on the measurement results, and it is proven to achieve the optimal QEM cost. 

{We propose practical methods to efficiently implement both CNI and noise compression. This practical method leverages two key observations. First, the quasi-probability decomposition of any noise model can be expressed as a linear combination of universal basis operations~\cite{Endo2018}, which are either Clifford operations or projectors and thus permits efficient classical simulation. Second, for real quantum hardware, local gate noise or spatiotemporal noise models with a limited number of decomposition terms are sufficiently accurate to capture noise effects. Leveraging these two points, CNI operates by inverting noise models that admit an efficient classical representation. In terms of noise compression, it is efficiently realized through a subroutine we refer to as $\bbZ$-rows comparison for distinct noise components, which leverages the efficient classical simulation of Clifford gates~\cite{Aaronson2004}.}

{The proposed framework can be efficiently applied to a specific yet highly practical type of circuits termed \textit{error-propagatable circuits}. Full classical simulation of such circuits themselves becomes computationally intractable when the number of qubits is large; in contrast, simulating noise propagation within error-propagatable circuits can be accomplished efficiently. Crucially, error-propagatable circuits include not only classically simulable circuits, but also classically non-simulable circuits with important applications. 
We identify two major categories of such circuits. The first encompasses Clifford+T fault-tolerant architectures~\cite{Bravyi2005,Fowler2012} where non-Clifford operations are implemented via T-state injection. As shown in Fig.~\ref{fig: intro} (b), in such circuits, noise propagating forwards passes through only a finite number of $T$ gates, independent of the circuit size. The second category covers the part of a circuit that can be efficiently simulated classically before measurement. As shown in Fig.~\ref{fig: intro} (c), these protocols start with an unknown or classically non-simulable quantum state, which is then processed by a classically simulable circuit such as Clifford~\cite{Aaronson2004} or matchgates~\cite{Terhal2002,Jozsa2008}, as well as shallow-depth circuits~\cite{Napp2022,Bravyi2024}, then measured in the computational basis. Examples include the circuits used in randomized measurement techniques~\cite{Elben2022} (e.g., classical shadows~\cite{Huang2020a}), as well as the simultaneous estimation of commuting observables adopted in many-body Hamiltonian learning~\cite{Gokhale2020a,Miller2024,Reggio2024}. 
}

We demonstrate the efficacy of CNI by integrating it with classical shadow estimation~\cite{Aaronson2020, Huang2020a} to establish a broadly applicable framework for robust quantum state learning.
As shown in Fig.~\ref{fig: intro}(c), by applying CNI to every random measurement, CNI is integrated with thrifty shadow estimation~\cite{Helsen2023,Zhou2023}, a method that improves resource efficiency through circuit reuse. The CNI-based robust shadow framework is universally applicable to shadow estimation with all kinds of unitary ensembles~\cite{Nguyen2022,Hu2022b,Zhang2024a,Wan2023}, circuit structures~\cite{Hu2023,qingyue2025,You2025} and shallow-circuit variants~\cite{Bertoni2024a,Schuster2025,Hu2025}. Meanwhile, the framework inherits the advantages of CNI: it achieves unbiased estimation even under the challenging gate-dependent noise while maintaining low variance. By comparison, as demonstrated by our numerical simulations, previous methods are either biased \cite{Chen2021} or introduce greater variance \cite{Jnane2024}.

The remainder of this paper is organized as follows. In Sec.~\ref{sec: CNI}, we introduce the theoretical framework of CNI, demonstrating its effectiveness without circuit-sampling cost and its resilience against gate-dependent noise. Subsequently in Sec.~\ref{sec: partial CNI}, we derive the necessary and sufficient conditions for the validity of CNI, and put forward the partial CNI scheme refined to gate-dependent noise to address scenarios where the conditions fail to hold. We then propose noise compression in Sec.~\ref{sec: optimality} and prove that it attains the minimal-cost mitigation formula. Next in Sec.~\ref{sec: CNI-based shadow}, we develop the theory of CNI-based robust shadow estimation, demonstrating how its variance relates to the ideal shadow norm and the noise model. Finally in Sec.~\ref{sec: numerical }, we provide numerical simulations illustrating the effectiveness of our methods and show various advantages over existing techniques. In the Methods section, 
we provide an efficient subroutine for noise compression and provide details on the numerical simulation. 

\section{Classical noise inversion}
\label{sec: CNI}
We first introduce the Classical Noise Inversion (CNI) method and analyze its cost. CNI mitigates the effects of noise using only classical post-processing for the following general measurement setup: An unknown quantum state (density operator) $\rho$ evolves under a unitary operation $\calU:\rho\rightarrow U\rho U^\dagger$, followed by a measurement in the computational basis. We remark that when the unitary operator $U$ is sampled randomly from an ensemble, the CNI method can also be applied to randomized measurements and shadow estimation. 

In this work, we primarily consider an $n$-qubit system with Hilbert space $\mathcal{H}_d = \mathcal{H}_2^{\otimes n}$, and denote the computational basis by ${\ket{b}}$, where $b$ is an $n$-bit string. However, our method is general and can also be extended to qudit or continuous-variable systems. The unitary operations $U$ we consider are chosen to be Clifford circuits. We remark, however, that the method may also apply to matchgate circuits or other shallow circuits, as long as they can be efficiently simulated classically. 

Generally, the expectation of a linear function $F$ of the computational basis states, e.g., $F=\sum_b f(b)\ketbra{b}{b}$, is estimated using the measurement outcomes. The noiseless expectation value of $f$ is  
\begin{equation}\label{eq:f}
   f= \dbra{F}  \calM_{\bbZ} \calU \dket{\rho}. 
\end{equation}
Here, we clarify the definition and notation used in Eq.~\eqref{eq:f} and throughout the paper. $\dket{\rho}$ is the vectorization of $\rho$ in the Pauli basis: $\dket{\rho} = \sum_{\sigma} \Tr (\sigma^\dagger \rho ) \dket{\sigma}$, where  $\sigma = P / \sqrt{d}$ with $P\in \bbP =  \{I, Z, X, Y\}^{\otimes n }$. $\{\dket{\sigma}\}$ are orthonormal bases in terms of the Hilbert-Schmidt inner product: $\dbraket{\sigma}{\sigma^\prime} = \Tr (\sigma^\dagger \sigma^\prime) = \delta_{\sigma, \sigma^\prime}$. Correspondingly, $\calU$ denotes the unitary channel of $U$, represented by its Pauli transfer matrix (PTM), which is a $4^n \times 4^n$ matrix in the Pauli basis. The PTM of a general quantum channel $\calE$ is written as $\calE = \sum_{\sigma, \sigma^\prime}\lambda_{\sigma, \sigma^\prime} \dketbra{\sigma}{\sigma^\prime}$, where the coefficients are $\lambda_{\sigma, \sigma^\prime} = \Tr \left[\sigma^\dagger \calE (\sigma^\prime)\right]$. PTM allows us to express the channel concatenation as matrix multiplication: $\calE_2 \circ \calE_1 (\rho) = \calE_2 \calE_1 \dket{\rho}$. $\calM_{\bbZ} = \sum_{b\in\{0,1\}^n} \dketbra{b}{b}$ is the channel of computational basis measurement, where $\bbZ = \{I, Z\}^{\otimes n}$ denotes the set of $n$-qubit sign-free Pauli-Z operators. One can find that $\calM_{\bbZ}$ is a completely dephasing channel $ \calM_{\bbZ} = \sum_{\sigma\in \bbZ} \dketbra{\sigma}{\sigma}$. $\dbra{F}$ is the conjugate of $\dket{F}$. 

In real experiments, quantum operations and measurements are noisy, and the intended quantum process $\calU$ is replaced by $\tilde{\calU} := \calE\calU$, where $\calE = \tilde{\calU}\calU^\dagger$ captures the net effect of both circuit noise and computational-basis measurement noise. 
To mitigate the effects of noise, conventional probabilistic error cancellation techniques~\cite{Temme2017,Endo2018} involve inserting random channels $\calB$ into the quantum circuit. This approach is rooted in the quasi-probability decomposition of noise inversion, which is mathematically expressed as:
\begin{equation}\label{eq:oriInv}
\calE^{-1} = \sum_{\calB\in \bbB} Q_{\calB} \calB,
\end{equation}
where $Q_{\calB}$ are real-valued coefficients, and $\bbB$ is a set of complete basis channels that are classically simulable. Examples include the universal basis channels introduced in \cite{Endo2018}, which comprise Clifford gates and projectors onto the eigenbases of the Pauli operators $X$, $Y$ and $Z$. 
The factor $ \gamma = \sum_{\calB}|Q_{\calB}| $, known as total negativity, quantifies the difficulty of classically simulating quantum channels~\cite{Veitch2012, Pashayan2015, Howard2017, Bennink2017}. In the context of quantum error mitigation, $\gamma$ corresponds to the quantum circuit sampling overhead needed to attain a specific estimation precision.
As elaborated in the Introduction, this conventional approach encounters two key challenges: the huge time consumption of sampling quantum circuits and instability under gate-dependent noise. 

CNI effectively addresses both aforementioned challenges by eliminating the need for quantum-circuit sampling. We assume that the noise can propagate through the final measurement channel, i.e., that there exists $\calE^\prime$ such that $\calE^{\prime} \calM_{\bbZ} = \calM_{\bbZ}\calE$. {We will present the necessary and sufficient conditions for this assumption in the next section, and propose effective approaches to address cases where the assumption fails to hold.} CNI mitigates the noise effect by classically simulating $\calE^{\prime -1}$ in classical post-processing and applying it to the measurement outcomes of the noisy circuit. This yields an unbiased estimator of the ideal expectation value since $\calE^{\prime -1}\calM_{\bbZ} \tilde{\calU} = \calM_{\bbZ} \calU$. 

The technique of simulating $\calE^{\prime -1}$ is versatile, as it can be implemented using either Monte Carlo methods or tensor-network-based approaches etc. In this paper, we focus on the Monte Carlo method, which is grounded in the quasi-probability decomposition: $\calE^{\prime -1} = \sum_{\calB\in \bbB} Q^{\prime}_{\calB} \calB$. Notably, the coefficients $Q^{\prime}_{\calB}$ can be efficiently determined if the errors in the circuit are well-characterized (e.g. via gate-set tomography~\cite{Greenbaum2015, Nielsen2021}). 
In the Monte Carlo simulation, upon recording a measurement outcome $b$, we sample instances of $\calB$ from the distribution $\Prob{\calB} = |Q_{\calB}^{_\prime}|/\gamma^\prime$ with $\gamma^\prime = \sum_{\calB\in \bbB} |Q^{\prime}_{\calB}|$. In practice, we can measure the circuit for multiple shots and sample multiple basis channels for each shot. We present the unbiased estimator given by CNI and its performance in the following theorem, with the full proof left in Appendix~\ref{app: proof of variance of CNI}. 

\begin{theorem}[Performance guarantee of CNI] 
\label{th: perf CNI}
With one-shot noisy measurement outcome $b$ and one instance of a random basis channel $\calB$, the unbiased estimator of the expectation value $f$ in Eq.~\eqref{eq:f} is given by
\begin{equation}\label{eq:hatf}
  \hat{f}_{b, \calB} = \gamma' \frac{Q'_{\calB}}{|Q'_{\calB}|} \dbra{F}\calB\dket{b}. 
\end{equation}

Suppose that the noisy circuit is measured for $K$ shots. For each shot, $L$ instances of $\calB$ are generated, yielding $K\times L$ estimators of $f$ constructed through \eqref{eq:hatf} in total. Let $\hat{f}_{\mathrm{tot}}$ denote the mean of these estimators, and the variance satisfies
\begin{equation}
  \Var{\hat{f}_{\mathrm{tot}}} \leq \frac{\gamma^{\prime 2}}{K} \left(\frac{1}{L} \norm{F}_{\star}^2 + \left(1 - \frac{1}{L}\right)\norm{F}_{\circ}^2\right),
  \label{eq: variance of CNI}
\end{equation}
where $\norm{\cdot}_{\star}$ and $\norm{\cdot}_{\circ}$ are semi-norms defined in Eq.~\eqref{eq: norm F star} and Eq.~\eqref{eq: norm F o} in Appendix~\ref{app: proof of variance of CNI}, and satisfy $\norm{\cdot}_{\circ}\leq \norm{\cdot}_{\star}\leq \norm{\cdot}_{\infty}$ with $\norm{\cdot}_\infty$ denoting the Schatten $\infty$-norm (largest absolute singular value). 
\end{theorem}

As established in Theorem~\ref{th: perf CNI}, 
CNI achieves an unbiased estimate of the ideal expectation value with precision $\epsilon$ using $O(\gamma^{\prime 2}/\epsilon^2)$ measurement shots. In comparison, conventional quantum error mitigation requires $O(\gamma^{2}/\epsilon^2)$ quantum circuit samples. Since the time consumed by sampling random quantum circuits is far greater than that of sampling measurement results of a fixed circuit~\cite{Karalekas2020}, CNI significantly improves practical efficiency. Furthermore, conventional methods rely on the assumption that noise models remain identical across sampled variant circuits, an assumption that is never perfectly satisfied by realistic quantum systems. By contrast, our approach maintains effectiveness even under gate-dependent noise models, as it operates entirely in classical postprocessing without modifying the quantum circuit itself. 
Finally, due to $\norm{F}_{\circ} \leq \norm{F}_{\star}$, CNI allows  further reduction of the variance by increasing $L$, which is purely classical postprocessing, without the need of measuring the quantum circuit. 

In addition to the theoretical advantages of CNI itself, our further innovation ensures that the overhead $\gamma^\prime$ of CNI is smaller than the overhead $\gamma$ of conventional methods.  We recognize that only noise directly affecting measurement outcomes requires inversion, not the entire noise channel. Building on this insight, we develop an efficient procedure called noise compression to achieve the optimal value of $\gamma^\prime$ while ensuring it remains lower than that of conventional methods. 

\section{{Partial CNI for classically intractable noise}}  
\label{sec: partial CNI}
{In this section, we establish the condition under which the CNI method can cancel the noise completely, and provide tangible solutions in the case where the condition is violated.}

As previously mentioned, CNI is valid if and only if the overall noise can propagate through the computational-basis measurement channel. Formally, this requires the existence of a noise channel $\calE^\prime$ such that the following commutation relation holds: $\calE^{\prime} \calM_{\bbZ} = \calM_{\bbZ}\calE$. Note that $\calE^\prime$ may not be unique.   
To characterize precisely when this holds, we establish in the following proposition the necessary and sufficient condition for CNI validity, with the proof deferred to Appendix~\ref{app: channel propagation}. Intuitively, CNI is valid if and only if the noise channel does not map any non-$\bbZ$ Pauli operator into a Pauli operator in $\bbZ$.
\begin{proposition}\label{prop:valid}
    CNI is valid if and only if 
  \begin{equation}
    \dbra{\sigma}\calE \dket{\nu} = 0, \; \forall\; \sigma \in \bbZ, \nu \notin \bbZ, \label{eq: condition}
  \end{equation} 
  where $\bbZ = \{I, Z\}^{\otimes n}$ denotes the set of $n$-qubit sign-free Pauli-Z operators. 
\end{proposition} 

{In the $\bbZ\oplus\bbZ^\perp$ partition of the PTM space with projectors $\Pi_\bbZ$ onto $\bbZ$ and $\Pi_\perp=\id-\Pi_\bbZ$, Eq.~\eqref{eq: condition} can be expressed as $\Pi_\bbZ \calE \Pi_{\perp} = 0$. As shown in Fig.~\ref{fig: validity and optimality}(a), Eq.~\eqref{eq: condition} requires the upper-right block $\Pi_\bbZ \calE \Pi_{\perp}$ of the noise PTM to vanish. We term this block the \textit{classically intractable} block of the noise, as any non-zero element in this block cannot be undone by classical post-processing. The dominant sources of noise, including depolarizing, dephasing, and amplitude-damping, satisfy Proposition~\ref{prop:valid}, while components that violate this condition exist in general~\cite{Fruitwala}.
} 

{CNI alone mitigates the noise only partially when this classically intractable part persists, and the residual must be removed inside of the circuit. To that purpose, we propose partial CNI, which mitigates the classically tractable part via CNI while mitigating the classically intractable part by inserting error-mitigating gates into the circuit. The latter resembles conventional probabilistic error cancellation~\cite{Temme2017, Endo2018}, except that the quasi-probabilities are refined to account for the noise on the error-mitigating gates, which enables unbiased error mitigation even in the presence of gate-dependent noise. 
This quantum-classical hybrid approach adopts partial circuit sampling in place of full measurement shots. Nevertheless, when the classically intractable component is limited in scale, its circuit sampling overhead is substantially lower than that of traditional methods. 

To make this precise, we write the PTM of the noise in blocks with respect to the $\bbZ\oplus\bbZ^\perp$ partition, $\calE_{00}=\Pi_\bbZ\calE\Pi_\bbZ$, $\calE_{01}=\Pi_\bbZ\calE\Pi_\perp$, $\calE_{10}=\Pi_\perp\calE\Pi_\bbZ$, and $\calE_{11}=\Pi_\perp\calE\Pi_\perp$. The noise channel then factorizes into
\begin{equation}
  \calE = \calE_c \calE_q,
  \label{eq: qc factorization}
\end{equation}
where $\calE_c$ is the classically tractable factor and $\calE_q$ is the classically intractable factor. The explicit expression of the factors is stated in the following proposition, which is proved in Appendix~\ref{app: tractable intractable decomposition}.

\begin{proposition}[Quantum-classical factorization]
\label{prop: Quantum-classical factorization}
Let $\calE$ be a trace-preserving map whose $\calE_{00}$ is invertible on the $\bbZ^\perp$-subspace. Then the factors in Eq.~\eqref{eq: qc factorization} are 
\begin{eqnarray}
  \calE_c&=&\begin{pmatrix}\calE_{00}&0\\[2pt]\calE_{10}&\calE_{11}-\calE_{10}\calE_{00}^{-1}\calE_{01}\end{pmatrix},\\
  \calE_q&=&\begin{pmatrix}\id&\calE_{00}^{-1}\calE_{01}\\[2pt]0&\id\end{pmatrix},
\end{eqnarray}
in which both factors are trace-preserving, the factor $\calE_c$ is classically tractable, satisfying $\Pi_{\bbZ}\calE_c\Pi_{\perp} = 0$, and $\calE_q=\calI+\calN$ deviates from the identity only through the classically intractable shift $\calN=\Pi_\bbZ\,(\calE_{00}^{-1}\calE_{01})\,\Pi_\perp$, which maps $\bbZ^\perp$ into $\bbZ$.
\end{proposition}

We eliminate $\calE_c$ via classical post-processing with CNI  according to $\calE_c^{-1} = \sum_{\calB\in \bbB}Q_{\calB, c}\calB$: a classically simulable  channel $\calB^\prime$ satisfying $\calB^\prime \calM_{\bbZ} = \calM_{\bbZ}\calB$ is sampled and applied to the measurement outcomes. Meanwhile, we suppress $\calE_q$ by sampling random circuits based on the decomposition $\calE_q^{-1} = \sum_{\calB\in\bbB }Q_{\calB, q} \calE_\calB\calB$, where $\calE_\calB$ is the gate-dependent noise associated with $\calB$. The quasi-probabilities $\{Q_{\calB,q}\}$ accounting for both $\calE_q$ and $\calE_\calB$ for all $\calB\in \bbB$ are obtained by solving the linear system}
\begin{equation} 
  \sum_{\calB\in\bbB} Q_{\calB, q}\, \calE_{\calB}\, \calB\, \calE_q = \calI.
\end{equation} 
This solution exists as long as the noisy operations $\{\calE_\calB\calB\}_{\calB\in\bbB}$ span the relevant space, and is unique when they are linearly independent; both hold at small error rates. 

\begin{figure}
  \begin{center}
    \includegraphics[width=1\linewidth]{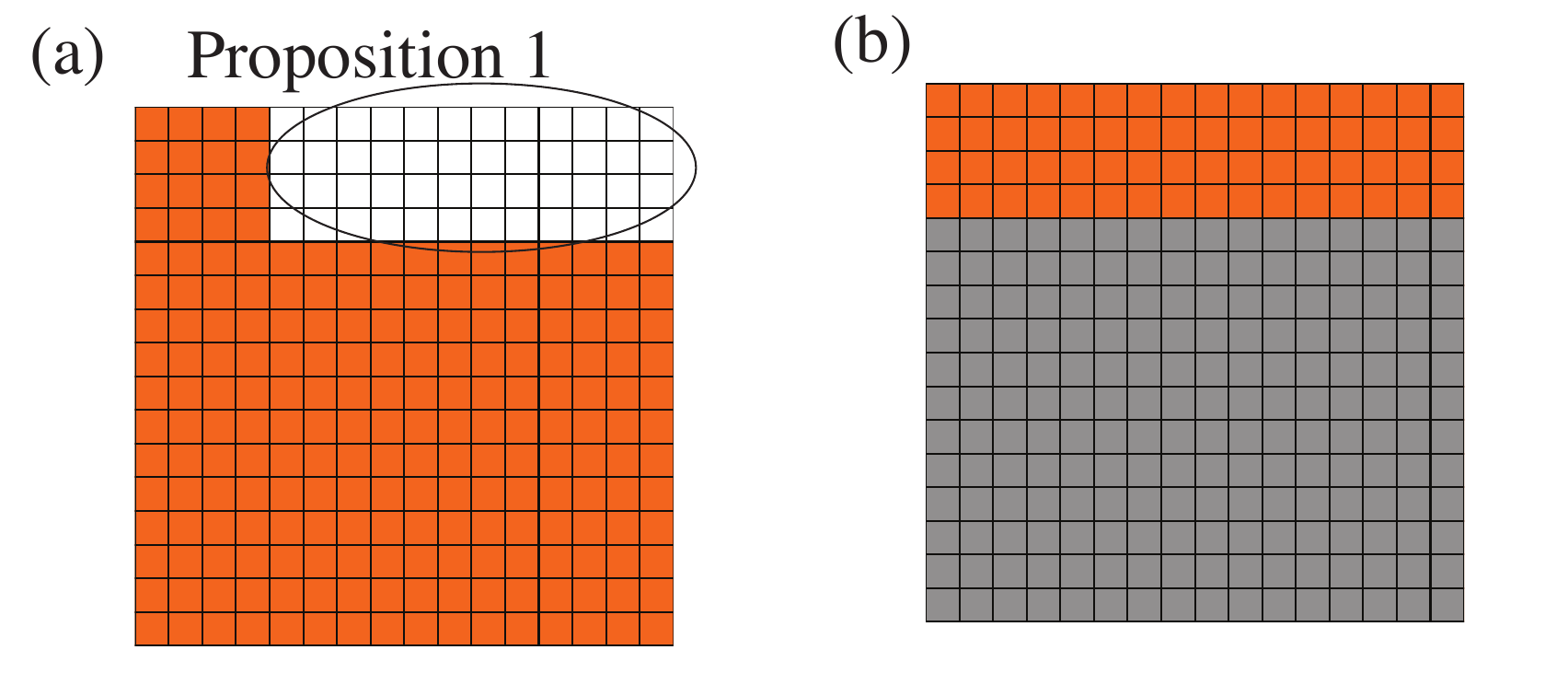} 
    \caption{{A two-qubit example illustrating the validity and optimality of CNI. (a) The upper-right block of the Pauli transfer matrix (PTM) being zero (empty plaquettes within the ellipse) is the necessary and sufficient condition for CNI validity (Proposition~\ref{prop:valid}). (b) Noise compression: we only need to consider the $\bbZ$-rows (plaquettes filled with orange), and neglect the non-$\bbZ$-rows elements (plaquettes filled with gray). }}
    \label{fig: validity and optimality}
  \end{center}
\end{figure}

{
\begin{theorem}[Coarse-grained variance bound of hybrid estimator]
\label{th: coarse grained variance}
Consider the partial quantum sampling scheme that cancels the classically intractable factor $\calE_q$ by sampling random noisy circuits and the classically tractable factor $\calE_c$ by CNI, drawing $M$ random circuits, $K$ measurement shots per circuit, and $L$ classical basis-channel instances per shot. The resulting estimator $\hat f_{\mathrm{tot}}$ is unbiased, $\bbE[\hat f_{\mathrm{tot}}]=f$, and its variance satisfies 
\begin{equation}
  \Var{\hat f_{\mathrm{tot}}}\ \le\ \frac{\gamma_q^2}{M}\Big(1+\frac{\gamma_c^2}{K}\Big)\norm F_\infty^2, \label{eq: hybrid variance main}
\end{equation}
where $\gamma_q=\sum_{\calB}\absLR{Q_{\calB,q}}$ and $\gamma_c=\sum_{\calB}\absLR{Q_{\calB,c}}$ are the quantum and classical sampling overheads, and $\norm F_\infty$ is the Schatten $\infty$-norm of $F$.
\end{theorem}
}

{The bound's irreducible floor $\gamma_q^2\norm F_\infty^2/M$ carries only the quantum sampling overhead $\gamma_q$ of the intractable factor and is suppressed solely by drawing more random circuits $M$; the potentially large classical overhead $\gamma_c$ enters only the cheaper $\gamma_q^2\gamma_c^2\norm F_\infty^2/(MK)$ term. The bound is coarse-grained, as $L$ does not appear explicitly here. Nevertheless, the variance can be reduced further by increasing $L$, exactly as in pure CNI: in the $1/(MK)$ term $\norm F_\infty^2$ may be replaced by a tighter term $\tfrac1L\norm F_{\star,\prime}^2+(1-\tfrac1L)\norm F_{\circ,\prime}^2$, where $F_{\star,\prime}$ and $F_{\circ,\prime}$ are semi-norms satisfying $ \norm F_{\circ,\prime} \leq \norm F_{\star,\prime} \leq \norm F_\infty$. 
Finally, when the classically intractable part vanishes ($\calE_q=\calI$, $\gamma_q=1$), the quantum sampling cost disappears and the bound reduces to the pure-CNI result of Theorem~\ref{th: perf CNI}: the cost reverting entirely to the measurement shots $K$.}

{Distributing the inversion over the two factors is generally sub-optimal, but the penalty is tightly bounded. Theoretically, if $\gamma$ ($\gamma_c$ and $\gamma_q$) saturates the information-theoretic optimal cost of simulating $\calE$ ($\calE_c$ and $\calE_q$), the following holds
\begin{equation}
\begin{aligned}
  \gamma_q&\ \le\ 1+\norm{\calN}_\diamond,\\
   \gamma_q\gamma_c&\ \le\ \big(1+\norm{\calN}_\diamond\big)^2\,\gamma,
\end{aligned}\label{eq: gamma_q}
\end{equation}
as proved by Proposition~\ref{prop: near optimal factorization} in Appendix~\ref{app: penalty}. The quantum sampling overhead $\gamma_q$ and the penalty $c\le(1+\norm{\calN}_\diamond)^2$ are governed solely by the diamond norm of the classically intractable shift $\calN$. When the intractable part is small ($\norm{\calN}_\diamond\ll1$), $c\gtrsim1$, $\gamma_q\gtrsim1$ and $\gamma_c\approx\gamma$, therefore isolating the intractable contribution comes essentially for free.
In practical implementations, the cost is actually larger than the information-theoretic optimum, since the noise is decomposed by the universal basis channels instead of two CPTP maps that minimize the cost.
When the universal bases are aligned with the noise~\cite{takagiOptimalResourceCost2021}, i.e., the bases are expressive enough to realize the optimal decomposition, the tight bound in Eq.~\ref{eq: gamma_q} is achieved; when they are misaligned, $\gamma_q$ can exceed $1+\norm{\calN}_\diamond$, and diverges when $\calN$ lies outside the span of the noisy basis set. }

{While this section focuses on the definite quantum-classical factorization in Eq.~\eqref{eq: qc factorization}, the cost can be prohibitive since the size of the PTM of $\calE$ grows exponentially with the number of qubits. To solve this problem, we propose a semi-definite program to efficiently obtain a decomposition while minimizing the quantum sampling overhead $\gamma_q$. The program exploits that the noise is typically either a propagated local noise channel or a global noise model with a polynomial-sized quasi-probability decomposition. For this regime, only a small basis set $\bbB$ is involved, and the program outputs the decomposition coefficients $\{Q_{\calB, c}\}$ and $\{Q_{\calB, q}\}$ by minimizing $\gamma_q$ under the constraint $\calE_c^{-1}\calE_q^{-1}\calE = \calI$, achieving a quantum sampling cost no larger than (generically below) that of the explicit block-$UL$ factorization of Eq.~\eqref{eq: qc factorization}. }

\section{Noise compression}
\label{sec: optimality}
{Having established the condition of validity for CNI and proposed a method with partial quantum sampling to remove the residual at a controlled circuit sampling cost when the condition is violated}, we now turn to the optimality, which concerns choosing the inverse channels $\calE_c^{\prime -1}$ and $\calE_q^{\prime - 1}$ that perfectly cancel the noise while minimizing the sampling cost. The optimal performance is achieved through \textit{noise compression}, which groups distinct basis channels with identical measurement effects.

{To illustrate the idea of noise compression, we start with the inversion of a general noise model $\calE^{-1}$, which incurs an overhead of $\gamma = \sum_{\calB} |Q_{\calB}|$ as given by Eq.~\eqref{eq:oriInv}. As shown in Fig.~\ref{fig: validity and optimality} (b), only the $\bbZ$-rows of the PTM affect the measurement outcomes: $\calM_Z\calE = \Pi_{\bbZ}\calE\Pi_{\bbZ} + \Pi_{\bbZ}\calE\Pi_{\perp}$. }
To reduce the sampling overhead, we divide the basis channel set $\bbB$ into subsets: $\{\bbB_1,\bbB_2, \cdots,\bbB_T\}$, where basis channels within the same subset share the same $\bbZ$-rows. Formally, any $\calB$ and $\calB^\prime$ within the same subset satisfy
\begin{equation}
  \Pi_{\bbZ}\calB\Pi_{\bbZ} + \Pi_{\bbZ}\calB\Pi_{\perp} = \Pi_{\bbZ}\calB^\prime\Pi_{\bbZ} + \Pi_{\bbZ}\calB^\prime\Pi_{\perp}. 
\end{equation}
Denote the set of representative basis channels as $\bbB^\prime = \{\calB_1, \calB_2, \cdots,\calB_T\}$, where $\calB_t$ is an arbitrarily chosen channel from $\bbB_t$, the optimal noise inversion shows 
\begin{equation}\label{eq:comInv}
  \calE^{\prime -1} = \sum_{t=1}^{T} Q^\prime_{\calB_t} \calB_t, 
\end{equation}
where the coefficients are updated by summing over the subset: $Q^\prime_{\mathcal{B}_t} = \sum_{\mathcal{B}\in \mathbb{B}_t} Q_{\mathcal{B}}$. We refer to the inverse in Eq.~\eqref{eq:comInv} as the compressed inverse, which perfectly cancels the error because $\calM_Z\calE^{\prime -1} = \calM_Z\calE^{-1}$ since $\calE^{\prime -1}$ and $\calE^{-1}$ share the same $\bbZ$-rows. The optimality can be characterized by the following proposition. 

\begin{proposition}
   The compressed inverse in Eq.~\eqref{eq:comInv} satisfies $\calM_{\bbZ}\calE^{\prime -1}\calE = \calM_{\bbZ}$, i.e. the noise $\calE$ is completely canceled. The sampling cost $\calE^{\prime -1}$ is always less than that of the original inverse $\calE^{-1}$ in \eqref{eq:oriInv}, i.e. $\gamma^\prime \leq \gamma$ where $\gamma^\prime = \sum_{t=1}^T |Q^\prime_{\calB_t}|$ and $\gamma = \sum_{\calB\in \bbB} |Q_{\calB}|$.
\end{proposition}

We finally summarize the compression procedure in Algorithm~\ref{alg: compression}, which loops over the basis channels $\calB\in \bbB$ and categorizes them by their $\bbZ$-rows. The key step of this algorithm involves verifying whether the $\bbZ$-rows of two basis quantum channels are identical. While the $\bbZ$-rows take the form of a $2^n \times 4^n$ matrix which scales exponentially with the number of qubits, we leverage the fact that stabilizer circuits admit efficient classical simulation to make comparisons of $\bbZ$-rows computationally tractable. As detailed in Sec.~\ref{sec: Z-rows pair propagation}, our comparison algorithm features a computational complexity that is polynomial in $n$.  

Another consideration is that the size of $\bbB$ may give rise to scalability concerns. However, practical implementations detailed in Sec.~\ref{sec: practical noise inversion} circumvent this limitation: the net effect of noise can be efficiently captured either via propagated gate noise channels or by a channel set $\bbB$ of polynomial size in $n$.

\begin{algorithm}[H] 
  \caption{Noise compression}
  \label{alg: compression}
  \begin{algorithmic}[1]
  \Require{Basis set $\bbB$, and coefficient set $\mathbb{Q} = \{Q_{{\calB}}\vert \forall \calB \in \bbB\}$} 
  \Ensure{Compressed basis set $\bbB^\prime$, and coefficients $\mathbb{Q}^\prime = \{Q^\prime_{{\calB}}\vert\forall \calB\in \bbB^\prime \}$}
  \Statex
  Initialize $\bbB^\prime= \{\}, \mathbb{Q}^\prime = \{\}$ as empty sets.
  \For{each $\calB\in \bbB$}
    \State Set Marker = 0.
    \For{each $\calB^\prime\in \bbB^\prime$}
      \If{$\calB$ and $\calB^\prime$ have the same $\bbZ$-rows}
        \State $Q^\prime_{{\calB^\prime}}\longleftarrow Q^\prime_{{\calB^\prime}} + Q_{{\calB}}$, 
        Marker $\longleftarrow 1$.
        \State \textbf{break}. 
      \EndIf
    \EndFor
    \If{Marker = 0}
        \State $\bbB^\prime_i\longleftarrow \bbB^\prime_i \cup \{\calB\},\; \mathbb{Q}^\prime_i\longleftarrow \mathbb{Q}^\prime_i \cup \{Q_{{\calB}}\}$
      \EndIf
  \EndFor
  \end{algorithmic} 
  \end{algorithm}

\section{CNI-based robust quantum learning}
\label{sec: CNI-based shadow}
First, let us have a brief review of shadow estimation and then develop the CNI-based robust shadow estimation.  
Shadow estimation has found wide applications like in machine learning~\cite{Huang2022a,Haug2023,Schreiber2023,Jerbi2024}, quantum simulation~\cite{Huggins2022,Hadfield2022,Boyd2024,Wu2023,Somma2025,Gresch2025,li2024nearly}, and state purification~\cite{Seif2023,Zhou2024,Yang2024a,Grier30Jun--03Jul2024,liu2024auxiliary}. 
This technique constructs classical state representations through randomized measurements with a unitary ensemble $\bbU$~\cite{Huang2020a}. 
To be specific, in a single-run of the protocol, a quantum state $\dket{\rho}$ is evolved by the unitary channel $\calU$ sampled from some ensemble $\bbU$ and measured in the computational basis to get an outcome $b$; one prepares $\calU^\dagger \dket{b}$ on the classical computer. The expectation of the whole quantum-classical process effectively imposes on $\dket{\rho}$ a channel reading $\calM = \bbE_{\calU \in \bbU} \left[\calU^\dagger \calM_{\bbZ} \calU\right]$, and  $\calM$ usually has a known expression. The classical snapshots or called shadow ${\dket{\hat{\rho}}}=\calM^{-1}\calU^{\dag}\dket{b}$ is an unbiased estimator of $\dket{\rho}$. A collection of shadows can be obtained by repeating this process, and used to predict many properties. 

In the noisy case, the channel is changed to $\tilde{\calM} = \bbE_{\calU \in \bbU} \left[\calU^\dagger \calM_{\bbZ} \tilde{\calU}\right]$, where $\tilde{\calU}$ is the real noisy experiment. Suppose one still uses the original inverse $\calM^{-1}$, it will introduce considerable biases. Existing methods for handling noise in classical shadows have limitations. Standard robust shadow estimation mitigates the effect of noise by using the updated inverse channel $\tilde{\calM}^{-1}$~\cite{Chen2021, Koh2022, Onorati2024, Wu2024, Zhao2024}. However, it assumes the noise is independent of the random unitary, i.e., $\tilde{\calU}= \calE\calU$ with $\calE$ being identical across different $\calU$; otherwise the mitigation still introduces bias~\cite{Brieger2025}. 
Meanwhile, both robust shallow shadows~\cite{Hu2025, Farias2024} relax this assumption; however, they still assume that the noise is independent of the single-qubit gates. Specifically, the method in Ref.~\cite{Farias2024} accounts for gate-dependent noise in the majority of gates within shallow-shadow circuits; however, it still relies on the assumption that noise in the layer of the single-qubit gates before measurement is gate-independent. The direct application of QEM to classical shadow~\cite{Jnane2024} also relies on the assumption that the noise model is independent of single-qubit gates, and it may introduce large circuit-sampling overhead.

Here, we develop the CNI-based robust shadow estimation method by implementing the CNI protocol for each sampled unitary. This new robust method not only inherits all advantages of CNI but also introduces tremendous variance reduction when integrated with the thrifty shadow estimation framework~\cite{Helsen2023,Zhou2023}. 

\begin{figure}
  \begin{center}
    \includegraphics[width=1.0\linewidth]{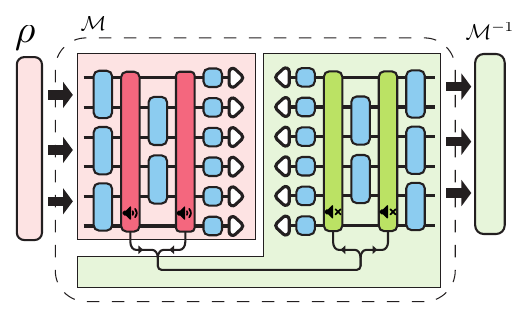}
    \caption{Illustration of CNI-based robust shadow estimation. The CNI protocol is implemented for every noisy random circuit, yielding an ideal randomized measurement map with noisy measurements.}
    \label{fig: CNI-shadow}
  \end{center}
\end{figure}

The CNI-based robust shadow estimation is illustrated in Fig.~\ref{fig: CNI-shadow}. Consider a noisy random unitary $\tilde{\calU} = \calE_{\calU}\calU$ with the overall noise $\calE_\calU = \tilde{\calU}\calU^\dagger$ depending on $\calU$. Suppose there exists $\calE^{\prime}_{\calU}$ satisfying $\calE^{\prime}_{\calU}\calM_{\bbZ} = \calM_{\bbZ} \calE_{\calU}$. We can let $\tilde{\calU}^{\prime - 1} = \calU^\dagger \calE^{\prime -1}_{\calU}$ and get $\tilde{\calU}^{\prime - 1}\calM_{\bbZ}\tilde{\calU}  = \calU^\dagger \calM_{\bbZ} \calU$, such that by taking expectation on $\calU$, one recovers the ideal randomized measurement: 
\begin{equation}
    \bbE_{\calU \in \bbU} \left[ \tilde{\calU}^{\prime - 1}\calM_{\bbZ}\tilde{\calU} \right] = \bbE_{\calU \in \bbU} \left[\calU^\dagger \calM_{\bbZ} \calU\right] = \calM, 
\end{equation}
as the error is mitigated in every random unitary. 

Our protocol executes random noisy circuit $\tilde{\calU}\dket{\rho}$ on a quantum computer, and realizes $\tilde{\calU}^{\prime - 1}$ by simulating its actions on the measurement outcomes $b$ of the noisy circuit. Since $\tilde{\calU}^{\prime - 1} = \calU^\dagger \calE^{\prime -1}_{\calU}$ and $\calU$ is selected to be classically simulable by shadow estimation protocol, the core lies in simulating $\calE^{\prime -1}_{\calU}$. This simulation is achieved by sampling basis channel $\calB_{\calU}$ with probability $\Prob{\calB\vert\calU} = |Q^\prime_{\calU, \calB}|/\gamma^\prime_{\calU}$, which is derived from the quasi-probability decomposition of $\calE^{\prime -1}_{\calU}$. Specifically, the decomposition takes the form  $\calE^{\prime -1}_{\calU} = \sum_{\calB_\calU\in\bbB}Q^\prime_{\calU, \calB}\calB_\calU$, with $\gamma^\prime_{{\calU}} = \sum_{\calB_\calU\in\bbB}|Q^\prime_{\calU, \calB}|$. In practice, we can sample a number of random circuits, followed by measuring each circuit over a large number of shots. Additionally, we generate multiple instances of the basis channel to simulate noisy inversion. Leveraging this sampling strategy, we introduce an unbiased estimator along with an analysis of its performance in the theorem below. The detailed proof of this theorem is deferred to Appendix~\ref{app: proof of robust shadow with CNI}.

\begin{theorem}[Performance guarantee of CNI-based robust shadow estimation]
\label{th:perf_robust_shadow}
With one sample of random unitary $\calU$ from an ensemble $\bbU$, one-shot random noisy measurement outcome $b$, and one instance of random basis channel $\calB$, the 
robust shadow estimator of the ideal expectation value $o=\Tr(O\rho)$ reads 
\begin{equation}
  \hat{o} = \gamma^\prime_{\calU} \frac{Q^\prime_{\calU, \calB}}{|Q^\prime_{\calU, \calB}| } 
  \dbra{O}\calM^{-1}\calU^\dagger \calB_{\calU} \dket{b}. 
\end{equation}

Suppose $M$ random circuits are sampled, and each is measured for $K$ shots. For each shot, $L$ random instances of $\calB$ are generated, yielding $M \times K \times L$ unbiased estimators of $o$ in total. 
Let $\hat{o}_{\mathrm{tot}}$ denote the mean of the estimators. The variance of the total estimator satisfies  
\begin{eqnarray}
  \Var{\hat{o}_{\operatorname{tot}}} &\leq& \frac{1}{M}\left\{
    \frac{\gamma^{\prime 2}_{\operatorname{max}}}{K}\left[\frac{1}{L} \norm{O}_{\operatorname{NS1} }^2 + \left(1 - \frac{1}{L}\right) \norm{O}_{\operatorname{NS2} }^2\right] \right. \nonumber\\
    && \left. \quad + \left(1 - \frac{1}{K}\right)\norm{O}_{\operatorname{XS}}^2
  \right\}. \label{eq: var CNI-based RS}
\end{eqnarray}
where $\gamma^{\prime}_{\operatorname{max}} = \operatorname{max}_{\calU}\gamma^\prime_\calU $ and  $\norm{O}_{\operatorname{NS2}} \leq \norm{O}_{\operatorname{NS1}}$ are the noisy shadow norms we defined in Eq.~\eqref{eq: NS1} and Eq.~\eqref{eq: NS2} in Appendix~\ref{app: proof of robust shadow with CNI}, and $\norm{O}_{\operatorname{XS}}$ is the (noiseless) cross shadow norm~\cite{Zhou2023}. 
\end{theorem}

Theorem~\ref{th:perf_robust_shadow} establishes CNI-based robust shadow estimation as a resource-efficient and noise-robust framework for learning quantum properties under general gate-dependent noise, while also highlighting significant advantages of our method over existing methods. First, the noise-induced cost is reduced to $O(\gamma^{\prime 2}_{\operatorname{max}}/{MK})$, whereas conventional probabilistic error cancellation~\cite{Temme2017,Endo2018,Nguyen2022} saturates at $O(\gamma^{2}_{\operatorname{max}}/{M})$ (see Appendix~\ref{app: variance of robust shadow with cPEC} for full derivation). Since sampling random circuits is far slower than sampling measurement outcomes~\cite{Karalekas2020}, $K$ is set to $K\gg1$ in practice; the $1/K$ suppression therefore yields a dramatic improvement. Second, in contrast to existing methods~\cite{Nguyen2022, Chen2021}, particularly standard robust shadow estimation~\cite{Chen2021} which is only reliable under gate-independent noise~\cite{Brieger2025}, our method enables unbiased estimation even in the presence of gate-dependent noise. These aforementioned advantages are numerically demonstrated in the next section, with key results presented in Fig.~\ref{fig: numerical main-MS}. 

Apart from the advantages mentioned above, the CNI-based robust shadow estimation protocol offers additional benefits unattainable with existing approaches. These include: reduced sampling complexity via noise compression ($\gamma^\prime_{\operatorname{max}}\leq \gamma_{\operatorname{max}}$); and variance reduction using only classical resources, which is achieved by increasing $L$, given that $\norm{O}_{\operatorname{NS2}} \leq \norm{O}_{\operatorname{NS1}}$. Numerical simulations illustrating these advantages are presented in Fig.~\ref{fig: numerical main}.

Theorem~\ref{th:perf_robust_shadow} provides a rigorous characterization of how the noise-induced factor $\gamma^\prime$ affects the variance of shadow estimators. Beyond characterizing the variance response to noise, it is equally important to investigate the upper bounds on the noisy shadow norms, as these bounds delineate how the cost of robust shadow estimation scales with system size and observable structure. 

\begin{proposition}[Upper bound of noisy shadow norm]
  \label{prop:NS_bound}

Let $g = \min_{\calU} \Prob{\calI|\calU}$ denote the minimum conditional probability of being noise-free, and 
\begin{equation}
  h=\max \left\{\max _{\calU, \calB}\left\{\frac{\operatorname{Pr}(\calB \mid \calU)}{\operatorname{Pr}(\calB)}\right\}, \max_{\calU}\left\{\frac{\operatorname{Pr}(\calI \mid \calU)-g}{\operatorname{Pr}(\calI)-g}\right\}\right\},  
\end{equation}
where $\Prob{\calB} = \bbE_{\calU}\Prob{\calB\vert\calU} $ is the marginal distribution of $\calB$.
The noisy shadow norm $\norm{O}_{\operatorname{NS1}}$ appearing in Eq.~\eqref{eq: var CNI-based RS}
is upper bounded by 
\begin{eqnarray}
    \norm{O}_{\operatorname{NS1}}^2 &\leq& g \norm{O}_{\operatorname{S}}^2 + (1 - g) h \max_{\bar{\calB}} \norm{\bar{\calB}^\dagger(O)}_{\operatorname{S}}^2,
\end{eqnarray} 
where $\bar{\calB} = \calM^{-1}\calU^\dagger \calB \calU \calM$ and $\norm{\cdot}_{\operatorname{S}}$ is the ideal shadow norm~\cite{Huang2020a}. 
\end{proposition} 

\begin{figure*}
  \begin{center}
    \includegraphics[width=1\linewidth]{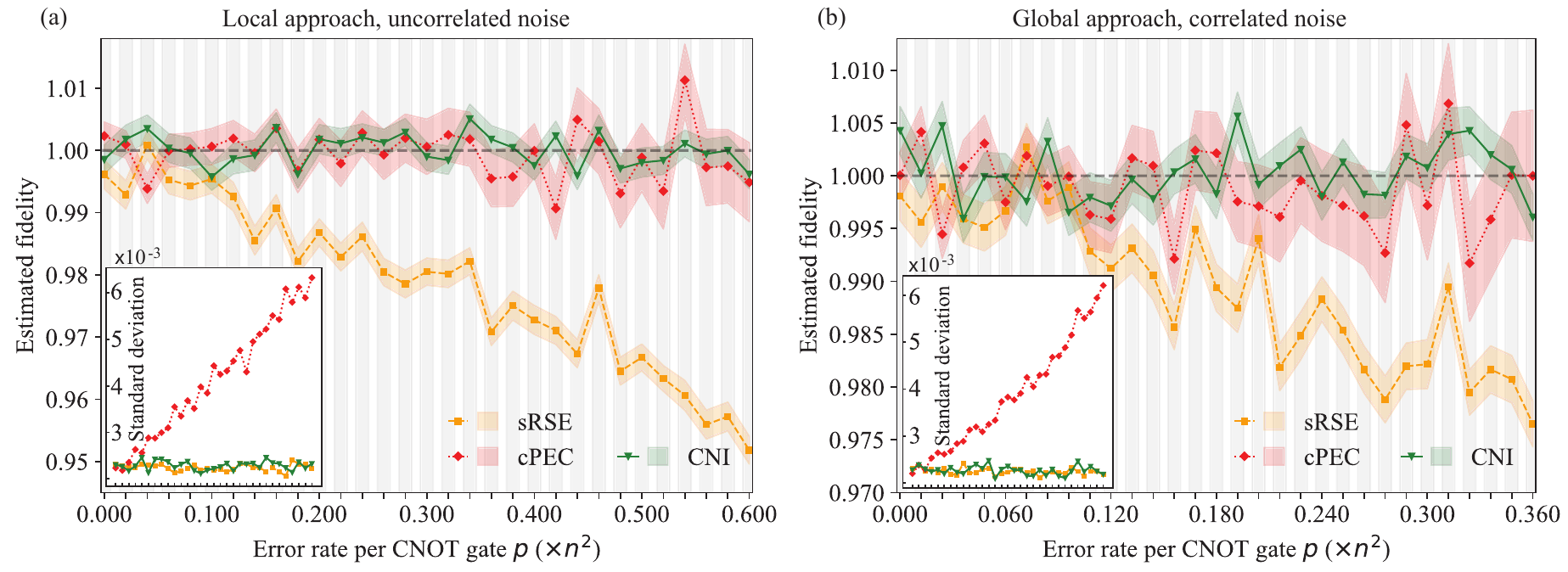}
    \caption{    
Numerical comparison of our method and existing methods under the multi-shot measurement scheme. In the figure, CNI stands for the proposed CNI-based robust shadow estimation, sRSE stands for standard robust shadow estimation and cPEC stands for conventional probabilistic error cancellation. The fidelity of a $4$-qubit GHZ state is estimated at a cost of $M=10^3,K=10^3,L=1$ across varying error rates; this estimation procedure is repeated 320 times. Data points represent the mean values of the repeated estimations, while the shaded areas indicate one standard deviation. The standard deviations are also plotted separately in the inset. The calibration cost for sRSE is $M=3.2\times 10^5, K=1,L=1$. (a) Local approach for uncorrelated noise.   (b) Global approach for correlated noise. } 
    \label{fig: numerical main-MS}
  \end{center}
\end{figure*}
Proposition~\ref{prop:NS_bound} rigorously establishes the impacts of noise on the shadow norm; the proof of this proposition can be found in Appendix~\ref{app: Comparison of noisy shadow norms and ideal shadow norm}. The factor $g$ represents the noise severity, $h$ represents the correlation between the noise model and the random unitary, and $\bar{\calB}^\dagger (O)$ is the deviated (by noise) observable. We explain this in detail in the following. 

When the total error rate $p$ is small, the factor $g$ satisfies $g \approx 1 - p + O(p^2)$. In the noiseless case, $g=1$ and the noisy shadow norm becomes the ideal shadow norm.

The factor $h$ quantifies the ceiling of the correlation between $\calU$ and $\calB$.  In fact, $\log(h)$ is an upper bound of the mutual information between $\calU$ and $\calB$.
$h=1$ if the noise model is independent of $\calU$ (which results in $\Prob{\calB\vert \calU} = \Prob{\calB}$ for all $\calU$ and $\calB$), and $h> 1$ if the noise model is correlated with $\calU$. In random quantum circuits, the overall noise model can be considered as global depolarizing noise with small perturbations~\cite{Qin2023, Dalzell2021}, therefore $h$ is a small number. 

As for the deviated observable $\bar{\calB}^\dagger (O)$, we find that the upper bound of the shadow norm of the deviated observable $\norm{\bar{\calB}^\dagger(O)}_{\text{S}}$ is the same as that of the original observable $\norm{O}_{\text{S}}$ in practical cases. Indeed, this holds true for both random Clifford measurements and random Pauli measurements under the  Pauli error model. We can show that $\norm{\bar{\calB}^\dagger(O)}_{\text{S}}^2 \leq \Tr(O^2)$ for random Clifford measurements and $\norm{\bar{\calB}^\dagger(O)}_{\text{S}}^2 \leq 4^k \norm{O}_{\infty}^2$ for random Pauli measurements and weight-$k$ observables $O$.

\section{Numerical simulations} 
\label{sec: numerical }

\begin{figure*}
  \begin{center}
    \includegraphics[width=1\linewidth]{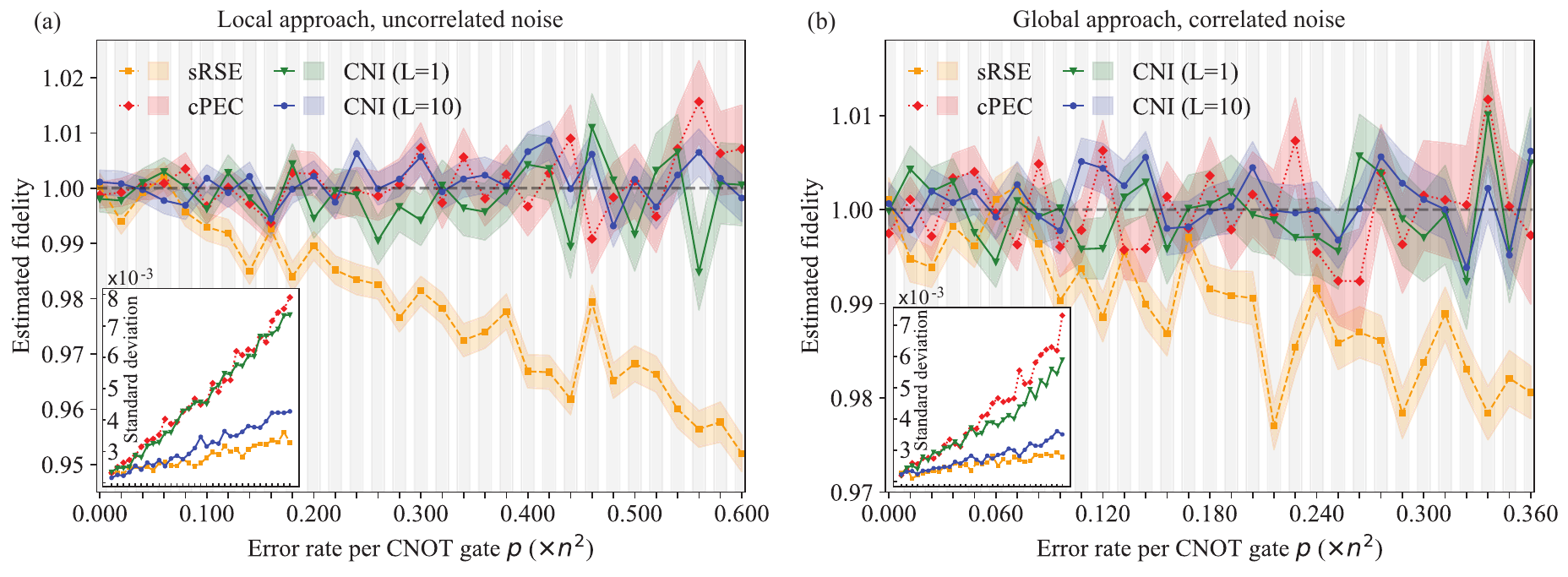}
    \caption{
Numerical comparison of our method and existing methods under the single-shot measurement scheme. The details of the plots are identical to those in Fig.~\ref{fig: numerical main-MS} except that the estimation cost is set to $M=10^3,K=1,L=1$. Additionally, the results of CNI with an estimation cost of $M=10^3,K=1,L=10$ are included.  
}
    \label{fig: numerical main}
  \end{center}
\end{figure*}

We demonstrate CNI-based robust shadow with numerical simulations. To comparatively illustrate the advantages of our new protocol, we also simulate previous protocols, including standard robust shadow estimation (sRSE)~\cite{Chen2021, Koh2022} and conventional probabilistic error cancellation (cPEC)~\cite{Temme2017, Endo2018, Jnane2024}, under the same noise conditions. 

The simulated task is the fidelity estimation of stabilizer states using noisy random Clifford measurements. Specifically, we take stabilizer states $\rho$ as input to the protocols, generate classical shadow estimators $\hat{\rho}$, and then compute the overlap $\hat{o} = \Tr(\hat{\rho}\rho)$. When noise is perfectly canceled, the theoretical expectation of $\hat{o}$ equals 1. The variance of $\hat{o}$ reflects the cost of the protocols.

The random Clifford unitaries are decomposed into circuits consisting of CNOT gates and single-qubit Clifford gates. In the numerical simulations, we only consider that CNOT gates are noisy while single-qubit gates are noise-free. Notice that in realistic situations, the noise can be uncorrelated or correlated, meaning that the noise on different gates may have inseparable joint (quasi)-distributions. We consider both correlated and uncorrelated noise in our numerical simulations. 
For uncorrelated noise, the local approach suffices to preserve an unbiased estimator. For correlated noise, the global approach is required for the same guarantee. We detail both local and global approaches in the Methods section.

For the uncorrelated gate-wise noise model and local approach, we consider the two-qubit bit-flip error on every CNOT gate: 
$\calE = (1-p) \calI^{\otimes 2} + p \calX^{\otimes 2}$, where $p$ is the error rate and $\calX(\cdot) = X\cdot X$. To demonstrate the CNI-based robust shadow estimation, we simulate the noise inverse $\mathcal{E}^{-1}=\frac{1}{1-2p}\left[(1-p)\mathcal{I}^{\otimes2}-p\mathcal{X}^{\otimes2}\right]$ using the introduced Monte-Carlo method at the positions of the corresponding gate during the post-processing phase. Noise compression is incorporated by checking whether the propagated XX gate is a tensor product of $I$ and $Z$  gates. If it is, we do not simulate the inverse of the noise. The above procedure produces unbiased estimates of the ideal shadow snapshots. Then, applying the ideal shadow inversion $\calM^{-1}(\cdot) = (2^n+1)(\cdot) - \bbI$, we complete the CNI-based robust shadow estimation procedure. 

For the correlated noise model and global approach, we consider the following noise model. For a random Clifford unitary  $\calU = \prod_{i=1}^{N} \calU_i$ composed of $N$ gates, the corresponding noisy circuit is 
\begin{equation}
  \tilde{\calU} =(1-Np)\calU + p \sum_{i=1}^{N}\tilde{\calU}^{(i)}, 
  \label{eq: noise global-1}
\end{equation}
where $\tilde{\calU}^{(i)}$ is subjected to a two-qubit depolarizing error $\frac{1}{15}\sum_{\calP\in \bbP^{\otimes 2}\setminus\{\calI^{\otimes 2}\} }\calP$ at the $i$-th gate while the other gates are noiseless. This noise model exhibits significant temporal error correlation, as the probability that two two-qubit gates are erroneous simultaneously is zero, which is not equal to the product of the error rates of the two gates. Moreover, the model approximately captures the effective noise arising from gate-wise depolarizing channels with small error rates (See Appendix~\ref{app: series expansion}). To obtain the compressed noise model, we input the propagated error operations $\calP$ from every noisy gate and their corresponding coefficients $p/15$ into Algorithm~\ref{alg: compression}. Then, we simulate the inverse of the compressed noise model using the global approach 
detailed in the Methods section. 

We briefly introduce the techniques of sRSE and cPEC in the following. The sRSE method does not perform any suppression on each noisy circuit; instead, it applies the inverse of the noisy randomized measurement map on the shadow snapshots. For random Clifford measurement, it reads $\tilde{\calM}^{-1}=\frac{1}{r} \calI - \frac{1-r}{r} \calD$, where $r$ depends on the noise and is calibrated by $\hat{r} = \frac{2^n\hat{R} - 1}{2^n-1}$, where $\hat{R} = \dbra{0}\calU\dket{b}$ with $b$ being a measurement outcome of $\tilde{\calU}\dket{\rho}$. The cPEC method mirrors CNI but differs in two key ways: it cancels noise by sampling random quantum circuits rather than classical simulation, and it inverts the full noise model instead of the compressed one. 

We illustrate the numerical results in Fig.~\ref{fig: numerical main-MS} and Fig.~\ref{fig: numerical main}. These results demonstrate that the CNI-based robust shadow yields unbiased estimates with low variance, whereas previous methods are constrained by a fundamental bias-variance trade-off. In comparison to sRSE, the CNI-based method maintains unbiased estimation even when noise is dependent on the random unitary—a scenario where sRSE fails to preserve unbiasedness. Additionally, relative to cPEC, the CNI-based method achieves substantially lower variance, highlighting its advantage in efficiency. 

Our method’s variance reduction advantages are clearly demonstrated across two figures, each highlighting a distinct, complementary mechanism. Fig.~\ref{fig: numerical main-MS} employs a multi-shot scheme to showcase our method’s strength in sampling efficiency driven by circuit reuse. As shown in Eq.~\eqref{eq: var CNI-based RS}, the noise-induced variance (quantified by $\gamma_{\operatorname{max}}^2$) is reduced by a factor of 
$1/(MK)$, whereas this reduction factor is $1/M$ for cPEC. 
In contrast, Fig.~\ref{fig: numerical main} uses a single-shot scheme to focus on variance reduction from two key strategies: noise compression and classical resource exploiting. Under this scheme, the numerical results demonstrate that our method still delivers superior variance reduction via implementing noise compression (transforming $\gamma$  to $\gamma^\prime$), and incorporating an additional classical cost of $L = 10$ to leverage classical computing resources. Neither of the two existing methods can reduce variance by adjusting $L$ to utilize classical resources, as their classical operations are fixed to ideal unitary circuits~\cite{Chen2021,Jnane2024}. 

Another advantage revealed by the numerical results is that our method achieves unbiased estimation—even when noise is dependent on the sampled random unitary—whereas existing methods cannot. As illustrated in both figures, our method enables unbiased fidelity estimation, while the sRSE method still introduces bias that increases with the error rate. Although the results indicate that cPEC also yields unbiased estimation, this outcome is specific to our numerical simulation, where single-qubit gates were assumed to be noise-free. In realistic scenarios, where noise is dependent on single-qubit gates, cPEC would still be expected to introduce bias.

\section{Methods}

{
  Building on the theoretical frameworks presented in the earlier sections, we now turn to devise practically efficient strategies to realize the key components of the proposed protocols. Specifically, we present methods for practical noise inversion, the semi-definite approach to factorize the noise into classically tractable and intractable components, and efficient $\bbZ$-rows comparison, which is an essential technique for noise compression.
}

\subsection{Practical noise inversion} 
\label{sec: practical noise inversion}

{

Inverting a general quantum channel is challenging, as the size of its PTM grows exponentially with the number of qubits. In practical scenarios, however, the total noise typically originates from local gate noise that propagates throughout the circuit (as shown in Fig.~\ref{fig: practical implementation} (a)), or conforms to a spacetime-distributed noise model that can be well approximated by a quantum channel whose quasi-probability decomposition consists of only a polynomial number of basis channels (as shown in Fig.~\ref{fig: practical implementation} (b)). In these two cases, noise inversion can be efficiently implemented without compromising its efficacy. We propose two approaches, namely the local and global approach, to implement the noise inversion in practice. 
}

\begin{figure}
  \begin{center}
    \includegraphics[width=1\linewidth]{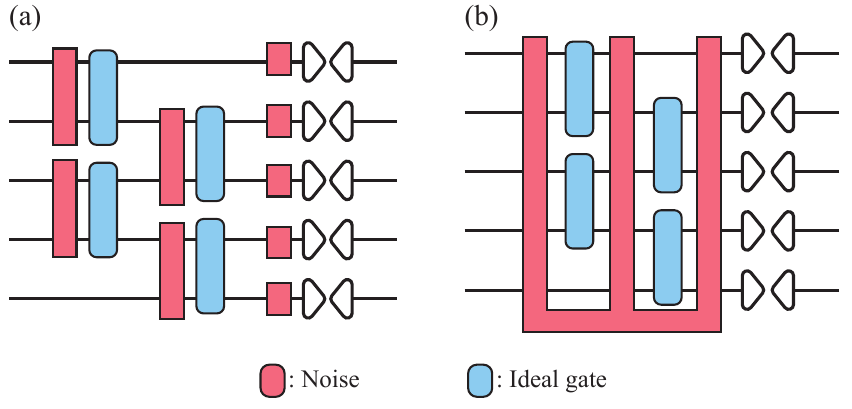}
    \caption{Illustration of efficiently invertible noise models for quantum circuits: (a) Local gate noise models. (b) Correlated noise models distributed across the circuit spacetime with a limited number of terms in the quasi-probability decomposition.}
    \label{fig: practical implementation}
  \end{center}
\end{figure}

\subsubsection{Local approach}

\label{sec: local approach}

The local approach applies the proposed protocol to each noisy gate and is suitable for well-characterized circuits with gate-level noise. Suppose the ideal circuit is $\calU$ and its noisy implementation is $\tilde{\calU}$. With the assumption of gate-level noise, we can write the noisy circuit as $\tilde{\calU} = \calE_{N+1} \tilde{\calU}_{N} \tilde{\calU}_{N-1}\cdots\tilde{\calU}_1$ with $\tilde{\calU}_{i} = \calU_{i} \calE_{i}$, where $\calE_{i}$ denotes the noise channel associated with the $i$-th gate, and $\calE_{N+1}$ represents the noise on the measurement. 

Noise compression and inversion are therefore implemented in a gate-wise manner. Denote the noise on the $j$-th gate as $\calE_j = \sum_{\calB\in \bbB^{(j)}} C_{j, \calB} \calB$, where $\bbB^{(j)}$ is the set of basis channels supported on a few qubits. It is possible that, due to cross-talk, the noise on this gate may affect not only these two qubits but also their physically neighboring qubits; however, the total number of qubits affected by the noise is constant and independent of the system size. Therefore, due to the locality of the noise model, the coefficients $C_{j, \calB}$ can be efficiently determined using gate-set tomography~\cite{Greenbaum2015, Nielsen2021} and the inverse of the noise $\calE_j^{-1} = \sum_{\calB\in \bbB^{(j)}} Q_{j, \calB} \calB$ can be computed efficiently using a classical computer. After propagation, the noise inversion becomes $\bar{\calE}_j^{-1} = \sum_{\calB\in \bbB^{(j)}} Q_{j, \calB} \bar{\calU}_j\calB \bar{\calU}^\dagger_j$, where $\bar{\calU}_j = \calU_{N+1}\calU_{N}\cdots\calU_{j+1}$ denotes the ideal circuit after the $j$-th gate. Finally, we take $\bar{\calE}_j^{-1}$ as the input to Algorithm~\ref{alg: compression} to obtain the compressed noise inversion, and then use the Monte-Carlo method to simulate the action of the noise inversion on the measurement outcomes. 

The local approach yields unbiased estimation of the ideal result,  provided that the gate-level noise is well-characterized. The limitation of this method lies in its reliance on the gate-level noise assumption. Additionally, this method has a relatively low noise compression ratio, as it only considers the compression of intra-gate noise and not that of inter-gate noise. To ensure our method is applicable under general noise models (without the local noise assumption) while further improving the compression ratio, we propose the global approach below. The discussion on why the global approach maintains high compression ratio even for large quantum circuits is deferred to Appendix~\ref{app: discussion on the compression ratio}. 

\subsubsection{Global approach}
\label{sec: global approach}

In a general setting, the noise in the circuit may exhibit significant spatial and temporal correlations, and the noisy circuit is no longer a concatenation of noisy gates. In this case, we need to employ the global approach. 


The subsequent steps of our protocol require knowledge of the noise model.  Although the circuit noise model is distributed over the space-time of the circuit, there are efficient methods to learn the net noise model $\calE = \tilde{\calU}\calU^\dagger$, where $\calU$ is the ideal circuit and $\tilde{\calU}$ is the noisy circuit. The efficient learning of $\calE$ stems from two aspects. First, there are efficient representations of the noise model. Specifically, since any quantum channel can be written as $\calE = \sum_{\calB\in \bbB}C_{\calB}\calB$ with $C_{\calB}$ being real coefficients, we can model a general noise channel as
\begin{equation} 
  {\calE} =  \eta\sum_{\calB\in \bbB} \Prob{\calB} \operatorname{sgn}(\calB) \calB, \label{eq: noise global}
\end{equation} where $\eta$ is a scalar and $\bbB$ is a set of basis channels. We can encode the probability distribution $\Prob{\calB}$ into an Ansatz, such as a neural network, a tensor network, or simply a list of probabilities defined over a polynomial-sized basis set $\bbB$, enabling efficient sampling of $\calB$ from the distribution. Second, utilizing the efficiency of the classical simulation of ${\calE}\calU$, we can learn the noise model by minimizing the difference between the actual distribution $\Prob{b\vert \rho} = \dbra{b}\tilde{\calU}\dket{\rho}$ and the model distribution  $\hat{\operatorname{Pr}}(b\vert \rho) = \dbra{b}{\calE}\calU\dket{\rho}$, where $\rho$ is sampled from a set of complete states, such as the tensor product of the eigenstates of the Pauli operators $X, Y$ and $Z$. This can be done by optimizing the loss function~\cite{Wang2021f,Strikis2021}, the cross entropy~\cite{Neill2018, Boixo2018}, or by employing Bayesian noise learning~\cite{Hu2025}. 

The core step of our protocol is simulating the inverse of the noise model $\calE$. Though it is intractable to find the explicit expression of $\calE^{-1}$ in general, the Neumann series expansion of $\calE^{-1}$ enables us to simulate it probabilistically~\cite{Xie2025}. The series expansion of $\calE^{-1}$ is 
\begin{equation}
  \calE^{-1} = \frac{1}{\eta \Prob{\calI}} \sum_{l=0}^\infty \left( - \frac{\Prob{\calN}}{\Prob{\calI}} \calN  \right)^l,
  \label{eq: inverse Neumann}
\end{equation} where $ \Prob{\calN} = \sum_{\calB\in \bbB\setminus \{\calI\}} \Prob{\calB}$ is the noise probability, and $\calN =  \sum_{\calB \in {\bbB}\setminus\{ \calI\}}\frac{\Prob{\calB}}{\Prob{\calN}} \operatorname{sgn}(\calB) \calB$ is the noise. By generating a random number $ l $ with probability $ \left[ 1 - \frac{\Prob{\calN}}{\Prob{\calI}}\right] \left[\frac{\Prob{\calN}}{\Prob{\calI}}\right]^l $, which can be achieved via inverse transform sampling, then generating $ l $ random signed basis channels $ \operatorname{sgn}(\calB_1)\calB_1, \operatorname{sgn}(\calB_2)\calB_2, \dots, \operatorname{sgn}(\calB_l)\calB_l $ with probability $ \frac{\Prob{\calB}}{\Prob{\calN}}$, we obtain the following unbiased estimator of the inverse:
\begin{equation}
  \hat{\calE}^{-1} = \frac{1}{\eta \left[\Prob{\calI} - \Prob{\calN}\right]}(-1)^l \prod_{i=1}^l \operatorname{sgn}(\calB_i) \calB_i. 
  \label{eq: estimator neumann inversion}
\end{equation} 
This estimator introduces a sample overhead of $\gamma = \frac{1}{\eta \left[\Prob{\calI} - \Prob{\calN}\right]}$. For detailed derivations of the estimator presented in Eq.~\eqref{eq: estimator neumann inversion}, refer to Appendix~\ref{app: Simulating the noise inversion without explicit inversion}. 

Noise compression can also be implemented to reduce the variance for the global approach. Since $\calE^{-1}$ in Eq.~\eqref{eq: inverse Neumann} is an infinite series, it is impossible to implement noise compression for $\calE^{-1}$ directly. Instead, we can implement noise compression for $\calE$, and the sample overhead can still be reduced to a lower value $\gamma^\prime$, which we prove in Appendix~\ref{app: Simulating the noise inversion without explicit inversion}.  
For the noise model defined in Eq.~\eqref{eq: noise global}, the approach to noise compression hinges on the encoding scheme of the coefficients. If these coefficients are represented as a sum over a polynomial-sized set $\bbB$, they can be directly fed into Algorithm~\ref{alg: compression} to achieve the desired noise compression.
However, when the probability distribution is encoded via a neural network or tensor network, the noise compression process becomes less straightforward; nevertheless, practical implementation strategies still exist. For instance, in the context of a Pauli error model, noise compression can be accomplished by tracing out the neurons (in the case of neural network encoding) or contracting the network nodes (in the case of tensor-network encoding) that correspond to the probabilities of Pauli-Z components.

\subsection{{Optimizing the noise factorization}}
\label{sec: optimal in-circuit}

{
The factorization $\calE=\calE_c\calE_q$ in Eq.~\eqref{eq: qc factorization} provides a constructive approach for compensating CNI via partial quantum sampling. However, obtaining this factorization requires computing the inverse $\calE_{00}^{-1}$ of the exponentially large block. To address this, we propose a semi-definite program to efficiently determine the optimal quantum-classical factorization that minimizes the quantum sampling overhead $\gamma_q$. For propagated local or truncated global noise models, which involve only a small basis set $\bbB$, the optimal quantum-classical factorization can be found directly, avoiding the need to explicitly factorize the large matrix as indicated by Proposition~\ref{prop: Quantum-classical factorization}.  
}

{
The Ansatz for the quantum factor is $\calE_q^{ -1}=\sum_{\calB\in\bbB}Q_{\calB,q}\calE_\calB\calB$, where $\{\calE_\calB\calB\}$ are the noisy basis operations; while the Ansatz for the classical factor is $\calE_c^{ -1}=\sum_k Q_{k,c}\calC_k$ where $\{\calC_k\}$ is a set of ideal and classically tractable operations. We note that such $\{\calC_k\}$ includes Pauli channels and the linear combinations of non-Pauli Clifford channels whose upper-right blocks cancel, e.g.,\ $\dketbra{I}{I}+\dketbra{X}{Z}=\tfrac12(\calH+\calS_y)$ with $\calH$ the Hadamard and $\calS_y$ the $\tfrac\pi2$ rotation about $Y$. 
The optimal factorization with minimal quantum sampling overhead is obtained by the semi-definite program 
}
\begin{equation}
  \gamma_q =\min_{\{Q_{\calB,q}\},\{Q_{k,c}\}}\ \sum_{\calB\in\bbB}\absLR{Q_{\calB,q}}\quad\mathrm{s.t.}\quad \calE_q^{ -1}\,\calE_c^{ -1}\,\calE=\calI ,
  \label{eq: semi-def main}
\end{equation}
{
which is posed entirely over the small basis set and never requires the dense inverse $\calE_{00}^{-1}$.
The factorization obtained by the semi-definite program in Eq.~\eqref{eq: semi-def main} improves on the explicit factorization given by Eq.~\eqref{eq: qc factorization} by obtaining a smaller quantum sampling overhead $\gamma_q$, as long as the optimum is achieved. The proofs regarding the semi-definite program are given in Appendix~\ref{app: minimum cost in-circuit}.
}

\subsection{{Efficient $\bbZ$-rows comparison}}
\label{sec: Z-rows pair propagation}

Comparing the $\bbZ$-rows of two channels is the central sub-routine for noise comparison. Although the size of the $\bbZ$-rows grows exponentially with the number of qubits, the comparison task can be executed efficiently for Clifford channels. We introduce the efficient algorithm here. 

We begin with the notations and preliminaries. 
For a general noise model distribution over the space-time of the circuit, the objective is to compare the $\bbZ$-rows of $\bar{\calB}$ and $\bar{\calB}^\prime$, where $\bar{\calB}^\dagger=\prod_{j=1}^{N+1}\bar{\calB}_j^\dagger$ with $\bar{\calB}_j^\dagger=\bar{\calU}_j^\dagger\calB_j^\dagger\bar{\calU}_j$, and likewise for $\bar{\calB}^{\prime\dagger}$. Here, $\bar{\calU}_j = \calU_{N+1}\calU_{N}\cdots\calU_{j+1}$ denotes the ideal circuit after the $j$-th gate, $\calB_j$ represents an basis Clifford channel at the $j$-th gate, and $\bar{\calU}_j \calB_j \bar{\calU}_j^\dagger$ corresponds to the propagated Clifford channel to the end of circuit.

In general, $\calB_j$ and $\calB^\prime_j$ are drawn from the universal basis channels~\cite{Endo2018}, which comprises not only unitary Clifford operations but also projection operations such as $\dketbra{0}{0}$ and $\dketbra{1}{1}$. This basis channel allows us to efficiently simulate the operation of $\bar{\calB}$ and $\bar{\calB}^\prime$ using the algorithm presented in Ref.~\cite{Aaronson2004}, which is also illustrated in Appendix~\ref{sec: Efficient simulation of universal basis channels}. The only modification required is the inclusion of an additional factor to track the trace of the resulting state since the projections are not trace-preserving. 

Comparing the $\bbZ$-rows of $\bar{\calB}$ and $\bar{\calB}^\prime$ is equivalent to comparing the $\bbZ$-columns of $\bar{\calB}^\dagger$ and $\bar{\calB}^{\prime \dagger}$.
First, we consider the case where $\calB_j$ and $\calB^\prime_j$ are unitary Clifford operations for all $j = 1, 2, ..., N+1$. In this case, it holds that 
\begin{equation}
    \bar{\calB}^\dagger (P_1 P_2) = \bar{\calB}^\dagger (P_1)\bar{\calB}^\dagger(P_2) \label{eq: Cliff on Pauli}
\end{equation} for arbitrary Pauli operators $P_1$ and $P_2$. 
Because $\{Z_i \vert i \in [n]\}$ generates all sign-free Pauli-Z operators, $\bar{\calB}^\dagger$ and $\bar{\calB}^{\prime \dagger}$ have the same $\bbZ$-columns if $\bar{\calB}^\dagger(Z_i) = \bar{\calB}^{\prime \dagger}(Z_i)$ for all $i\in [n]$. The computational cost of this $\bbZ$-rows comparison algorithm is $O(Nn^2)$ bit operations. 
This complexity arises because we must simulate the conjugation of Clifford operations on each of the $n$ Pauli-Z operators $Z_i$, and each such conjugation requires $O(Nn)$ bit operations. 

Second, we consider the general case, where $\calB_j$ and $\calB'_j$ can be an operation of either a Clifford gate or a projector. In this case, Eq.~\eqref{eq: Cliff on Pauli} is violated. For example, $\dketbra{0}{0}$ maps $I=XX$ to $\frac{1}{2} (I + Z)$ but maps $X$ to zero. Consequently, the method employed for the unitary case cannot be applied directly. Nevertheless, we propose an alternative efficient algorithm below.

We write the condition of $\bar{\calB}$ and $\bar{\calB}^\prime$ having the same $\bbZ$-rows as 
\begin{equation}
    \bar{\calB}^{\dagger} \dket{b} = \bar{\calB}^{\prime \dagger} \dket{b}, \; \forall \; b \in \{0, 1\}^n,  \label{eq: Z-blocck same}
\end{equation} 
whose necessary and sufficient condition is that, for every $b\in\{0,1\}^n$, the sub-normalized stabilizer states $\bar{\calB}^\dagger\dket{b}$ and $\bar{\calB}^{\prime \dagger}\dket{b}$ coincide, i.e., they have the same stabilizer generators.  Note that the sub-normalization is due to that $\calB_j$ can be projectors. 

We note some preliminaries on stabilizer state simulation before presenting the $\bbZ$-rows comparison algorithm. 
A sub-normalized stabilizer state can be efficiently represented by the triplet
\begin{eqnarray}
  \left[t, c, A\right],
\end{eqnarray}
where $t$ is the trace of stabilizer, $c$ is an $n$-element binary vector and $A$ is an $n\times 2n$ binary matrix. $c$ and $A$ specify the generators of the stabilizer group of the state, e.g., the $k$-th generator is
\begin{equation}
P_{k}=(-1)^{c^{(k)}} \prod_{r=1}^n(-i)^{A^{(k,r)} A^{(k, r+n)}} Z_r^{A^{(k,r)}} X_r^{A^{(k,r+n)}},  
\end{equation} where $Z_r = I^{\otimes (r-1)} \otimes Z \otimes I^{\otimes (n-r)}$ and $X_r = I^{\otimes (r-1)} \otimes X \otimes I^{\otimes (n-r)}$, $c^{(k)}$ is the $k$-th element of $c$ and $A^{(k,l)}$ denotes the element on the $k$-th row and $l$-th column. 

We denote the stabilizer representation of the evolved state at the $j$-th step, i.e., $\bar{\calB}_j^\dagger\bar{\calB}_{j-1}^\dagger\cdots\bar{\calB}_1^\dagger\dket{b}$, as $[t_j(b), c_j(b), A_j]$. The evolution of $A_j$ is independent of $b$, while $c_j(b)$ is affine in $b$, written as $c_j(b) = C_j b \oplus d_j$ with $C_j$ and $d_j$ being a binary matrix and vector of compatible dimensions, respectively, which can be determined algorithmically following the Gottesman-Knill theorem~\cite{Aaronson2004}. As for the evolution of $t_j(b)$, it can become $0$ if the input stabilizer state is orthogonal to the projector, and the subset $\mathbb{S}\subseteq \{0, 1\}^n$ of $b$ for which $t_j(b)\neq 0$ is an affine subspace. As an example, we consider the case that $\bar{\calB}_j$ is the projector $\dketbra{0}{0}$ on the $r$-th qubit. In this case, $t_j(b) = 0$ if $-Z_r$ is a stabilizer of the state evolved by previous $j-1$ steps. Given that $\pm Z_r$ is a stabilizer of the state evolved by previous $j-1$ steps and the projector is $\dketbra{0}{0}$, the affine subspace $\mathbb{S}\subseteq \{0, 1\}^n$ is determined by 
\begin{equation}
  \mathbf{1}^{(r)} \cdot (D_j b \oplus h_j) = 0,
  \label{eq: projector constraint}
\end{equation}
where $\mathbf{1}^{(r)}$ is the vector whose $r$-th element is $1$ and the other elements are $0$, and $D_j$ and $h_j$ are a fixed binary matrix and vector obtained from the affine sign map $c_j(b)=C_j b\oplus d_j$. The constraint extracts the sign of the $Z_r$ stabilizer and requires it to be $+$, so that the deterministic projection preserves the state rather than annihilating it.

{
We now detail the efficient $\bbZ$-rows comparison algorithm, which verifies if Eq.~\eqref{eq: Z-blocck same} holds for all $b \in \{0, 1\}^n$. 
The final states $\bar{\calB}^\dagger \dket{b}$ and $\bar{\calB}^{\prime \dagger} \dket{b}$ are represented by $[t_{N+1}(b), c_{N+1}(b), A_{N+1}]$ and $[t'_{N+1}(b), c'_{N+1}(b), A'_{N+1}]$, respectively. For any fixed $b$, if $t_{N+1}(b) = t'_{N+1}(b) = 0$, both states vanish and Eq.~\eqref{eq: Z-blocck same} is true.  If only one vanishes, Eq.~\eqref{eq: Z-blocck same} is not true. If both are non-zero, Eq.~\eqref{eq: Z-blocck same} is true if and only if their stabilizer generators coincide. We emphasize that this test involves only the stabilizer generators, i.e., $A$ and $c$: the two $\bbZ$-rows are treated as identical once their generators coincide, even if the traces $t_{N+1}(b)$ and $t'_{N+1}(b)$ (in the case both are non-zero) are not the same. Overall, the certification of Eq.~\eqref{eq: Z-blocck same} follows these steps:

\begin{enumerate}
    \item \textit{Comparison of $A$.---}The binary matrices $A_{N+1}$ and $A'_{N+1}$ are independent of $b$. We simulate their evolution (costing at most $O(Nn^3)$ bit operations) and verify if $A_{N+1} = A'_{N+1}$. If not, the $\bbZ$-rows differ. 

    \item \textit{Comparison of supports.---}If the test of $A_{N+1} = A'_{N+1}$ passes, we then check whether the affine subspaces $\mathbb{S} = \{b : t_{N+1}(b) \neq 0\}$ and $\mathbb{S}' = \{b : t'_{N+1}(b) \neq 0\}$ are the same. Each deterministic projector in $\bar{\calB}^\dagger$ contributes one linear constraint on $b$, e.g., Eq.~\eqref{eq: projector constraint}, stacking into a linear system that constrains $b$ into the subspace $\mathbb{S}$. Solving the linear system via Gaussian elimination yields a base point $s_0$ and a linear basis $\{v_1, \dots, v_d\}$ for the subspace, providing an affine basis $\mathbb{T} = \{s_0, s_0 \oplus v_1, \dots, s_0 \oplus v_d\}$ with $d \le n$. Similarly, we obtain the affine basis $\mathbb{T}^\prime = \{s^\prime_0, s^\prime_0 \oplus v^\prime_1, \dots, s^\prime_0 \oplus v^\prime_{d^\prime}\}$ with $d^\prime \le n$. Since Gaussian elimination yields a unique canonical form, $\mathbb{S} = \mathbb{S}^\prime$ if and only if $\mathbb{T} = \mathbb{T}^\prime$. If they differ, the $\bbZ$-rows differ.

    \item \textit{Comparison of $c$.---}If the supports coincide, $\mathbb{S} = \mathbb{S}'$, it remains to check whether $c_{N+1}(b) = c'_{N+1}(b)$ holds for all $b \in \mathbb{S}$; the values at $b \notin \mathbb{S}$ are irrelevant, since both states vanish there. Since $c_{N+1}(b)$ and $c'_{N+1}(b)$ are affine in $b$, their restriction to the affine subspace $\mathbb{S}$ is fixed by their values on the affine basis $\mathbb{T}$, so it suffices to verify $c_{N+1}(b) = c'_{N+1}(b)$ at the $d+1 \le n+1$ points $b \in \mathbb{T}$. If they all agree, the $\bbZ$-rows coincide; otherwise, they differ.
\end{enumerate}

The overall computational complexity of this algorithm is bounded by $O(Nn^4)$ bit operations, where $O(Nn^3)$ is due to the simulation of $\bar{\calB}^\dagger$ which contains $O(N)$ Clifford operations, and the factor of $n$ is because the Clifford circuit simulation is repeated for every point of the affine basis.
}

\textit{Acknowledgements---}We thank Qingyue Zhang, Zhou You, and Feng Xu for useful discussions.  D.Q.$\&$Y.Z. acknowledge the support from the Innovation Program for Quantum Science and Technology Grant Nos.~2024ZD0301900 and 2021ZD0302000, the National Natural Science Foundation of China (NSFC) Grant Nos.~12205048 and 12575012, the Shanghai Science and Technology Innovation Action Plan Grant No.~24LZ1400200, Shanghai Pilot Program for Basic Research - Fudan University 21TQ1400100 (25TQ003), and the start-up funding of Fudan University.
Y.L. is supported by the National Natural Science Foundation of China (Grant Nos. 12225507, 12088101) and NSAF (Grant No. U1930403).

\section{Discussion}

This work establishes classical noise inversion (CNI) as a potent and versatile framework for quantum error mitigation (QEM). By fundamentally shifting the computational overhead from quantum to classical resources, CNI directly addresses two critical bottlenecks in the field: the prohibitive circuit-sampling cost and the reliance on unrealistic, gate-independent noise models. 
We then develop a rigorous framework characterizing the classical tractability of quantum noise. {To address scenarios where classically intractable noise persists, we introduce partial CNI, which leverages a refined probabilistic error cancellation technique to suppresses the classically intractable gate-dependent noise, at a circuit-sampling cost controlled solely by the classically intractable noise component, rather than the full noise model.}

Furthermore, we propose noise compression to enhance the efficiency of practical quantum error mitigation. This method achieves the minimal sampling overhead by treating distinct noise components with identical effects as indistinguishable. The development of CNI-based robust shadow estimation provides a practical and superior protocol for reliable quantum state learning on noisy devices, as confirmed by our theoretical and numerical analyses. 

{
  The proposed protocols in this paper are applicable to what we term error-propagatable circuits. Classically simulating the computational outputs of such circuits incurs exponential overhead with increasing qubit number and circuit depth, whereas simulating error propagation throughout such circuits scales polynomially with these two parameters. Many circuits with critical applications are error-propagatable, including those used in Clifford+T fault-tolerant computation, measurement-based quantum computing, Hamiltonian evaluation via simultaneous estimation of commuting observables, and a wide range of randomized measurement protocols for characterizing quantum systems. We therefore anticipate CNI to become a valuable tool across near-term and fault-tolerant quantum applications, ranging from quantum benchmarking and variational algorithms to quantum simulations. 
}

Looking forward, several research avenues promise to further amplify the impact of the CNI paradigm. First, while our implementation employed unbiased Monte Carlo methods, future work could leverage more powerful classical simulators, such as tensor networks, to model the noise inversion deterministically. This could enhance accuracy, suppress statistical variance, and improve the scalability of the protocol. Second, the noise compression technique introduced here is a standalone innovation. Generalizing its principles to non-Clifford and deeper quantum circuits presents a compelling direction for reducing the overhead of quantum error mitigation at large. Third, it is intriguing to extend the framework and learning protocols introduced here to continuous systems, such as bosonic systems~\cite{mele2024CV}, and to their qudit analogs~\cite{Mao2025}. In addition, nonlinear observables are also worthy of study~\cite{liu2024auxiliary, zhang2025measuring}, as our framework may mitigate the noise in the joint measurements required for their estimation. 

Beyond its immediate technical contributions, this work underscores a pivotal conceptual shift: the strategic use of classical computational resources to manage and invert quantum errors. As quantum processors continue to scale, we envision that such co-designed, classical-quantum solutions will be indispensable for unlocking the full potential of quantum technologies, bridging the gap between the noisy devices of today and the fault-tolerant computers of the future.

\bibliography{CNI-arxiv-20260806.bbl}

\clearpage
\appendix
\onecolumngrid

\section{Performance of classical noise inversion}
\label{app: variance of error-mitigated estimators}
In this section, we analyze the performance of classical noise inversion (CNI) and compare it with that of conventional probabilistic error cancellation. We prove Theorem~\ref{th: perf CNI} (establishing CNI's performance guarantee) in Sec.~\ref{app: proof of variance of CNI} and derive the variance bound for conventional probabilistic error cancellation (cPEC) in Sec.~\ref{app: variance of conventional error mitigation } to contextualize CNI's advantage.

We restate the setup here. The noisy quantum circuit (including state preparation and measurement) is $\calM_{\bbZ}\tilde{\calU} \dket{\rho}$ where $\dket{\rho}$ is the input state, $\tilde{\calU}$ represents the operation of the noisy unitary, and $\calM_{\bbZ}$ is the channel of measurement in the computational basis. We denote $\tilde{\calU} = \calE\calU$ where $\calU$ is the ideal unitary and $\calE = \tilde{\calU}\calU^\dagger$ represents the overall effect of both gate noise and measurement noise. The expectation value of a linear function (e.g., $F=\sum_b f(b)\ketbra{b}{b}$) is estimated using the measurement outcomes. The ideal expectation value of $f$ is
\begin{equation}
f = \dbra{F} \calM_{\bbZ} \calU \dket{\rho}. 
\end{equation}

\subsection{Performance of CNI}
\label{app: proof of variance of CNI}
We prove Theorem~\ref{th: perf CNI} in this subsection. In CNI, we implement the inverse of the compressed error by classically sampling $\calB$ according to 
\begin{equation}
  \calE^{\prime -1} = \sum_{\calB} Q_\calB^\prime \calB^\prime = \gamma^\prime \sum_{\calB} \Prob{\calB} \operatorname{sgn}(\calB) \calB, 
\end{equation}
where $\gamma^\prime = \sum_{\calB}|Q_{\calB}^\prime|$, $\Prob{\calB}^\prime = |Q_{\calB}^\prime|/\gamma^\prime$ and $\operatorname{sgn}(\calB) = Q_{\calB}^\prime/ |Q_{\calB}^\prime|$. 
With one-shot measurement outcome $b$, and one instance of $\calB$, the estimator of $f$ is
\begin{equation}
    \hat{f}_{b, \calB} = \gamma'\operatorname{sgn}(\calB) \dbra{F}\calB\dket{b}. 
\end{equation}

Suppose we measure the noisy circuit for $K$ shots and generate $L$ instances of $\calB$ for each shot. The total estimator of $f$ using CNI is
\begin{equation} 
  \hat{f}_{\operatorname{tot}} = \frac{1}{K} \sum_{k=1}^K \left(\frac{1}{L} \sum_{l=1}^L \gamma^\prime \operatorname{sgn}(\calB_{k, l}) \dbra{F}\calB_{k, l}\dket{b_k}\right),
\end{equation} 
where $b_k$ is the measurement outcome of the $k$-th shot, and $\calB_{k, l}$ is the $l$-th instance of $\calB$ for the $k$-th shot.
$\hat{f}_{\operatorname{tot}}$ is an unbiased estimator of $f$:
\begin{eqnarray}
  \bbE[\hat{f}_{\operatorname{tot}}] &=& \bbE_{\calB, b} [\gamma^\prime \operatorname{sgn}(\calB) \dbra{F}\calB\dket{b}] \nonumber\\
  &=& \sum_{\calB} \sum_{b} \Prob{\calB} \gamma^\prime \operatorname{sgn}(\calB) \dbra{F}\calB\dket{b}\dbra{b}\tilde{\calU}\dket{\rho} \nonumber\\
  &=& \sum_{b} \dbra{F} \calE^{\prime -1}\dket{b} \dbra{b}\tilde{\calU} \dket{\rho} \nonumber\\
  &=& \dbra{F}  \calM_{\bbZ} \calU \dket{\rho} = f,
\end{eqnarray} where we have used $\calM_{\bbZ} = \sum_{b} \dketbra{b}{b}$, $\tilde{\calU} = \calE \calU$ and $\calE^\prime \calM_{\bbZ} = \calM_{\bbZ} \calE$. The variance of the estimator is 
  \begin{eqnarray}
    \Var{\hat{f}_{\operatorname{tot}}} &=& \frac{1}{K}\operatorname{Var}\left[\frac{1}{L} \sum_{l=1}^L \gamma^\prime \operatorname{sgn}(\calB_l) \dbra{F}\calB_l\dket{b}\right] \nonumber\\
    &\leq& \frac{1}{K}\bbE_{b, \{\calB_l\}}\left[\left(\frac{1}{L} \sum_{l=1}^L \gamma^\prime \operatorname{sgn}(\calB_l) \dbra{F}\calB_l\dket{b}\right)^2\right] \nonumber\\
    &=& \frac{1}{K} \bbE_{b, \{\calB_l\}}\left[\frac{\gamma^{\prime 2}}{L^2} \sum_{l=1}^L \dbra{F}\calB_l\dket{b}^2 + \frac{\gamma^{\prime 2}}{L^2}\sum_{l\neq l^\prime } \operatorname{sgn}(\calB_l)\dbra{F}\calB_l\dket{b} \operatorname{sgn}(\calB_{l^\prime}) \dbra{F}\calB_{l^\prime}\dket{b}\right] \nonumber\\
    &=& \frac{\gamma^{\prime 2}}{K} \left(\frac{1}{L} \bbE_{b, \calB}[\dbra{F}\calB\dket{b}^2] + (1 - \frac{1}{L})\bbE_{b, \calB, \calB^\prime}[(\operatorname{sgn}(\calB)\dbra{F}\calB\dket{b})(\operatorname{sgn}(\calB^\prime)\dbra{F}\calB^\prime\dket{b})]\right) \nonumber\\
    &=& \frac{\gamma^{\prime 2}}{K} \left(\frac{1}{L} \bbE_{b, \calB}[\dbra{F}\calB\dket{b}^2] + (1 - \frac{1}{L})\bbE_{b}\left[ \left(\bbE_{\calB}\left[\operatorname{sgn}(\calB)\dbra{F}\calB\dket{b}\right]\right)^2 \right]\right).
    \label{eq: var F}
  \end{eqnarray}
Define the following semi-norms regarding the expectations in the above equation:
\begin{eqnarray}
     \norm{F}_{\star} = \max_{\rho} \left(\bbE_{b, \calB}[\dbra{F}\calB\dket{b}^2]\right)^{1/2} &=& \max_{\rho} \left(\sum_{b}\dbra{b}\calE \calU \dket{\rho} \sum_{\calB} \Prob{\calB} \dbra{F}\calB\dket{b}^2\right)^{1/2}, \label{eq: norm F star}\\
      \norm{F}_{\circ} = \max_{\rho} \left(\bbE_{b}\left[ \left(\bbE_{\calB}\left[\operatorname{sgn}(\calB)\dbra{F}\calB\dket{b}\right]\right)^2 \right]\right)^{1/2}
    &=& \max_{\rho} \left(\sum_{b} \dbra{b} \calE\calU\dket{\rho} \left(\sum_{\calB}\Prob{\calB} \operatorname{sgn}(\calB) \dbra{F}\calB\dket{b}\right)^2\right)^{1/2} \label{eq: norm F o}   
\end{eqnarray} which are semi-norms rather than norms, as they do not guarantee the definiteness condition that $\norm{F} = 0$ if and only if $F=0$. The semi-norms satisfy $\norm{F}_{\circ}\leq \norm{F}_{\star}\leq \norm{F}_{\infty}$, where the infinity norm $\norm{\bullet}_{\infty}$ corresponds to the largest absolute singular value. Substituting Eq.~\eqref{eq: norm F star} and Eq.~\eqref{eq: norm F o} into Eq.~\eqref{eq: var F}, we obtain Eq.~\eqref{eq: variance of CNI} in Theorem~\ref{th: perf CNI}.

\subsection{Performance of conventional error mitigation }
\label{app: variance of conventional error mitigation }
In conventional quantum error mitigation, specifically probabilistic error cancellation, we randomly generate $\calB$ according to the quasi-probability decomposition of $\calE^{-1}$:
\begin{equation}
  \calE^{-1} = \sum_{\calB} Q_\calB \calB = \gamma \sum_{\calB} \Prob{\calB} \operatorname{sgn}(\calB) \calB, \label{eq: inv CEM}
\end{equation} where $\gamma = \sum_{\calB}|Q_{\calB}|$, $\Prob{\calB} = |Q_{\calB}|/\gamma$ and $\operatorname{sgn}(\calB) = Q_{\calB}/ |Q_{\calB}|$, and execute the noisy circuit $\calB\tilde{\calU}\dket{\rho}$. The estimator of $f$ given by conventional probabilistic error cancellation is $\gamma \operatorname{sgn}(\calB)\dbraket{F}{b}$. Suppose we generate $M$ random noisy circuits, and each circuit is measured $K$ times. The estimator of $f$ using conventional error mitigation is
\begin{equation}
  \hat{f}_{\operatorname{CEM}} = \frac{1}{M} \sum_{m=1}^M \gamma \operatorname{sgn}(\calB_m)\left(\frac{1}{K}\sum_{k=1}^K \dbraket{F}{b_{m ,k}}\right), 
\end{equation} where $b_{m ,k}$ is the outcome of the $k$-th shot of the $m$-th random circuit. $\hat{f}_{\operatorname{CEM}}$ is an unbiased estimator of $f$:
\begin{eqnarray}
  \bbE[\hat{f}_{\operatorname{CEM}}] &=& \bbE [\gamma \operatorname{sgn}(\calB) \dbraket{F}{b}] \nonumber\\
  &=& \sum_{\calB} \sum_{b} \Prob{\calB} \gamma \operatorname{sgn}(\calB) \dbra{b}\calB \calE {\calU}\dket{\rho} \dbraket{F}{b} \nonumber\\
  &=& \sum_{b} \dbraket{F}{b} \dbra{b} {\calU}\dket{\rho} = f. 
\end{eqnarray} 
The variance of $\hat{f}_{\operatorname{CEM}}$ is 
\begin{eqnarray}
\text{Var} [\hat{f}_{\operatorname{CEM}}] &=& \frac{1}{M} \operatorname{Var} \left[ \frac{1}{K} \sum_{k=1}^K \gamma \operatorname{sgn}(\calB) \dbraket{F}{b_k} \right] \nonumber\\
  &\leq& \frac{1}{M}\bbE_{\calB, \{b_k\}}\left[\left(\frac{1}{K} \sum_{k=1}^K \gamma \operatorname{sgn}(\calB) \dbraket{F}{b_k}\right)^2\right] \nonumber\\
  &=& \frac{\gamma^2}{M} \bbE_{\calB, \{b_k\}}\left[\frac{1}{K^2} \sum_{k=1}^K \dbraket{F}{b_k}^2 + \frac{1}{K^2}\sum_{k\neq k^\prime }^K \dbraket{F}{b_k}\dbraket{F}{b_{k^\prime}}\right] \nonumber\\
  &=& \frac{\gamma^2}{M} \left(\frac{1}{K} 
  \bbE_{\calB, b}[\dbraket{F}{b}^2]  
  + (1 - \frac{1}{K})
  \bbE_{\calB, b, b^\prime}[\dbraket{F}{b}\dbraket{F}{b^\prime}]
  \right). 
\end{eqnarray} 

The expectations in the above equation are
\begin{eqnarray}
\bbE_{\calB, b}[\dbraket{F}{b}^2] & =& \sum_{\calB} \Prob{\calB} \sum_{b} \dbra{b} \calB \tilde{\calU} \dket{\rho} \dbraket{F}{b}^2, \\
\bbE_{\calB, b, b^\prime}[\dbraket{F}{b}\dbraket{F}{b^\prime}] & =& 
\sum_{\calB} \Prob{\calB} \left(\sum_{b} \dbra{b} \calB \tilde{\calU} \dket{\rho} \dbraket{F}{b}\right)^2. 
\end{eqnarray}

\section{Channel propagation}
\label{app: channel propagation}
In this section, we discuss the error propagation through the quantum gates and the measurement. Any channel can propagate through a unitary gate, that is, there always exists a unique $\BcalE$ such that $\BcalE \calU = \calU \calE$ for any channel $\calE$ and unitary gate $\calU$. The propagated channel is $\BcalE = \calU \calE \calU^\dagger$. However, not every channel can propagate through the measurement. The following theorem proves the necessary and sufficient condition for propagating a quantum channel through the measurement, and provides the propagated channel.

\begin{theorem}
  \label{the: the channel propagation}
  A quantum channel $\calE$ can propagate through the measurement if and only if 
  \begin{equation}
    \dbra{\sigma}\calE \dket{\sigma^\prime} = 0, \; \forall\; \sigma \in \bbZ, \sigma^\prime \notin \bbZ. \label{eq: can propagated through measurement}
  \end{equation}The propagated channel can be any channel that admits the following Pauli transfer matrix:
  \begin{equation}
    \calEeff = \begin{pmatrix}
      \calE_{00} & A \\
      0 & B
      \end{pmatrix}, \label{eq: propagate channel m}
  \end{equation} where $\calE_{00}$ is the $\bbZ$-rows of $\calE$, $A, B$ are arbitrary with compatible dimensions.
\end{theorem}
\begin{proof}
  A quantum channel $\calE$ can propagate through the measurement means there exists $\calE^\prime$ such that $\calE^\prime \calM_{\bbZ} = \calM_{\bbZ} \calE,$ where 
  \begin{equation}
    \calM_{\bbZ} = \sum_{\sigma\in \bbZ} \dketbra{\sigma}{\sigma} = 
    \begin{pmatrix}
      \mathbb{I} & 0 \\
      0 & 0
      \end{pmatrix}.
  \end{equation}
  Notice that 
  \begin{eqnarray}
    &&\calEeff \MZ = \begin{pmatrix}
      \calEeff_{00} & \calEeff_{01} \\
      \calEeff_{10} & \calEeff_{11}
    \end{pmatrix}
    \begin{pmatrix}
      \id & 0\\
      0 & 0
    \end{pmatrix} = \begin{pmatrix}
      \calEeff_{00} & 0 \\
      \calEeff_{10} & 0
    \end{pmatrix} \\
    &&\MZ \calE = 
    \begin{pmatrix}
      \id & 0\\
      0 & 0
    \end{pmatrix}
    \begin{pmatrix}
      \calE_{00} & \calE_{01} \\
      \calE_{10} & \calE_{11}
    \end{pmatrix}
    = \begin{pmatrix}
      \calE_{00} & \calE_{01} \\
      0 & 0
    \end{pmatrix},
  \end{eqnarray}
the necessary and sufficient condition for $\calE^\prime \calM_{\bbZ} = \calM_{\bbZ} \calE$ is: $\calE_{01} = 0,\; \calE^{\prime}_{00} = \calE_{00},\; \calE^\prime_{10} = 0$. The first condition $\calE_{01} = 0$ is an alternative expression for Eq.~\eqref{eq: can propagated through measurement}. The second and third conditions are satisfied by any quantum channel whose PTM admits the form Eq.~\eqref{eq: propagate channel m}. Thus, the theorem is proved.
\end{proof}

\section{Performance of CNI-based robust shadow estimation}

We begin this section by reviewing the variance bounds of classical shadow tomography in noiseless settings~\cite{Huang2020a,Zhou2023}. Subsequently, we establish the performance guarantees of our robust shadow protocol with CNI in Appendix~\ref{app: proof of robust shadow with CNI} (i.e., we prove Theorem~\ref{th:perf_robust_shadow}), and compare the noise shadow norms with the ideal shadow norms in Appendix~\ref{app: Comparison of noisy shadow norms and ideal shadow norm}. For comparative analysis, Appendix~\ref{app: variance of robust shadow with cPEC} derives the variance of robust shadow tomography using conventional error mitigation. 

The shadow estimation protocol applies a random unitary $\calU$ channel from some ensemble $\bbU$ to the input state, followed by computational basis measurement. Subsequently, during post-processing, the inverse $\calM^{-1}$ of the measurement channel $\calM = \bbE_{U\in\bbU}\calU^\dagger\calM_{\bbZ}\calU$ is applied to the sampled shadow. Note that $\calM_{\bbZ} = \sum_{b} \dketbra{b}{b}$.

The multi-shot classical shadow estimator for  $o = \dbraket{O}{\rho}$ is given by 
\begin{equation}
  \hat{o}_{\operatorname{tot}} = \frac{1}{MK} \sum_{m=1}^M \sum_{k=1}^K \dbra{O} \calM^{-1}\calU_{m}^\dagger\dket{b_{m, k}},
\end{equation}
where $\{\calU_{m}|m=1,\ldots,M\}$ are $M$ i.i.d. random unitaries and $\{b_{m, k}|k=1,\ldots,K\}$ are the outputs of measuring the $m$-th circuit $\calU_{m}\dket{\rho}$ for total $K$ shots. $\hat{o}_{\operatorname{tot}}$ is an unbiased estimator for $\dbraket{O}{\rho}$:
\begin{eqnarray}
  \bbE\left[\hat{o}_{\operatorname{tot}}\right] &=& \bbE_{\calU, b} \dbra{O}  \calM^{-1} \calU^\dagger \dket{b} \nonumber\\
  &=& \bbE_{\calU}\sum_{b}  \dbra{O}  \calM^{-1} \calU^\dagger \dket{b} \dbra{b} \calU \dket{\rho} \nonumber\\
  &=& \dbra{O}  \calM^{-1} \calM \dket{\rho} = \dbraket{O}{\rho}, 
\end{eqnarray} where we have used the definitions $\calM_{\bbZ} = \sum_{b}\dketbra{b}{b}$ and $\calM = \bbE[\calU^\dagger \calM_{\bbZ} \calU]$. 
The variance of $\hat{o}_{\operatorname{tot}}$ is upper bounded by 
\begin{equation}
  \Var{\hat{o}_{\operatorname{tot}}} \leq \frac{1}{M}\left[\frac{1}{K}\left\|O\right\|_{\text {S }}^2+\left(1-\frac{1}{K}\right)\left\|O\right\|_{\text {XS }}^2\right],
  \end{equation} where the shadow norm $\norm{O}_{\text{S}}$~\cite{Huang2020a} and the cross-shadow norm $\norm{O}_{\text{XS}}$~\cite{Zhou2023} are defined as
\begin{eqnarray}
  \norm{O}_{\text{S} } &=& \max_{\rho}\left( \bbE_{\calU}\sum_{b} \dbra{\calM^{-1}(O)}\calU^\dagger \dket{b}^2 \dbra{\rho} \calU^\dagger \dket{b}\right)^{1/2}, \label{eq: shadow norm}\\
  \norm{O}_{\text{XS} } &=& \max_{\rho} \left(\bbE_{\calU} \sum_{b, b^\prime} \dbra{O}\calM^{-1}\calU^\dagger \dket{b}\dbra{O}\calM^{-1}\calU^\dagger \dket{b^\prime} \dbra{b} \calU \dket{\rho} \dbra{b^\prime} \calU \dket{\rho}\right)^{1/2}. \label{eq: cross-shadow norm}
\end{eqnarray}

\subsection{CNI-based robust shadow estimation}
\label{app: proof of robust shadow with CNI}

We prove Theorem~\ref{th:perf_robust_shadow} in this subsection.
In our method, we randomly generate quantum channel $\calU$ from an ensemble $\bbU$, implement the noisy quantum circuit $\tilde{\calU}\dket{\rho}$ on a noisy quantum computer and measure it in the computational basis to record a measurement outcome $b$. All the mitigation and reconstruction are left in the post-processing stage, which is the core difference compared to conventional error mitigation shown in the last section. In the post-processing step, we randomly generate a basis channel $\calB$ according to 
\begin{equation}
  \calE^{\prime -1}_{\calU} = \gamma^\prime_{_\calU} \sum_{\calB} \Prob{\calB\vert\calU} \operatorname{sgn}(\calB) \calB,  
\end{equation} 
where $\calE^{\prime -1}$ is the compressed noise inversion and satisfies $\calE^{\prime -1}_{\calU} \calM_{\bbZ}\calE_{\calU} =  \calM_{\bbZ}$, then we classically compute 
\begin{equation}
  \hat{o} = \gamma^\prime \operatorname{sgn}(\calB)\dbra{F}\calM^{-1}\calU^\dagger \calB \dket{b}.
\end{equation}
Suppose we have generated $M$ noisy random circuits $\{\calE_m \calU_m\}$, and we have measured $\calE_m \calU_m\dket{\rho}$ for $K$ times. The key difference between our method and the previous method is that we mitigate the error in the post-processing step, which allows us to leverage classical resources to reduce the variance of the estimator. Suppose for each random circuit $\calE_m \calU_m$ and each measurement outcome $b_{m, k}$, we have generated $L$ random basis channels $\{\calB_{m,k,l}\vert l=1,2,...,L \}$. The total estimator is 
  \begin{equation}\label{eqapp:oRshadow}
    \hat{o}_{\operatorname{tot}} = \frac{1}{M} \sum_{m=1}^M \frac{1}{K} \sum_{k=1}^{K} \frac{1}{L}\sum_{l=1}^{L} \gamma^\prime_{m} \operatorname{sgn}(\calB_{m,k,l})  \dbra{O}\calM^{-1}\calU_m^\dagger \calB_{m,k,l} \dket{b_{m,k}},
  \end{equation} 
which is an unbiased estimator:
\begin{eqnarray}
  \bbE[\hat{o}_{\operatorname{tot}}] &=& \bbE \left[ \gamma^\prime_{_\calU} \operatorname{sgn}(\calB)  \dbra{O}\calM^{-1}\calU^\dagger \calB \dket{b} \right] \nonumber\\
  &=& \bbE_{\calU} \sum_{b} \sum_{\calB} \Prob{{\calB|\calU}} \gamma^\prime_{_\calU}  \operatorname{sgn}(\calB) \dbra{O}\calM^{-1}\calU^\dagger \calB \dket{b} \dbra{b} \tilde{\calU}\dket{\rho} \nonumber \\
  &=& \bbE_{\calU} \sum_{b} \dbra{O}\calM^{-1}\calU^\dagger \calE^{\prime -1}_{\calU} \dket{b} \dbra{b} \calE_{\calU}\calU\dket{\rho} = \dbraket{O}{\rho}. 
\end{eqnarray}
 Next, we derive the variance of the estimator. For simplicity, we place the summation over $l$ inside the expectation and denote  
\begin{equation}
  \hat{\calE}^{\prime - 1}_{\calU_m, b_{m,k}} = \frac{1}{L} \sum_{l=1}^{L} \gamma^\prime_{m} \operatorname{sgn}(\calB_{m,k,l}) \calB_{m,k,l}, 
\end{equation} which allows us to write the estimator in \eqref{eqapp:oRshadow} as  
\begin{equation}
  \hat{o}_{\operatorname{tot}} = \frac{1}{M} \sum_{m=1}^M \frac{1}{K} \sum_{k=1}^{K}  \dbra{O}\calM^{-1}\calU_m^\dagger \hat{\calE}^{\prime - 1}_{\calU_m, b_{m,k}} \dket{b_{m,k}}. 
\end{equation} The variance of the estimator is 
  \begin{eqnarray}
  \text{Var}[\hat{o}_{\operatorname{tot}}] &=& \frac{1}{M} \text{Var}\left[\frac{1}{K} \sum_{k=1}^{K}  \dbra{O}\calM^{-1}\calU^\dagger \hat{\calE}^{\prime - 1}_{\calU,b_k} \dket{b_{k}}\right] \label{eq: var cni shadow 1}\\ 
    &\leq& \frac{1}{M} \left( \underset{\calU, \{b_k\}, \{\calB_{k,l}\}}{\bbE} \left[ \frac{1}{K^2} \sum_{k=1}^{K}  \dbra{O}\calM^{-1}\calU^\dagger \hat{\calE}^{\prime - 1}_{\calU, b_k} \dket{b_{k}}^2 \right]+ \right. \nonumber\\
    &&\left.\underset{\calU, \{b_k\}, \{\calB_{k,l}\}}{\bbE}\left[ \frac{1}{K^2} \sum_{k\neq k^\prime} \dbra{O}\calM^{-1}\calU^\dagger \hat{\calE}^{\prime - 1}_{\calU, b_k} \dket{b_{k}} \dbra{O}\calM^{-1}\calU^\dagger \hat{\calE}^{\prime - 1}_{\calU, b_{k^\prime}} \dket{b_{k^\prime}} \right]\right) \label{eq: var cni shadow 2}\\
    &=& \frac{1}{M} \left(\frac{1}{K} \left(\underset{\calU, b, \{\calB_l\}}{\bbE}\left[\dbra{O}\calM^{-1}\calU^\dagger \hat{\calE}^{\prime - 1}_{\calU, b} \dket{b}^2\right]\right) + \nonumber \right.\\
    && \left.(1 - \frac{1}{K}) 
    \underset{\calU, b, b^\prime, \{\calB_l\},\{\calB_{l^\prime}\}}{\bbE}\left[\dbra{O}\calM^{-1}\calU^\dagger \hat{\calE}^{\prime - 1}_{\calU, b} \dket{b}\dbra{O}\calM^{-1}\calU^\dagger \hat{\calE}^{\prime - 1}_{\calU, b^\prime} \dket{b^\prime}\right]\right). \label{eq: var cni shadow 3} 
  \end{eqnarray} 
  Here, we make some clarification about the above derivation. In Eq.~\eqref{eq: var cni shadow 1}, we used that $\{\calU_m\}$ are independently and identically distributed (i.i.d.) samples following a uniform distribution over $\bbU$. Eq.~\eqref{eq: var cni shadow 2} is a direct expansion of Eq.~\eqref{eq: var cni shadow 1}. Here, $\{b_k\}$ denotes a set of independent and identically distributed (i.i.d.) measurement outcomes, which follow the distribution given by $\Prob{b\vert \calU} = \dbra{b}\tilde{\calU}\dket{\rho}$. By leveraging that $\{b_k\}$ are i.i.d., we derive Eq.~\eqref{eq: var cni shadow 3}, where $\hat{\calE}^{\prime - 1}_{\calU, b} = \frac{1}{L} \sum_{l=1}^{L} \gamma^\prime_{\calU} \operatorname{sgn}(\calB_{l}) \calB_{l}$ is the mean of $L$ instance of $\calB$ for the measurement outcome $b$. 
  The first expectation in Eq.~\eqref{eq: var cni shadow 3} reads
  \begin{eqnarray}
    && \bbE_{\calU, b, \{\calB_l\}} \left[\dbra{O}\calM^{-1}\calU^\dagger \hat{\calE}^{\prime - 1}_{\calU, b} \dket{b}^2\right] \nonumber\\
    &&= \bbE_{\calU} \sum_{b} \dbra{b} \tilde{\calU} \dket{\rho} 
    \bbE_{\{\calB_l\}} \left[ \dbra{O}\calM^{-1}\calU^\dagger \left(\frac{1}{L}\sum_{l=1}^{L} \gamma^\prime_{_\calU} \operatorname{sgn}(\calB_{l}) \calB_{l} \right) \dket{b}^2 \right]  \nonumber\\
    &&=  \bbE_{\calU} \sum_{b} \dbra{b} \tilde{\calU} \dket{\rho} 
    \bbE_{\{\calB_l\}} \left[ \frac{\gamma^{\prime 2}_{_\calU}}{L^2} \sum_{l=1}^{L} \dbra{O}\calM^{-1}\calU^\dagger   \calB_{l} \dket{b}^2 \right] \nonumber\\
    && \quad + \bbE_{\calU} \sum_{b} \dbra{b} \tilde{\calU} \dket{\rho} 
    \bbE_{\{\calB_l\}} \left[ \frac{\gamma^{\prime 2}_{_\calU}  }{L^2} \sum_{l\neq l^\prime}^{L} \dbra{O}\calM^{-1}\calU^\dagger \operatorname{sgn}(\calB_{l})  \calB_{l} \dket{b}  \dbra{O}\calM^{-1}\calU^\dagger \operatorname{sgn}(\calB_{l^\prime}) \calB_{l^\prime} \dket{b} \right] \nonumber\\
    &&=  \bbE_{\calU} \sum_{b} \dbra{b} \tilde{\calU} \dket{\rho} 
    \frac{\gamma^{\prime 2}_{_\calU}}{L} \sum_{\calB} \Prob{{\calB|\calU}} \dbra{O}\calM^{-1}\calU^\dagger \calB \dket{b}^2 \nonumber\\
    && \quad +  \bbE_{\calU} \sum_{b} \dbra{b} \tilde{\calU} \dket{\rho} 
    \frac{L(L-1)}{L^2} \sum_{\calB, \calB^\prime}\Prob{{\calB|\calU}}\Prob{{\calB^\prime|\calU}} \gamma^{\prime 2}_{_\calU} \dbra{O}\calM^{-1}\calU^\dagger \operatorname{sgn}(\calB) \calB \dket{b} \dbra{O}\calM^{-1}\calU^\dagger \operatorname{sgn}(\calB^{\prime}) \calB^{\prime} \dket{b} \nonumber\\
    &&=  \bbE_{\calU} \sum_{b} \dbra{b} \tilde{\calU} \dket{\rho} 
    \frac{\gamma^{\prime 2}_{_\calU}}{L} \sum_{\calB} \Prob{{\calB|\calU}} \dbra{O}\calM^{-1}\calU^\dagger \calB \dket{b}^2 \nonumber\\
    && \quad +  \bbE_{\calU} \sum_{b} \dbra{b} \tilde{\calU} \dket{\rho} 
    \left(1 - \frac{1}{L}\right) \gamma^{\prime 2}_{_\calU} \left(\sum_{\calB}\Prob{{\calB|\calU}}  \dbra{O}\calM^{-1}\calU^\dagger \operatorname{sgn}(\calB) \calB \dket{b} \right)^2\label{eq: tttt}\\
    &&\leq \frac{\gamma^{\prime 2}_{\text{max}}}{L}  \bbE_{\calU} \sum_{b}  \dbra{b} \tilde{\calU} \dket{\rho} \sum_{\calB}  \Prob{{\calB|\calU}}  
   \dbra{O}\calM^{-1}\calU^\dagger \calB \dket{b}^2 \nonumber\\
    && \quad + \left(1 - \frac{1}{L}\right)\gamma^{\prime 2}_{\text{max}}  \bbE_{\calU} \sum_{b} \dbra{b} \tilde{\calU} \dket{\rho} 
     \left(\sum_{\calB}\Prob{{\calB|\calU}}  \dbra{O}\calM^{-1}\calU^\dagger \operatorname{sgn}(\calB) \calB \dket{b} \right)^2,
  \end{eqnarray} where $\gamma^\prime_{\operatorname{max}} = \max_{\calU\in \bbU} \gamma^\prime_{\calU}$ is the maximum value of $\gamma^\prime$ over all the quantum channels $\calU$. 
  Define the noisy shadow norms $\norm{O}_{\text{NS1}}$ and $\norm{O}_{\text{NS2}}$ as
  \begin{eqnarray}
     \norm{O}_{\text{NS1}} &=& \max_{\rho} \left(\bbE_{\calU} \sum_{b} 
     \dbra{b} \tilde{\calU} \dket{\rho}   \sum_{\calB}  \Prob{{\calB|\calU}} 
    \dbra{O}\calM^{-1}\calU^\dagger \calB \dket{b}^2 \label{eq: NS1}\right)^{1/2}\\
    \norm{O}_{\text{NS2}} &=& \max_{\rho} 
      \left(\bbE_{\calU} \sum_{b} \dbra{b} \tilde{\calU} \dket{\rho} 
     \left(\sum_{\calB}\Prob{{\calB|\calU}}  \dbra{O}\calM^{-1}\calU^\dagger \operatorname{sgn}(\calB) \calB \dket{b} \right)^2\right)^{1/2}.  
     \label{eq: NS2} 
  \end{eqnarray} 
  We can show that $\norm{O}_{\text{NS1}} \geq \norm{O}_{\text{NS2}}$ for any $O$, because 
\begin{equation}
  \norm{O}_{\text{NS1}}^2 - \norm{O}_{\text{NS2}}^2 = \bbE_{\calU, b} \left(\operatorname{Var}_{\calB\sim\Prob{\calB|\calU}}\left[ \dbra{O}\calM^{-1}\calU^\dagger \operatorname{sgn}(\calB) \calB \dket{b}\right] \right) \geq 0. 
\end{equation}
  As a result, the first expectation in Eq.~\eqref{eq: var cni shadow 3} is simplified to  
  \begin{eqnarray}
    \bbE\left[\dbra{O}\calM^{-1}\calU^\dagger \hat{\calE}^{\prime - 1}_{\calU, b} \dket{b}^2\right] \leq \gamma^{\prime 2}_{\text{max}} \left(
    \frac{1}{L} \norm{O}_{\text{NS1}}^2 + \left(1 - \frac{1}{L}\right) \norm{O}_{\text{NS2}}^2 \right). 
  \end{eqnarray}
Next, we consider the second expectation term in~Eq.~\eqref{eq: var cni shadow 3}, which is 
\begin{eqnarray}
  && \bbE_{\calU, b, b^\prime, \{\calB_l\}, \{\calB_{l^\prime}\}}\left[\dbra{O}\calM^{-1}\calU^\dagger \hat{\calE}^{\prime - 1}_{\calU, b} \dket{b}\dbra{O}\calM^{-1}\calU^\dagger \hat{\calE}^{\prime - 1}_{\calU, b^\prime} \dket{b^\prime}\right]  \nonumber\\
  &=& \bbE_{\calU} \sum_{b, b^\prime} \dbra{O}\calM^{-1}\calU^\dagger {\calE}^{-1} \dket{b}\dbra{O}\calM^{-1}\calU^\dagger {\calE}^{-1} \dket{b^\prime} \dbra{b} \calE \calU \dket{\rho} \dbra{b^\prime} \calE \calU \dket{\rho}  \nonumber\\
  &=& \bbE_{\calU} \sum_{b, b^\prime} \dbra{O}\calM^{-1}\calU^\dagger \dket{b}\dbra{O}\calM^{-1}\calU^\dagger \dket{b^\prime} \dbra{b} \calU \dket{\rho} \dbra{b^\prime} \calU \dket{\rho}, \nonumber\\
  &\leq& \norm{O}_{\text{XS}}^2. \label{eq: var cni shadow 4}
\end{eqnarray}
To sum it up, the variance is upper bounded by 
\begin{eqnarray}
  \text{Var}[\hat{o}_{\text{tot}}] &\leq& \frac{1}{M}\left[
    \frac{\gamma^{\prime 2}_{\text{max}}}{K}\left(\frac{1}{L} \norm{O}_{\text{NS1}}^2 + \left(1 - \frac{1}{L}\right) \norm{O}_{\text{NS2}}^2\right) 
    + (1 - \frac{1}{K}) \norm{O}_{\text{XS}}^2
  \right]. 
\end{eqnarray}

\subsection{Comparison of noisy shadow norms and ideal shadow norm}
\label{app: Comparison of noisy shadow norms and ideal shadow norm}

Here, we prove Proposition~\ref{prop:NS_bound}. 
As established by Eq.~\eqref{eq: var cni shadow 4}, CNI recovers the ideal cross-shadow norm, demonstrating its theoretical advantage in preserving measurement statistics under noise. We now compare the noisy shadow norms $\norm{\bullet}_{\text{NS1}}$ and $\norm{\bullet}_{\text{NS2}}$ against the ideal shadow norm $\norm{\bullet}_{\text{S}}$.

Next, we focus on the upper bound of $\norm{O}_{\text{NS1}}$.
Since we are always considering gate-dependent noise, it does not matter whether we write $\calB$ on the left or right side of $\calU$. Thus, we can rewrite $\norm{O}_{\text{NS1}}$ in Eq.~\eqref{eq: NS1} as
\begin{eqnarray}
  \norm{O}_{\text{NS1}}^2 &=& 
  \bbE_{\calU} \sum_{b} 
  \dbra{b} {\calU} \dket{\bar{\calE}(\rho)} \sum_{\calB}  \Prob{{\calB|\calU}} 
  \dbra{\bar{\calB}^{\dagger}(O)}\calM^{-1}\calU^\dagger \dket{b}^2, 
\end{eqnarray}
where $\bar{\calE} = \calU^\dagger \calE \calU$ and 
$\bar{\calB} = \calM^{-1}\calU^\dagger \calB \calU \calM$. 
Let $g = \min_{\calU} \Prob{\calI|\calU}$ be the minimum conditional probability of the identity channel $\calI$ over all $\calU$, and denote
\begin{eqnarray}
\operatorname{Pr}^{\prime} (\calI\vert\calU) = \frac{\Prob{\calI|\calU} - g}{1 - g}, \\ \operatorname{Pr}^{\prime} (\calB\vert\calU) = \frac{\Prob{\calB|\calU}}{1 - g},\forall \calB\neq \calI, 
\end{eqnarray}
which satisfies $\sum_{\calB}\operatorname{Pr}^{\prime}(\calB\vert\calU) = 1$ for all $\calU$. Note that $\calB = \calI$ is also included in the summation over $\calB$. Then, we can write $\norm{O}^{2}_{\text{NS1}}$ as
\begin{eqnarray}
  \norm{O}_{\text{NS1}}^2 &=& g \left(\bbE_{\calU} \sum_{b} 
  \dbra{b} {\calU} \dket{\bar{\calE}(\rho)}   \sum_{\calB}
  \dbra{O}\calM^{-1}\calU^\dagger \dket{b}^2 \right) \nonumber\\
  && + (1 - g) \left(\bbE_{\calU} \sum_{b} 
  \dbra{b} {\calU} \dket{\bar{\calE}(\rho)}   \sum_{\calB} \operatorname{Pr}^{\prime}(\calB\vert\calU) 
  \dbra{\bar{\calB}^{\dagger}(O)}\calM^{-1}\calU^\dagger \dket{b}^2 \right) \\
  & \leq & g \left(\bbE_{\calU} \sum_{b} 
  \dbra{b} {\calU} \dket{\bar{\calE}(\rho)}   \sum_{\calB}
  \dbra{O}\calM^{-1}\calU^\dagger \dket{b}^2 \right) \nonumber\\
  && + (1 - g) h \left(\bbE_{\calU} \sum_{b} 
  \dbra{b} {\calU} \dket{\bar{\calE}(\rho)}   \sum_{\calB} \operatorname{Pr}^{\prime}(\calB) 
  \dbra{\bar{\calB}^{\dagger}(O)}\calM^{-1}\calU^\dagger \dket{b}^2 \right), 
\end{eqnarray}
where $\operatorname{Pr}^\prime{(\calB)} = \bbE_{\calU}\operatorname{Pr}^\prime{(\calB\vert \calU)}$ is the marginal distribution of $\calB$ and the factor $h$ is
\begin{equation}
  h = \max_{\calU, \calB} \frac{\operatorname{Pr}^\prime{(\calB\vert \calU)}}{\operatorname{Pr}^\prime{(\calB)}}.
\end{equation} 
Taking the maximum over $\bar{\calE}(\rho)$ for the first expectation, and the maximum over both $\bar{\calE}(\rho)$ and $\calB$ for the second expectation, we obtain that
\begin{eqnarray}
    \norm{O}_{\text{NS1}}^2 &\leq& g \norm{O}_{\text{S}}^2 + (1 - g) h \max_{\bar{\calB}} \norm{\bar{\calB}^\dagger(O)}_{\text{S}}^2.
\end{eqnarray} 
In the following, we discuss the factors $g$ and $h$, and the deviated observable $\bar{\calB}^\dagger(O)$.

The factor $g$ quantifies the effect of noise severity on the shadow norm. In the noiseless case, $g = 1$, and we recover the ideal shadow norm $\norm{O}_{\text{NS1}} =  \norm{O}_{\text{S}}$. In fact,  $g \approx 1-p+o\left(p^2\right)$ for small $p$ which we derive in the following. Consider $\calE = (1-p)\calI + p\calN$ is the noise channel attached to $\calU$ with minimum error rate $p$. The approximate noise inverse is $\calE_0^{-1} = \frac{1}{(1-p)^2} \left[ (1-p)\calI - p \calN \right]$, which results in $\Prob{\calI|\calU} = 1 - p$ and is able to mitigate the noise to the level of $o(p^2)$: $\calE_0^{-1}\calE = \calI - (\frac{p}{1-p})^2 \calN$.

The factor $h$ quantifies the ceiling of the correlation between $\calU$ and $\calB$. In the best case, where $\calB$ is independent of $\calU$, we have $h = 1$. In the worst case, where $\calB$ is determined by $\calU$ and different $\calB$ corresponds to different $\calU$, $h = |\bbU|$. In fact, $\log(h)$ is an upper bound on the mutual information between $\calU$ and $\calB$:
\begin{equation}
  I(\calU; \calB) = \sum_{\calU, \calB} \Prob{\calU, \calB} \log\frac{\Prob{\calB | \calU}}{\Prob{\calB}} \leq \sum_{\calU, \calB} \Prob{\calU, \calB} \log\frac{\operatorname{Pr}^\prime(\calB | \calU)}{\operatorname{Pr}^\prime(\calB)} \leq \log(h), 
\end{equation} where we have used (assume $\Prob{\calI} \neq g$, which is valid since it means $\Prob{\calI|\calU}$ are the same for all $\calU$) 
\begin{equation}
  \frac{\operatorname{Pr}^\prime(\calU | \calB)}{\operatorname{Pr}^\prime(\calB)} = \frac{\Prob{\calU | \calB}}{\Prob{\calB}}, \;
  \frac{\operatorname{Pr}^\prime(\calI | \calB)}{\operatorname{Pr}^\prime(\calB)} = \frac{\Prob{\calI | \calB} - g}{\Prob{\calI} - g} \geq \frac{\Prob{\calI | \calB}}{\Prob{\calB}}. 
\end{equation}
Note that, with the above equations, we can show that the relation between $h$ and the distribution $\Prob{\calU, \calB}$ is given by
\begin{equation}
  h=\max \left\{\max _{\calU, \calB}\left\{\frac{\operatorname{Pr}(\calB \mid \calU)}{\operatorname{Pr}(\calB)}\right\}, \max_{\calU}\left\{\frac{\operatorname{Pr}(\calI \mid \calU)-g}{\operatorname{Pr}(\calI)-g}\right\}\right\}. 
\end{equation}

The upper bound of the shadow norm of the deviated observable $\norm{\bar{\calB}^\dagger(O)}_{\text{S}}$ is the same as that of the original observable $\norm{O}_{\text{S}}$ for both random Clifford measurement and random Pauli measurement under the Pauli noise model.
Since $\bar{\calB} = \calM^{-1}\calU^\dagger \calB \calU \calM$ and $\calM$ is self-adjoint, the deviated observable reads $\bar{\calB}^\dagger(O) = \calM \calU^\dagger \calB^\dagger \calU \calM^{-1}(O)$. 
For random Clifford measurement, the map reads 
  \begin{equation}
    \calM (O) = \frac{1}{2^n + 1} \left( O + \Tr(O) \bbI \right),
  \end{equation} 
which commutes with any trace-preserving $(\Tr(\calC(O)) = \Tr(O))$ and unital channel ($\calC(\bbI) = \bbI$), i.e. $\calM \calC (O) = \calC \calM (O)$. Both $\calU$ (Clifford channel) and $\calB$ (Pauli channel) commute with $\calM$, so we have $\bar{\calB}^\dagger(O) =\calU^\dagger \calB^\dagger \calU(O)$. Using that $\norm{O}_{\text{S}}^2 \leq \Tr(O^2)$~\cite{Huang2020a}, we have 
\begin{equation}
  \norm{\bar{\calB}^\dagger(O)}_{\text{S}}^2 = \norm{\calU^\dagger \calB^\dagger \calU(O)}_{\text{S}}^2 \leq
  \Tr ((\calU^\dagger \calB^\dagger \calU(O))^2) = \Tr(O^2).   
\end{equation} 
For random Pauli measurements, we consider that $O=O_1 \otimes \cdots \otimes O_k \otimes \mathbb{I}^{\otimes(n-k)}$. Since the map $\calM$ is now a tensor product of local channels, and $\calU^\dagger \calB \calU$ is a tensor product of local Pauli channels, the locality of $\bar{\calB}^\dagger(O)$ is the same as that of $O$. As a result, the upper bound of the shadow norm of $\bar{\calB}^\dagger (O)$ is the same as that of $O$, i.e. $\norm{\bar{\calB}^\dagger(O)}_{\text{S}}^2 \leq 4^k \norm{O}_{\infty}^2$.

\subsection{Classical shadow with conventional error mitigation}
\label{app: variance of robust shadow with cPEC}
The procedure to implement error-mitigated classical shadow tomography with conventional error mitigation is: generate a random unitary channel $\calU$ from the ensemble $\bbU$ and a random basis channel $\calB$ according to the probability $\Prob{{\calB|\calU}}$ by the decomposition
\begin{equation}
  \calE^{-1}_{\calU} = \gamma_{_\calU} \sum_{\calB} \Prob{{\calB|\calU}} \operatorname{sgn}(\calB) \calB, \label{eq: err inv U}
\end{equation} and then measure the noisy circuit from the evolution $\calB \tilde{\calU} \dket{\rho}$ in the computational basis. Note that the general form of the noisy evolution can be written as $\tilde{\calU} = \calE_{\calU}\calU$, where the error channel $\calE_{\calU}$ may depend on $\calU$. Suppose we have generated $M$ noisy circuits $\{\calB_m\tilde{\calU}_m|m=1,\ldots,M\}$ and measured each circuit for $K$ shots, the estimator given by classical shadow with conventional error mitigation is 
  \begin{equation}
    \hat{o}_{_\text{CEM}} = \frac{1}{M} \sum_{m=1}^M \gamma_{m} \operatorname{sgn}(\calB_m) \frac{1}{K}\sum_{k=1}^{K} \dbra{O}\calM^{-1}\calU_m^\dagger \dket{b_{m,k}},
  \end{equation} 
  where $b_{m,k}$ is the $k$-th measurement outcome of the $m$-th circuit $\calB_m\tilde{\calU}_m\dket{\rho}$. $\hat{o}_{_\text{CEM}}$ is also an unbiased estimator:
  \begin{eqnarray}
    \bbE[\hat{o}_{_\text{CEM}}]  &=& \bbE_{\calU, \calB, b} [ \gamma_{\calU} \operatorname{sgn}(\calB)\dbra{O}  \calM^{-1} \calU^\dagger \dket{b}] \nonumber \\
     &=& \bbE_{\calU} \sum_{\calB} \Prob{{\calB|\calU}} \gamma_{_\calU} \operatorname{sgn}(\calB) \sum_{b} \dbra{O}\calM^{-1}\calU^\dagger \dket{b}\dbra{b}\calB \tilde{\calU} \dket{\rho} \nonumber\\
    &=& \bbE_{\calU} \sum_{\calB} \Prob{{\calB|\calU}}  \gamma_{_\calU} \operatorname{sgn}(\calB) \dbra{O}\calM^{-1}\calU^\dagger  \calM_{\bbZ} \calB \tilde{\calU} \dket{\rho}\nonumber \\
    &=& \bbE_{\calU} \dbra{O}\calM^{-1}\calU^\dagger  \calM_{\bbZ} \calE^{-1}_{\calU} \calE_{\calU} {\calU} \dket{\rho} \nonumber\\
    &=& \dbra{O}\calM^{-1}\calM \dket{\rho} = \dbraket{O}{\rho}. 
  \end{eqnarray}

  Denote $\gamma_{\text{max}} = \max_{\calU\in \bbU} \gamma_\calU$ in \eqref{eq: err inv U}, the variance of the above estimator is 
  \begin{eqnarray}
    \text{Var}[\hat{o}_{_\text{CEM}}] &=& \frac{1}{M} \text{Var} \left[\frac{\gamma_{\text{max}}}{K}\operatorname{sgn}(\calB)\sum_{k=1}^{K} \dbra{O}\calM^{-1}\calU^\dagger \dket{b_k}\right]  \nonumber\\ 
    &\leq&  \frac{\gamma_\text{max}^2}{M} \bbE\left[\left(\frac{1}{K}\sum_{k=1}^{K} \dbra{O}\calM^{-1}\calU^\dagger \dket{b_k}\right)^2\right] \nonumber\\ 
    &=&  \frac{\gamma_\text{max}^2}{M} \bbE\left[\frac{1}{K^2}\sum_{k=1}^{K} \dbra{O}\calM^{-1}\calU^\dagger \dket{b_k}^2+\frac{1}{K^2}\sum_{k\neq k^\prime}\dbra{O}\calM^{-1}\calU^\dagger \dket{b_k}\dbra{O}\calM^{-1}\calU^\dagger \dket{b_{k^\prime}}\right]  \nonumber\\ 
    &=& \frac{\gamma_\text{max}^2}{M}\left( \frac{1}{K} \bbE_{\calU}\sum_{\calB}\Prob{{\calB|\calU}}\sum_{b} \dbra{b}\calB\tilde{\calU}\dket{\rho} \dbra{O}\calM^{-1}\calU^\dagger \dket{b}^2  \right.  \nonumber\\
    && +  \left. \frac{K(K-1)}{K^2} \bbE_{\calU}\sum_{\calB}\Prob{{\calB|\calU}}\sum_{b,b^\prime} 
    \dbra{b}\calB\tilde{\calU}\dket{\rho}
    \dbra{O}\calM^{-1}\calU^\dagger \dket{b}  
    \dbra{b^\prime}\calB\tilde{\calU}\dket{\rho} \dbra{O}\calM^{-1}\calU^\dagger \dket{b^\prime}  \right) \nonumber\\
    &=& \frac{\gamma_\text{max}^2}{M}\left( \frac{1}{K} \bbE_{\calU}\sum_{\calB}\Prob{{\calB|\calU}}\sum_{b} \dbra{b}\calB\tilde{\calU}\dket{\rho} \dbra{O}\calM^{-1}\calU^\dagger \dket{b}^2  \right.  \nonumber\\
    &&  +  \left.  \left(1 - \frac{1}{K}\right) \bbE_{\calU}\sum_{\calB}\Prob{{\calB|\calU}}
    \left(\sum_{b} 
     \dbra{b}\calB\tilde{\calU}\dket{\rho}
    \dbra{F}\calM^{-1}\calU^\dagger \dket{b} \right)^2  \right)
  \end{eqnarray} 
In the above derivation, because there is only one random $\calB$ for each $\calU$ and $\operatorname{sgn}(\calB)^2 = 1$, the sign vanishes in the first equality. Define the noisy shadow norm and the noisy cross shadow norm respectively as 
  \begin{eqnarray}
    \norm{O}^2_{\text{NS} } &=& \max_{\rho} \bbE_{\calU}\sum_{\calB}\Prob{{\calB|\calU}}\sum_{b} \dbra{b}\calB\tilde{\calU}\dket{\rho} \dbra{O}\calM^{-1}\calU^\dagger \dket{b}^2 \label{eq: NS norm}\\
    \norm{O}^2_{\text{NXS} } &=& \max_{\rho} \bbE_{\calU}\sum_{\calB}\Prob{{\calB|\calU}}\left(\sum_{b} 
     \dbra{b}\calB\tilde{\calU}\dket{\rho}
    \dbra{O}\calM^{-1}\calU^\dagger \dket{b} \right)^2 , \label{eq: NXS norm}
  \end{eqnarray} 
  we simplify the expression of the variance to  
  \begin{equation}
    \text{Var}[\hat{o}_{\text{CEM}}] \leq \frac{\gamma_\text{max}^2}{M}\left[\frac{1}{K} \norm{O}^2_{\text{NS} } + \left(1- \frac{1}{K}\right)\norm{O}^2_{\text{NXS} } \right]. 
  \end{equation}

\section{Efficient simulation of universal basis channels}
\label{sec: Efficient simulation of universal basis channels}

Efficient simulation of error propagation and $\bbZ$-rows comparison relies on the simulation of the universal basis channels~\cite{Endo2018}. These channels comprise a set of Clifford operations along with the projectors $\dketbra{0}{0}$ and $\dketbra{1}{1}$. Here, we present the corresponding simulation algorithms, which are largely identical to those described in Ref.~\cite{Aaronson2004}, with the key distinction that we introduce an additional factor $g$ to store the trace of the simulated state.

A stabilizer state is represented by its stabilizer generators, which are Pauli operators with a global sign $\pm 1$. We use the isomorphism between stabilizer generators and $2n+1$-bit binary strings:
\begin{equation}
P_{\bsa}=(-1)^{a_0} \prod_{r=1}^n(-i)^{a_r a_{r+n}} Z_r^{a_r} X_r^{a_{r+n}}, 
\end{equation} where $\bsa = \{a_0, a_1, \cdots, a_{2n-1}, a_{2n}\}$ is a binary string of $2n+1$ bits, $Z_r = I^{\otimes (r-1)} \otimes Z \otimes I^{\otimes (n-r)}$ and $X_r = I^{\otimes (r-1)} \otimes X \otimes I^{\otimes (n-r)}$. We can represent a stabilizer state as 
\[
A =
\begin{pmatrix}
a^{(1)}_{0} & a^{(1)}_{1} & \cdots & a^{(1)}_{2n}\\[4pt]
a^{(2)}_{0} & a^{(2)}_{1} & \cdots & a^{(2)}_{2n}\\[4pt]
\vdots      & \vdots      & \ddots & \vdots\\[4pt]
a^{(n)}_{0} & a^{(n)}_{1} & \cdots & a^{(n)}_{2n}
\end{pmatrix}. 
\]
The Clifford gates update $A$ as following:
\begin{itemize}
  \item Hadamard gate ($H$) on $r$-th qubit: for all $l\in [n]$, let $a^{(l)}_0=a^{(l)}_0 + a^{(l)}_r a^{(l)}_{r+n} \operatorname{mod} 2$ and swap $a^{(l)}_r$ and $a^{(l)}_{r+n}$, 
  \item Phase gate ($S$) on $r$-th qubit: for all $l\in [n]$, let $a^{(l)}_0=a^{(l)}_0 + a^{(l)}_r a^{(l)}_{r+n} \operatorname{mod} 2$ and let $a^{(l)}_r=a^{(l)}_r + a^{(l)}_{r+n} \operatorname{mod} 2$, 
  \item CNOT on $r_1$-th (control) and $r_2$-th (target) qubits: for all $l\in [n]$, let $a^{(l)}_0=a^{(l)}_0+a^{(l)}_{r_2} a^{(l)}_{r_1+n}$ mod 2, $a^{(l)}_{r_1}=a^{(l)}_{r_1}+a^{(l)}_{r_2}$ mod 2, $a^{(l)}_{r_2+n}=a^{(l)}_{r_2+n}+a^{(l)}_{r_1+n}$ mod 2. 
    \item Computational basis measurement ($\dketbra{0}{0}+\dketbra{1}{1}$) on $r$-th qubit, the output state depends on the input state:
    \begin{itemize}
      \item (a) $\exists \bbL \subseteq [n]$, such that $a_{r+n}^{(l)}=1 $ for all $l \in \bbL$. We update A as following: performing row multiplication such that only one row ($l$-th) satisfying $a_{r+n}^{(l)}=1$, then letting $a^{(l)}_0 = 0 $ or $1$ with probability $\frac{1}{2}$. 
      \item (b) if $\forall l \in [n]$, $a_{r+n}^{(l)}=0$, the computational basis measurement does not change the state. We can determine if the outcome is $\dket{0}$ or $\dket{1}$ as follows: first, performing row Gaussian elimination on $A$, such that $\forall 1 \leq k, \exists t \in [n]$ such that $a_{t+n}^{(l)}=1$, and $a_{t+n}^{(l)}=0$ $\forall l>k, \forall t \in[n]$, second performing row Gaussian elimination on rows from $k+1$ to $n$, such that $a_r^{(n)}=1$, and $a_t^{(n)}=0\; \forall t \in [2n]\setminus\{r\}$, third the measurement outcome is $\dket{0}$ if $a_r^{(n)}=0$, and $\dket{1}$ if $a_r^{(n)}=1$. 
    \end{itemize}
  \item Projectors $\dketbra{0}{0}$ and $\dketbra{1}{1}$ on $r$-th qubit. The simulation is similar to the computational basis measurement, but we need an extra factor to record the trace of the resultant state, since the projectors are not trace preserving. We denote the factor by $g$, whose initial value is $1$. The output state also depends on the input state:
      \begin{itemize}
        \item (a) $\exists \bbL \subseteq [n]$, such that $a_{r+n}^{(l)}=1 $ for all $l \in \bbL$. This is similar to the case (a) in the above, except that, instead of ``letting $a^{(l)}_0 = 0 $ or $1$ with probability $\frac{1}{2}$", we let $g\leftarrow \frac{1}{2} g$ and let $a^{(l)}_0 = 0 $ if the projector is $\dketbra{0}{0}$, and $a^{(l)}_0 = 1 $ if the projector is $\dketbra{1}{1}$. 
        \item (b) if $\forall l \in [n]$, $a_{r+n}^{(l)}=0$. This is similar to the case (b) in the above, except that, instead of ``the measurement outcome is $\dket{0}$ if $a_r^{(n)}=0$, and $\dket{1}$ if $a_r^{(n)}=1$", we let $g=0$ if the projector is $\dketbra{0}{0}$ and $a^{(r)}_n=1$, and let $g=0$ if the projector is $\dketbra{1}{1}$ and $a^{(r)}_n=0$. 
      \end{itemize}
\end{itemize}

\section{Truncation of the total noise model}
\label{app: series expansion}

The aim of this section is to demonstrate that the noise model presented in Eq.~\eqref{eq: noise global-1} adequately captures the effects of gate-wise noise, and to further show that the global approach can be effectively applied to uncorrelated noise models—for this purpose, we adopt the truncated noise model as the input for the global approach. 

We consider a quantum circuit consisting of $N$ gates: $\calU = \calU_{N}\calU_{N-1}\cdots \calU_1$,  and its noisy implementation on a realistic quantum computer: $\tilde{\calU} = \calE_{N+1} \tilde{\calU}_{N} \tilde{\calU}_{N-1}\cdots\tilde{\calU}_1$, where $\tilde{\calU}_{i} = \calU_i \calE_i$ with $\calE_i$ being the noise on the $i$-th gate and  $\calE_{N+1}$ is the noise on the upcoming measurement after $\calU_N$. Without loss of generality, we consider each $\calU_i$ to be a two-qubit gate for all $i$, and $\calE_i$ is supported on the qubits that the corresponding $\calU_i$ acts on.

To better elaborate on the process of the noise model phase, we first introduce the process matrix formalism.
 Suppose $\calU_i$ acts on $l$-th and $(l+1)$-th qubits, the process matrix representation of $\calE_i$ is 
\begin{equation}
  \calE_i(\bullet) = \sum_{P,S\in \bbP^{(i)}} \chi_{P,S}^{(i)} P(\bullet) S, 
\end{equation}
where $\bbP^{(i)} =\calI^{\otimes (l-1)} \otimes \{\calI, \calZ, \calX, \calY\}^{\otimes 2}\otimes \calI^{\otimes (n-l-1)}$ and $\chi$ is Hermitian i.e. $\chi_{P, S}^{(i)} = \chi_{S, P}^{(i)*} $. We note that, for Pauli error models (e.g., Pauli twirling is implemented on $\calE_i$), we can express it as $\calE_i (\bullet)= \sum_{P\in\bbP^{(i})} \chi^{(i)}_{P,P} P(\bullet)P$, because Pauli twirling will transform $P(\bullet)S$ to zero if $P\neq S$: 
\begin{equation}
  \frac{1}{|\bbP|} \sum_{R\in \bbP} \chi_{P, S} RPR(\bullet)RSR = \frac{1}{|\bbP|}\left( \sum_{R\in \bbP} (-1)^{{<R, P> + <R, S>}} \right) P(\bullet)S , \label{eq: proc mat+twir}
\end{equation} where the coefficient $\sum_{R\in \bbP}(-1)^{<R, P> + <R, S>} = \sum_{R\in \bbP}(-1)^{<R, PS>} = 0$ if $PS\neq \bbI$.

Compared with PTM, the advantage of the process matrix lies in that the number of summation terms in the propagated noise model remains unchanged. To be more specific, propagating $\calE_i$ through $\bar{\calU}_i = \calU_{N+1}\calU_{N}\cdots\calU_i$ to the end of the circuit, the noise becomes 
\begin{equation}
  \bar{\calE}_{i} (\bullet) =  \sum_{P, S\in \bbP^{(i)}} \chi_{P,S}^{(i)} \bar{\calU}_i(P)(\bullet) \bar{\calU}_i (S). 
\end{equation} The total noise model is the concatenation of propagated gate errors, i.e. 
\begin{equation}
  \calE (\bullet) = \bar{\calE}_{N+1} \bar{\calE}_{N} \bar{\calE}_{N-1}\cdots\bar{\calE}_1 (\bullet), 
\end{equation} where $\bar{\calE}_{N+1} = \calE_{N+1} $. 

We now perform an order expansion on the propagated noise. In the truncation step, we only retain the first several orders of the expansion. 
The order expansion of the propagated total noise model is  
\begin{eqnarray}
  \calE (\bullet) &=& \left(\prod_{i=1}^{N+1} \chi_{\bbI, \bbI}^{(i)}\right) (\bullet) \nonumber \\
  &+& \sum_{j = 1}^{N+1} \left(\prod_{i\neq j} \chi_{\bbI, \bbI}^{(i)}\right) 
  \left(\sum_{P, S \in \bbP^{(j)}}^{\backsim} \chi_{P,S}^{(j)} \bar{\calU}_{j}(P)(\bullet)  \bar{\calU}_{j}(S) \right)\nonumber\\
  &+& \sum_{j_1 > j_2} \left(\prod_{i\notin\{j_1, j_2\}} \chi_{\bbI, \bbI}^{(i)} \right)\left(
  \sum_{P_{j_1}, S_{j_1} \in \bbP^{(j_1)}}^{\backsim} 
  \sum_{P_{j_2}, S_{j_2} \in \bbP^{(j_2)}}^{\backsim} 
  \chi_{P_{j_1},S_{j_1}}^{(j_1)} 
  \chi_{P_{j_2},S_{j_2}}^{(j_2)}
  \bar{\calU}_{j_1}(P_{j_1})\bar{\calU}_{j_2}(P_{j_2}) (\bullet) \bar{\calU}_{j_2}(S_{j_2})\bar{\calU}_{j_1}(S_{j_1})  \right) \nonumber\\
  &+& \cdots \nonumber \\
  &+& \sum_{j_1 > j_2 > \cdots > j_k} \left(\prod_{i\notin\{j_1, j_2, \cdots, j_k\}} \chi_{\bbI, \bbI}^{(i)} \right) \times \nonumber \\
  && \left(
  \sum_{P_{j_1}, S_{j_1} \in \bbP^{(j_1)}}^{\backsim} 
  \sum_{P_{j_2}, S_{j_2} \in \bbP^{(j_2)}}^{\backsim} \cdots 
  \sum_{P_{j_k}, S_{j_k} \in \bbP^{(j_k)}}^{\backsim}
  \left( \prod_{l=1}^{k} \chi_{P_{j_l}, S_{j_l}}^{(j_l)} \right)
  \left(\prod_{l=1}^{k} \bar{\calU}_{j_l}(P_{j_l})\right)\left(\bullet\right)
  \left(\prod_{l=k}^{1} \bar{\calU}_{j_l}(S_{j_l})\right)
  \right) \nonumber\\
  &+& \cdots ,
\end{eqnarray} where $\backsim$ over $\sum$ indicates that the case of $P=S=\bbI$ is excluded from the summation. Similar to Eq.~\eqref{eq: proc mat+twir}, Pauli twirling will transform $
  \left(\prod_{l=1}^{k} \bar{\calU}_{j_l}(P_{j_l})\right)\left(\bullet\right)
  \left(\prod_{l=k}^{1} \bar{\calU}_{j_l}(S_{j_l})\right)$ to zero if $
  \left(\prod_{l=1}^{k} \bar{\calU}_{j_l}(P_{j_l})\right)
  \left(\prod_{l=k}^{1} \bar{\calU}_{j_l}(S_{j_l})\right) \neq \bbI$, which results in $\calE$ becoming a Pauli error model. In the $k$-th order, only $k$ gates have noise, while the other gates are noise-free. As we will demonstrate next, when the total error rate is less than a certain value, the probability of $k$ gates failing simultaneously decreases exponentially with $k$. Therefore, we can approximate $\calE$ to a sufficient precision by retaining only the first few orders. To ensure normalization, all remaining terms can be treated as the identity map. 

Finally, we analyze the required truncation order to achieve a given precision and its corresponding computational complexity. We consider that gate errors are all two-qubit errors; then the $k$-th order term has $\begin{pmatrix}
   N+1 \\ k
\end{pmatrix} c^k \leq (c(N+1))^k$ sub-terms with the form of $P(\bullet)S$, where $c = 255$. Suppose $\chi_{\bbI, \bbI}^{(i)}\sim   1 - \epsilon$ and the others $\chi_{P, S}^{(i)}\sim \epsilon^\prime$, where $\epsilon^\prime \sim \frac{\epsilon}{15}$, the coefficient of each sub-term is $\sim(1-\epsilon)^{n-k}\epsilon^{\prime k}$. Therefore, the impact of the $k$-th order term is quantified by $\begin{pmatrix}
   N+1 \\ k
\end{pmatrix} c^k (1-\epsilon)^{n-k}\epsilon^{\prime k} \leq  (c\epsilon^\prime (N+1))^k$. If Pauli twirling is implemented, $c$ will become a number between 15 and 255 ($c=15$ if local Pauli twirling is used, and $c\geq 15$ if global Pauli twirling is used). We consider implementing local Pauli twirling. Then, denote the total error rate as $\epsilon_{tot} = (N+1)\epsilon$, the impact of the $k$-th order term is $\sim (\epsilon_{tot})^k$ which decreases exponentially with $k$ if $\epsilon_{tot}< 1$. Therefore, to approximate $\calE$ to an error smaller than $\delta$, we need to truncate the series expansion to the order of $k = {\log\delta}/{\log\epsilon_{tot}}$. Considering that the number of sub-terms in the $k$-th order term is smaller than $(c(N+1))^k$, the computing cost (the number of terms in the truncated noise model) is $(c(N+1))^{{\log\delta}/{\log\epsilon_{tot}}}$. 

\section{Simulating the noise inversion without explicit inversion}
\label{app: Simulating the noise inversion without explicit inversion}

Denote the overall noise model 
\begin{equation} 
    \calE = \sum_{\calB\in {\bbB}} C_{\calB} \calB = \eta \sum_{\calB\in \bbB} \Prob{\calB} \operatorname{sgn}(\calB) \calB, 
\end{equation} where the quasi-probabilities $C_{\calB}$ satisfy $\sum_{\calB\in {\bbB}} C_{\calB} =1$, $\eta = \sum_{\calB\in \bbB}|C_{\calB}|$, $\Prob{\calB} = |C_{\calB}|/\eta$, and $\operatorname{sgn}(\calB) = C_{\calB} / |C_{\calB}|$. 
Suppose we have access to the noise model $\calE$, meaning that the factor $\eta$ is known and we are able to sample $\operatorname{sgn}(\calB)\calB$ with probability $\Prob{\calB}$. We highlight that the ``access" does not require the explicit value of $\Prob{\calB}$, but only an efficient encoding of the probabilities, such as a neural-network or tensor-network representation. 

We can simulate the noise inversion $\calE^{-1}$ without finding the explicit inverse of $\calE$. The derivation of our method is as follows. We rewrite the noise model as
\begin{equation}
  \calE = \eta(\Prob{\calI}\calI + \Prob{\calN} \calN), 
  \label{eq: noise I and N}
\end{equation}
where 
\begin{equation}
  \Prob{\calN} = \sum_{\calB\in \bbB\setminus \{\calI\}} \Prob{\calB}, \quad 
  \calN = \sum_{\calB \in {\bbB}\setminus\{ \calI\}}\frac{\Prob{\calB}}{\Prob{\calN}} \operatorname{sgn}(\calB) \calB.  
\end{equation}
Since $\Prob{\calI} + \Prob{\calN} =  1$, Eq.~\eqref{eq: noise I and N} implies the noise channel occurs with probability $\Prob{\calN}$.
We can write the truncated noise model as $\calE = C_\calI \calI + C_\calN \calN$. According to the Neumann series $(\calI - \calA)^{-1} = \sum_{l=0}^{\infty} \calA^{l} $ for bounded operator $\calA$, we can write the inversion of $\calE$ as
\begin{equation}
    \calE^{-1} = \frac{1}{\eta \Prob{\calI}} \sum_{l=0}^\infty \left( - \frac{\Prob{\calN}}{\Prob{\calI}} \calN  \right)^l. 
\end{equation} 
Notice that 
\begin{equation}
  \sum_{l = 0}^{\infty} \left[\frac{\Prob{\calN}}{\Prob{\calI}}\right]^l = \frac{1}{1 - \Prob{\calN}/\Prob{\calI}}, 
\end{equation}
we can write $\calE^{-1}$ as 
\begin{equation}
    \calE^{-1} = \frac{1}{\eta \left[\Prob{\calI} - \Prob{\calN}\right]} \sum_{l=0}^\infty \left[ 1 - \frac{\Prob{\calN}}{\Prob{\calI}}\right] \left[\frac{\Prob{\calN}}{\Prob{\calI}}\right]^l (-1)^l \calN^l. 
    \label{eq: inversion neumann decomp}
\end{equation} 
The above derivation ensures that the coefficients in the summation in Eq.~\eqref{eq: inversion neumann decomp} are probabilities. 
According to Eq.~\eqref{eq: inversion neumann decomp}, we are able to simulate the noise inversion $\calE^{-1}$ without explicitly inverting $\calE$ as follows: 
\begin{enumerate}
    \item Generate a random number $l$ with probability $\left[ 1 - \frac{\Prob{\calN}}{\Prob{\calI}}\right] \left[\frac{\Prob{\calN}}{\Prob{\calI}}\right]^l$, which can be done using inverse transform sampling. 
    \item Generate $l$ random signed basis channels $\operatorname{sgn}(\calB_1)\calB_1, \operatorname{sgn}(\calB_2)\calB_2, \dots, \operatorname{sgn}(\calB_l)\calB_l$ with probability $\frac{\Prob{\calB}}{\Prob{\calN}}$. The concatenation of them $\prod_{i=1}^l \operatorname{sgn}(\calB_i) \calB_i$ is an unbiased estimator of $\calN^l$. 
    \item $\frac{1}{\eta \left[\Prob{\calI} - \Prob{\calN}\right]}(-1)^l \prod_{i=1}^l \operatorname{sgn}(\calB_i) \calB_i$ is an unbiased estimator of $\calE^{-1}$. 
\end{enumerate} This method introduces a sampling overhead of
\begin{equation}
  \gamma = \frac{1}{\eta \left[\Prob{\calI} - \Prob{\calN}\right]} = \frac{1}{|C_{\calI}| - \sum_{\calB\in \bbB\setminus \{\calI\}} |C_{\calB}|}  
\end{equation}
If the compression procedure is employed, the sampling overhead becomes 
\begin{equation}
  \gamma^\prime = \frac{1}{|C^\prime_{\calI}| - \sum_{\calB\in \bbB^\prime\setminus \{\calI\}} |C^\prime_{\calB}|},  
\end{equation} where $\bbB^\prime \subseteq \bbB$. 
 We can find that $\gamma^\prime \leq \gamma$ because $|C_{\calI}| \leq |C^{\prime}_{\calI}|$ and  $\sum_{\calB\in \bbB\setminus \{\calI\}} |C^\prime_{\calB}| \leq \sum_{\calB\in \bbB\setminus \{\calI\}} |C_{\calB}|$.  

\section{Discussion on the compression ratio}
\label{app: discussion on the compression ratio}

\begin{figure}
  \begin{center}
    \includegraphics{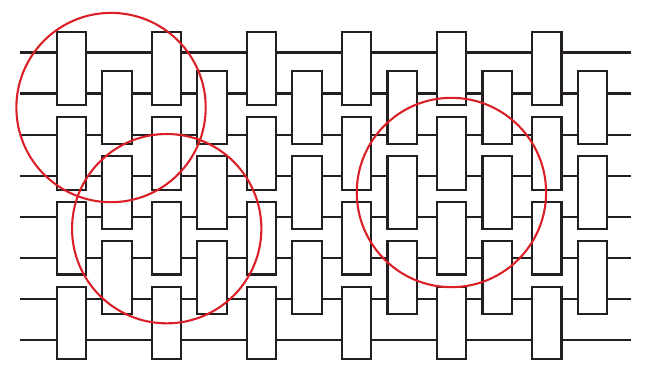}
    \caption{Regional inter-gate error auto-cancellation. The effect of an error on one gate has a chance of being canceled by an error on another gate. In each of the small regions, exemplified by the red circles, the probability of two errors canceling each other is independent of the circuit size.}
  \end{center}
\end{figure}

Here, we discuss the efficacy of noise compression. Specifically, we address the question of how the compression ratio scales with the number of qubits. The compression ratio is defined as the error rate before compression divided by the error rate after compression; a lower value indicates better efficacy.
We argue that for the local approach, when the number of qubits is large, the compression ratio can approach 1—meaning the method becomes increasingly ineffective as the qubit count grows. In contrast, for the global approach, the compression ratio remains high even as the number of qubits increases, ensuring sustained efficacy at scale. 

We first discuss the compression ratio for the local approach. Consider a Pauli error occurring on a gate within the quantum circuit. Such an error can be compressed, meaning it may be treated as the identity channel and does not require additional resource expansion to cancel, if the corresponding propagated error lies in $\bbZ$. For a sufficiently random quantum circuit, the probability of being in $\bbZ$ is $\frac{|\bbZ|}{|\bbP|} = \frac{1}{2^n}$. Therefore, for the local approach, the compression ratio drops to one exponentially with the number of qubits. 

This drawback of a low compression ratio in systems with a large number of qubits does not apply to the global approach. The reason is that the global approach accounts for two key factors: not only whether a propagated error falls within $\bbZ$, but also the auto-cancellation of inter-gate errors. Specifically, when errors occur on two distinct gates, there is a probability that these errors cancel each other out after propagation. Importantly, the probability of this auto-cancellation with errors on regional gates is independent of both the number of qubits and the circuit depth. This auto-cancellation of regional inter-gate errors thus supports the retention of a high compression ratio even in large-scale quantum circuits.

\section{{Partial Classical Noise Inversion}} 
\label{app: tractable intractable decomposition}

{
We work in the $\bbZ\oplus\bbZ^\perp$ partition of Pauli-transfer-matrix space, where $\bbZ=\{I,Z\}^{\otimes n}$ spans the $2^n$-dimensional $\bbZ$-subspace and $\bbZ^\perp$ its orthogonal complement. Writing $\Pi_\bbZ=\sum_{\sigma\in\bbZ}\dketbra\sigma\sigma$ and $\Pi_\perp=\id-\Pi_\bbZ$ for the two projectors, we block a map $\calE$ as
\begin{equation}
  \calE=\begin{pmatrix}\calE_{00}&\calE_{01}\\[2pt]\calE_{10}&\calE_{11}\end{pmatrix},\qquad
  \calE_{00}=\Pi_\bbZ\calE\Pi_\bbZ,\ \ \calE_{01}=\Pi_\bbZ\calE\Pi_\perp,\ \ \calE_{10}=\Pi_\perp\calE\Pi_\bbZ,\ \ \calE_{11}=\Pi_\perp\calE\Pi_\perp.
\end{equation}
The \textit{classically intractable part} of $\calE$ is its upper-right block $\calE_{01}$ [Fig.~\ref{fig: validity and optimality}(a)]: by Proposition~\ref{prop:valid}, CNI mitigates $\calE$ without bias if and only if $\calE_{01}=0$. Here, we prove Proposition~\ref{prop: Quantum-classical factorization} of the main text, which splits any $\calE$ into two factors $\calE_q$ and $\calE_c$, where $\calE_c$ is the classical (classically tractable) factor and $\calE_q$ is the quantum (classically intractable) factor. In addition, we provide detailed derivation of partial CNI, where $\calE_c$ is mitigated using CNI and $\calE_q$ is mitigated via a refined probabilistic error cancellation method resilient to gate-dependent noise. We also perform a rigorous analysis about the cost of partial CNI. } 


{
\begin{proof}[Proof of Proposition~\ref{prop: Quantum-classical factorization}]
Block multiplication gives
\begin{equation}
  \calE_c\calE_q=\begin{pmatrix}\calE_{00}&\calE_{00}\calE_{00}^{-1}\calE_{01}\\ \calE_{10}&\calE_{10}\calE_{00}^{-1}\calE_{01}+\calE_{11}-\calE_{10}\calE_{00}^{-1}\calE_{01}\end{pmatrix}=\calE.
\end{equation}
By construction the upper-right block of $\calE_c$ is zero, so $\calE_c$ obeys the validity condition Eq.~\eqref{eq: condition}. For trace preservation, recall that the PTM of any trace-preserving map has $I$-row (the first row) equal to $(1,0,\dots,0)$; applied to $\calE$ this forces the $I$-row of $\calE_{01}$ to vanish and the $I$-row of $\calE_{00}$ (hence of $\calE_{00}^{-1}$) to be $(1,0,\dots,0)$, so $\calN=\calE_{00}^{-1}\calE_{01}$ has vanishing $I$-row and $\calE_q=\calI+\calN$ is trace-preserving; consequently $\calE_c=\calE\calE_q^{-1}$ is trace-preserving as well. 
\end{proof}
}

{
Since $\calE_c$ is block-lower-triangular and $\calE_q$ is unit upper-triangular, their inverses share the same properties. According to Proposition~\ref{prop:valid},  $\calE_c$ can be mitigated with only classical post-processing (via CNI), while $\calE_q$ must be mitigated inside of the quantum circuit (via sampling random circuit). 
For $\calE_c$,  suppose the quasi-probability decomposition of its inverse is 
\begin{equation}
  \calE_c^{-1}=\sum_{\calB\in\bbB}Q_{\calB, c}\calB,
  \label{eq: quasi dep tmp}
\end{equation}
CNI enables us to mitigate $\calE_c$ with a measurement shot cost at $O(\gamma_c^2)$ with 
\begin{equation}
  \gamma_c = \sum_{\calB\in \bbB}\absLR{Q_{\calB, c}}. 
\end{equation} In practice, a classically simulable channel $\calB^\prime$ satisfying $\calB^\prime \calM_{\bbZ} = \calM_{\bbZ} \calB$ is sampled with probability $|Q_{\calB, c}|/\gamma_c$, realizing a classical channel $\calE_{c}^{\prime - 1}$ satisfying $\calE_{c}^{\prime - 1} \calM_{\bbZ} = \calM_{\bbZ}\calE_{c}^ {- 1}$. 
To mitigate $\calE_q$, we have to tackle the problem that the inserted error-mitigating gates (e.g., $\calB$) themselves are subjected to gate-dependent noise ($\calE_\calB$). To resolve this problem, we incorporate the noise into the quasi-probability decomposition:
\begin{equation}
  \calE_q^{-1}  = \sum_{\calB\in\bbB} Q_{\calB, q} \calE_{\calB} \calB,
  \label{eq: noisy quasi dep}
\end{equation}
which can be equivalently written as
\begin{equation}
  \sum_{\calB\in\bbB} Q_{\calB, q} \calE_{\calB} \calB  \calE_q = \calI.
\end{equation}
In the PTM view, the above becomes the linear system of equations:
\begin{equation}
    \sum_{\calB\in\bbB} Q_{\calB, q} \bigl[\calE_{\calB} \calB  \calE_q \bigr]_{i, j} = \bigl[\calI\bigr]_{i, j} \quad \forall i, j,
\end{equation} where $[\cdot]_{i, j}$ denote the $i,j$-th entry of a PTM. 
With prior knowledge on the noise models $\calE_q$ and $\calE_{\calB}$, we can obtain the quasi-probabilities $\{Q_{\calB, q}\}$ by solving the above linear system of equations. The solution to this linear system exists if the set of noisy operations $\{\calE_{\calB} \calB  \calE_q: \forall \calB\in \bbB\}$ is complete, and the solution is unique if the noisy operations are linearly independent. This existence and uniqueness holds when the error rate is small.  The sampling cost for mitigating $\calE_q$ is $O(\gamma_q^2)$ with 
\begin{equation}
  \gamma_q = \sum_{\calB\in\bbB}\absLR{Q_{\calB, q}}. 
\end{equation}

The overhead is sub-multiplicative under composition, i.e.
\begin{equation}
  \gamma\;\le\;\gamma_c\,\gamma_q. \label{eq: gamma submult}
\end{equation} 
This follows because the optimal sampling overhead $\gamma$ to realize $\calE^{-1}$ can be smaller than $\gamma_c\gamma_q = \sum_{\calB,\calB^\prime}\absLR{Q_{\calB, c}Q_{\calB^\prime, q}}$. Specifically, the product decomposition is rarely optimal, because the composed channels $\calB\calB^\prime$ may coincide or partially cancel. The factorization is nonetheless \textit{near}-optimal---the gap is bounded by a constant that depends only on the intractable shear $\calE_q$ and tends to one as the intractable part vanishes. We prove this in the following subsection. 
}

{

\subsection{The penalty on the cost due to the factorization}

\label{app: penalty}

The factorization of the total noise into classically tractable part $\calE_c$ and classically intractable part $\calE_q$ introduces a penalty in the total sampling overhead. In this subsection, we show that the penalty is upper bounded by $c = 1+O(\norm{\calN}_\diamond)$, which is small when $\calE_q$ is negligible compared to $\calE_c$. 

\begin{proposition}[Near-optimality of the factorization]
\label{prop: near optimal factorization}
Let $\Gamma(\calM)=\min\{c_++c_-:\calM=c_+\Lambda_+-c_-\Lambda_-,\ \Lambda_\pm\ \mathrm{CPTP},\ c_\pm\ge0\}$ be the simulation overhead of a map $\calM$. 
It is satisfied that 
\begin{equation}
  \Gamma(\calE^{-1})\;\le\;\Gamma(\calE_c^{-1})\,\Gamma(\calE_q^{-1})\;\le\; c\Gamma(\calE^{-1}), \label{eq: two sided gamma}
\end{equation}
where the constant $c:=\Gamma(\calE_q)\Gamma(\calE_q^{-1})\ge1$ and it depends only on the intractable shear $\calE_q=\calI+\calN$ (recall $\calE_q^{-1}=\calI-\calN$, as $\calN^2=0$). Moreover $c=1$ if and only if $\calE$ is classically tractable ($\calN=0$). Each factor's overhead is separately bounded by the diamond norm of the classically intractable shift,
\begin{equation}
  \Gamma(\calE_q^{-1})\;\le\;1+\norm{\calN}_\diamond,\qquad \Gamma(\calE_q)\;\le\;1+\norm{\calN}_\diamond,
  \label{eq: gamma_q diamond bound}
\end{equation}
so that circuit-sampling cost alone is controlled by $\norm{\calN}_\diamond$, and the penalty inherits
\begin{equation}
  c\;\le\;\big(1+\norm{\calN}_\diamond\big)^2\;=\;1+O(\norm{\calN}_\diamond). \label{eq: c bound}
\end{equation}
\end{proposition}

\begin{proof}
The left inequality of Eq.~\eqref{eq: two sided gamma} is Eq.~\eqref{eq: gamma submult}. For the right one, inverting $\calE=\calE_q\calE_c$ gives $\calE_c^{-1}=\calE^{-1}\,\calE_q$, so sub-multiplicativity of $\Gamma$ yields $\Gamma(\calE_c^{-1})=\Gamma(\calE^{-1}\,\calE_q)\le\Gamma(\calE^{-1})\,\Gamma(\calE_q)$; multiplying by $\Gamma(\calE_q^{-1})$ gives $\Gamma(\calE_c^{-1})\Gamma(\calE_q^{-1})\le\Gamma(\calE_q)\Gamma(\calE_q^{-1})\,\Gamma(\calE^{-1})=c\,\Gamma(\calE^{-1})$. Every quasi-probability decomposition of a trace-preserving map has coefficients summing to one, so $\Gamma\ge1$ and hence $c\ge1$. 

It remains to bound $\Gamma(\calE_q)=\Gamma(\calI+\calN)$ and $\Gamma(\calE_q^{-1})=\Gamma(\calI-\calN)$ through the diamond norm, defined for a Hermiticity-preserving map by
\begin{equation}
  \norm{\calN}_\diamond:=\sup_{\rho}\norm{(\calN\otimes\calI)(\rho)}_1 , \label{eq: diamond def}
\end{equation}
the supremum running over density operators $\rho$ of the system together with an arbitrary ancilla. The connection is the representation of the (trace-annihilating, Hermiticity-preserving) shift as a scaled difference of two channels,
\begin{equation}
  \calN=\nu\,(\Lambda_1-\Lambda_2),\qquad \Lambda_{1,2}\ \mathrm{CPTP},\quad \nu\ge0 . \label{eq: channel difference}
\end{equation}
Given Eq.~\eqref{eq: channel difference}, $\calI+\calN=(1+\nu)\,\tfrac{\calI+\nu\Lambda_1}{1+\nu}-\nu\,\Lambda_2$ exhibits $\calI+\calN$ as a quasi-probability mixture of the CPTP maps $\tfrac{\calI+\nu\Lambda_1}{1+\nu}$ (a convex combination of $\calI$ and $\Lambda_1$) and $\Lambda_2$, with overhead $1+2\nu$; the same holds for $\calI-\calN$, so $\Gamma(\calE_q),\Gamma(\calE_q^{-1})\le1+2\nu$ and $c\le(1+2\nu)^2$. We now show that the least admissible $\nu$ is exactly $\tfrac12\norm{\calN}_\diamond$.

{(i) $\nu\ge\tfrac12\norm{\calN}_\diamond$.} Apply the definition \eqref{eq: diamond def} to any representation \eqref{eq: channel difference}. For every $\rho$, the channels $\Lambda_{1,2}\otimes\calI$ send $\rho$ to density operators of unit trace norm, so the triangle inequality gives
\begin{equation}
  \norm{(\calN\otimes\calI)(\rho)}_1=\nu\,\norm{(\Lambda_1\otimes\calI)(\rho)-(\Lambda_2\otimes\calI)(\rho)}_1\le\nu\,(1+1)=2\nu .
\end{equation}
Taking the supremum over $\rho$ yields $\norm{\calN}_\diamond\le2\nu$.

{(ii) $\nu\le\tfrac12\norm{\calN}_\diamond$.} Let $Z\succeq0$ be any operator with $Z\succeq J(\calN)$, where $J(\calN)$ is the Choi operator of the shift, and let $\mathcal{K}_+,\mathcal{K}_-$ be the completely positive maps with Choi operators $Z$ and $Z-J(\calN)$, i.e., $J(\mathcal{K}_+)=Z, J(\mathcal{K}_-)=Z-J(\calN)$, so that $\mathcal{K}_+-\mathcal{K}_-=\calN$. Trace-annihilation, $\Tr_{\mathrm{out}}J(\calN)=0$, makes their output marginals coincide: $\Tr_{\mathrm{out}}J(\mathcal{K}_+)=\Tr_{\mathrm{out}}J(\mathcal{K}_-)=\Tr_{\mathrm{out}}Z$. Setting $\lambda=\norm{\Tr_{\mathrm{out}}Z}_\infty$ and adding to $\mathcal{K}_{\pm}$ a common CP map $\mathcal{R}$ with Choi $J(\mathcal{R}) = \tau\otimes(\lambda\id-\Tr_{\mathrm{out}}Z)\succeq0$ (any output state $\tau$) gives us the CPTP maps $\Lambda_{1,2}=(\mathcal{K}_\pm+\mathcal{R})/\lambda$, with $\calN=\lambda(\Lambda_1-\Lambda_2)$. Thus Eq.~\eqref{eq: channel difference} holds with $\nu\le\norm{\Tr_{\mathrm{out}}Z}_\infty$ for every feasible $Z$, and the semi-definite-programming characterization of the diamond norm, $\tfrac12\norm{\calN}_\diamond=\min_{Z\succeq0,\,Z\succeq J(\calN)}\norm{\Tr_{\mathrm{out}}Z}_\infty$, gives $\nu\le\tfrac12\norm{\calN}_\diamond$.

Combining (i) and (ii), the least admissible $\nu$ equals $\tfrac12\norm{\calN}_\diamond$, so $c\le(1+2\nu)^2=(1+\norm{\calN}_\diamond)^2$, which is Eq.~\eqref{eq: c bound}.
\end{proof}

\subsubsection{Expanded derivation of $\nu\le\tfrac12\norm{\calN}_\diamond$}

Recall that to a linear map $\Phi$ on operators of a $d$-dimensional system one associates its {Choi operator}
\begin{equation}
  J(\Phi)=(\Phi\otimes\calI)\Big(\textstyle\sum_{i,j=1}^{d}\ketbra{i}{j}\otimes\ketbra{i}{j}\Big),
\end{equation} 
an operator on $\calH_{\rm out}\otimes\calH_{\rm in}$ that depends {linearly} on $\Phi$. Conversely, the map is recovered from its Choi operator by the output state $\Phi(\rho)=\Tr_{\mathrm{in}}\big[J(\Phi)\,(\id_{\rm out}\otimes\rho^{T})\big]$, where $\rho^{T}$ is the transpose in the $\{|i\rangle\}$ basis and $\Tr_{\mathrm{in}}$ traces out the input system. Two standard facts translate the properties of a quantum channel into linear-algebraic conditions on $J(\Phi)$:
\begin{equation}
  \Phi\ \text{is completely positive}\iff J(\Phi)\succeq0,\qquad
  \Phi\ \text{is trace-preserving}\iff \Tr_{\mathrm{out}}J(\Phi)=\id_{\rm in}, \label{eq: choi facts}
\end{equation}
where $\Tr_{\mathrm{out}}$ traces out the output system. To verify the trace-preservation equivalence, take the output partial trace of the Choi operator:
\begin{equation}
  \Tr_{\mathrm{out}}J(\Phi)=\sum_{i,j}\Tr\!\big[\Phi(\ketbra{i}{j})\big]\,\ketbra{i}{j} . \label{eq: choi tp proof}
\end{equation}
Trace preservation, $\Tr[\Phi(X)]=\Tr X$ for all $X$, is---by linearity and because $\{\ketbra{i}{j}\}$ is a basis of operators---equivalent to $\Tr[\Phi(\ketbra{i}{j})]=\Tr[\ketbra{i}{j}]=\delta_{ij}$ for all $i,j$. Substituting these coefficients into Eq.~\eqref{eq: choi tp proof} gives $\Tr_{\mathrm{out}}J(\Phi)=\sum_i\ketbra{i}{i}=\id_{\rm in}$; conversely, reading off the coefficients of $\ketbra{i}{j}$ from $\Tr_{\mathrm{out}}J(\Phi)=\id_{\rm in}$ returns $\Tr[\Phi(\ketbra{i}{j})]=\delta_{ij}$, whence $\Tr[\Phi(X)]=\Tr X$ for every $X$. A {channel} (CPTP map) is therefore one whose Choi operator is positive semi-definite and has output marginal equal to the identity. 

We now consider the Choi form of $\calN$. Because $\calN$ is Hermiticity-preserving (a Hermitian operator can be written as a linear combination of Pauli operators, while $\calN$ maps a Pauli operator either to zero or to another Pauli operator), $J(\calN)$ is a Hermitian operator. Because $\calN$ is trace-annihilating, i.e.\ $\Tr[\calN(X)]=0$ for every $X$, its Choi operator has {zero} output marginal:
\begin{equation}
  \Tr_{\mathrm{out}}J(\calN)=\sum_{i,j}\Tr\!\big[\calN(\ketbra{i}{j})\big]\,\ketbra{i}{j}=0 .
\end{equation}

{We now split $\calN$ into a difference of two CP maps.} Pick {any} Hermitian operator $Z$ obeying the two conditions
\begin{equation}
  Z\succeq0\qquad\text{and}\qquad Z\succeq J(\calN)\ \ \text{(equivalently } Z-J(\calN)\succeq0\text{)} . \label{eq: feasible Z}
\end{equation}
Such operators always exist---for instance $Z$ is the positive part of $J(\calN)$. Define two maps $\mathcal{K}_+$ and $\mathcal{K}_-$ by declaring their Choi operators to be
\begin{equation}
  J(\mathcal{K}_+)=Z,\qquad J(\mathcal{K}_-)=Z-J(\calN).
\end{equation}
By Eq.~\eqref{eq: feasible Z} both are positive semi-definite, so by Eq.~\eqref{eq: choi facts} both $\mathcal{K}_\pm$ are completely positive. Subtracting and using linearity of the Choi map,
\begin{equation}
  J(\mathcal{K}_+-\mathcal{K}_-)=Z-\big(Z-J(\calN)\big)=J(\calN)\quad\Longrightarrow\quad \mathcal{K}_+-\mathcal{K}_-=\calN .
\end{equation}
We have thus written $\calN$ as a difference of two CP maps. They are not yet channels because they need not be trace-preserving: their output marginals, using $\Tr_{\mathrm{out}}J(\calN)=0$ from Step~1, are
\begin{equation}
  \Tr_{\mathrm{out}}J(\mathcal{K}_+)=\Tr_{\mathrm{out}}Z,\qquad
  \Tr_{\mathrm{out}}J(\mathcal{K}_-)=\Tr_{\mathrm{out}}Z-\Tr_{\mathrm{out}}J(\calN)=\Tr_{\mathrm{out}}Z .
\end{equation}
Both equal the {same} positive operator $G:=\Tr_{\mathrm{out}}Z\succeq0$. 

{We now pad both maps with a common map to fix the output marginal to identity, thereby ensuring the trace preservation.} Let $\lambda:=\norm{G}_\infty$ be the largest eigenvalue of $G$; then $\lambda\id-G\succeq0$. Fix any output state $\tau$ ($\tau\succeq0$, $\Tr\tau=1$) and define a map $\mathcal{R}$ through the Choi operator
\begin{equation}
  J(\mathcal{R})=\tau\otimes(\lambda\id-G)\succeq0 ,
\end{equation}
which is CP (its Choi operator is a tensor product of positive operators, hence positive) and has output marginal $\Tr_{\mathrm{out}}J(\mathcal{R})=\Tr(\tau)\,(\lambda\id-G)=\lambda\id-G$. Now add the {same} map $\mathcal{R}$ to each of $\mathcal{K}_\pm$ and rescale by $1/\lambda$:
\begin{equation}
  \Lambda_1:=\frac{\mathcal{K}_++\mathcal{R}}{\lambda},\qquad \Lambda_2:=\frac{\mathcal{K}_-+\mathcal{R}}{\lambda}.
\end{equation}
Each $\Lambda_{1,2}$ is CP, being a positive multiple of a sum of CP maps. Its output marginal is, by linearity,
\begin{equation}
  \Tr_{\mathrm{out}}J(\Lambda_{1,2})=\frac{\Tr_{\mathrm{out}}J(\mathcal{K}_\pm)+\Tr_{\mathrm{out}}J(\mathcal{R})}{\lambda}=\frac{G+(\lambda\id-G)}{\lambda}=\id .
\end{equation}
By Eq.~\eqref{eq: choi facts}, $\Lambda_1$ and $\Lambda_2$ are therefore genuine channels (CPTP).

Adding the {same} $\mathcal{R}$ to both maps cancels in the difference:
\begin{equation}
  \Lambda_1-\Lambda_2=\frac{(\mathcal{K}_++\mathcal{R})-(\mathcal{K}_-+\mathcal{R})}{\lambda}=\frac{\mathcal{K}_+-\mathcal{K}_-}{\lambda}=\frac{\calN}{\lambda}\quad\Longrightarrow\quad \calN=\lambda\,(\Lambda_1-\Lambda_2).
\end{equation}
This is precisely the channel-difference representation \eqref{eq: channel difference}, realized with $\nu=\lambda=\norm{\Tr_{\mathrm{out}}Z}_\infty$. Since the construction works for {every} $Z$ obeying Eq.~\eqref{eq: feasible Z}, the least admissible $\nu$ satisfies $\nu\le\norm{\Tr_{\mathrm{out}}Z}_\infty$ for all such $Z$.

Minimizing the bound over all feasible $Z$ gives
\begin{equation}
  \nu\;\le\;\min_{Z\succeq0,\ Z\succeq J(\calN)}\norm{\Tr_{\mathrm{out}}Z}_\infty .
\end{equation}
The quantity on the right is a semi-definite program, and it is exactly half the diamond norm of $\calN$ (proved explicitly in the next subsubsection):
\begin{equation}
  \min_{Z\succeq0,\ Z\succeq J(\calN)}\norm{\Tr_{\mathrm{out}}Z}_\infty=\tfrac12\norm{\calN}_\diamond , \label{eq: diamond sdp}
\end{equation}
which is the standard semi-definite-programming formula for the diamond norm of a (here trace-annihilating) Hermiticity-preserving map (illustrated in the following subsection). Hence $\nu\le\tfrac12\norm{\calN}_\diamond$, as claimed. 

\subsubsection{The semi-definite programming for diamond norm}

We now record the semi-definite-programming (SDP) characterization of the diamond norm and show that Eq.~\eqref{eq: diamond sdp} is its specialization to a trace-annihilating shift. For an arbitrary Hermiticity-preserving map $\Phi$ with Choi operator $J(\Phi)$, Watrous~\cite{Watrous} established that the diamond norm is the value of the SDP
\begin{equation}
  \norm{\Phi}_\diamond=\min_{Z_0,Z_1\succeq0}\Big\{\tfrac12\norm{\Tr_{\mathrm{out}}Z_0}_\infty+\tfrac12\norm{\Tr_{\mathrm{out}}Z_1}_\infty:\ \begin{pmatrix}Z_0&-J(\Phi)\\[2pt]-J(\Phi)^\dagger&Z_1\end{pmatrix}\succeq0\Big\}, \label{eq: watrous block sdp}
\end{equation}
the minimization running over positive-semi-definite operators $Z_0,Z_1$ on the output space. As a consistency check, for a channel ($\Phi$ trace-preserving) the choice $Z_0=Z_1=J(\Phi)$ is feasible---since $\big(\begin{smallmatrix}1&-1\\-1&1\end{smallmatrix}\big)\otimes J(\Phi)\succeq0$---and, by the trace-preservation identity in Eq.~\eqref{eq: choi facts}, attains the objective $\tfrac12(\norm{\id}_\infty+\norm{\id}_\infty)=1=\norm{\Phi}_\diamond$.

Our shift $\calN$ is Hermiticity-preserving, $J(\calN)=J(\calN)^\dagger$, and trace-annihilating, $\Tr_{\mathrm{out}}J(\calN)=0$; under these two properties Eq.~\eqref{eq: watrous block sdp} collapses to Eq.~\eqref{eq: diamond sdp} as follows. We first symmetrize Eq.~\eqref{eq: watrous block sdp}. Both the objective and the constraint of Eq.~\eqref{eq: watrous block sdp} are invariant under the exchange $Z_0\leftrightarrow Z_1$, so by convexity an optimum with $Z_0=Z_1=:Z_{\mathrm{sym}}$ exists. Conjugating the constraint by a Hadamard rotation block-diagonalizes it,
\begin{equation}
  \begin{pmatrix}Z_{\mathrm{sym}}&-J(\calN)\\[2pt]-J(\calN)&Z_{\mathrm{sym}}\end{pmatrix}\succeq0\iff Z_{\mathrm{sym}}\succeq J(\calN)\ \ \text{and}\ \ Z_{\mathrm{sym}}\succeq -J(\calN),
\end{equation}
whence the SDP reduces to the symmetric two-sided form
\begin{equation}
  \norm{\calN}_\diamond=\min\big\{\norm{\Tr_{\mathrm{out}}Z_{\mathrm{sym}}}_\infty:\ Z_{\mathrm{sym}}\succeq J(\calN),\ Z_{\mathrm{sym}}\succeq -J(\calN)\big\}. \label{eq: sdp two sided}
\end{equation} Then, by substituting $  Z=\tfrac12\bigl(Z_{\mathrm{sym}}+J(\calN)\bigr)$ into the above, we arrive at
\begin{equation}
    \norm{\calN}_\diamond=\min\big\{2\norm{\Tr_{\mathrm{out}}Z}_\infty:\ Z\succeq 0,\ Z\succeq J(\calN)\big\},
\end{equation}
which is precisely Eq.~\eqref{eq: diamond sdp}. Note that we have used trace annihilation, $\Tr_{\mathrm{out}}J(\calN)=0$, which gives $\Tr_{\mathrm{out}}Z_{\mathrm{sym}}=2\,\Tr_{\mathrm{out}}Z$.



Equation~\eqref{eq: two sided gamma} sandwiches $\Gamma(\calE_c^{-1})\Gamma(\calE_q^{-1})$ between $\Gamma(\calE^{-1})$ and $c\,\Gamma(\calE^{-1})$: the block-$UL$ factorization is optimal up to the factor $c$, and, crucially, $c$ is governed {solely} by the classically intractable shear $\calE_q$. However large the tractable cost $\gamma_c$, the split wastes nothing beyond the penalty $c=1+O(\norm{\calN}_\diamond)$. {In the case that noise has only a small intractable part ($\norm{\calN}_\diamond\ll1$), one has $\Gamma(\calE_c^{-1})\Gamma(\calE_q^{-1})\approx\Gamma(\calE^{-1})$, and so isolating the intractable cost $\Gamma(\calE_q^{-1})$---the irreducible price of the classically intractable block---comes essentially for free.} 

}

{

\subsection{Variance of Partial CNI}
\label{app: hybrid variance}


We consider implementing CNI with partial quantum sampling using three nested sampling parameters: $M$, the number of random circuits (independent draws of the random circuit specified by $\calB_q$, each a distinct physical circuit); $K$, the number of measurement shots taken on each of the $M$ circuits; and $L$, the number of classical basis-channel instances $\calB_c$ drawn per measurement shot. Labeling the draws $\calB_q^{(m)}$, the outcomes $b_{m,k}$, and the classical instances $\calB_c^{(m,k,l)}$, the hybrid estimator is
\begin{equation}
  \hat f_{\mathrm{tot}}=\frac1M\sum_{m=1}^M \underbrace{\gamma_q\operatorname{sgn}\!\big(Q_{\calB_q^{(m)}}\big)}_{\text{weight of random circuit}}\;\frac1K\sum_{k=1}^K\;\frac1L\sum_{l=1}^L\;\underbrace{\gamma_c\operatorname{sgn}\!\big(Q_{\calB_c^{(m,k,l)}}\big)\dbra F \calB_c^{(m,k,l)}\dket{b_{m,k}}}_{\text{classical CNI estimator }} .
\end{equation}
$\hat f_{\mathrm{tot}}$ is unbiased, $\bbE[\hat f_{\mathrm{tot}}]=f$: averaging over $\calB_q$ realizes $\calE_q^{-1}$ in the circuit (Eq.~\ref{eq: noisy quasi dep}), leaving effective noise $\calE_c$, which is inverted by the CNI average over $\calB_c$ (Eq.~\ref{eq: quasi dep tmp}).

We now derive the variance bound of Theorem~\ref{th: coarse grained variance}. We follow the same steps as the variance of CNI-based shadow estimation, Eqs.~\eqref{eq: var cni shadow 1}--\eqref{eq: var cni shadow 4}. Collect the $L$ classical instances of shot $(m,k)$ into the CNI inverse estimator
\begin{equation}
  \hat{\calE}_{c;\,m,k}^{-1}=\frac1L\sum_{l=1}^L\gamma_c\operatorname{sgn}\!\big(Q_{\calB_c^{(m,k,l)}}\big)\,\calB_c^{(m,k,l)} ,
\end{equation}
so the estimator is $\hat f_{\mathrm{tot}}=\frac1M\sum_m \gamma_q\operatorname{sgn}\!\big(Q_{\calB_q^{(m)}}\big)\frac1K\sum_k\dbra F\hat{\calE}_{c;m,k}^{-1}\dket{b_{m,k}}$. The $M$ random circuits are i.i.d.; bounding the variance by the second moment, 
\begin{align}
  \Var{\hat f_{\mathrm{tot}}}
  &=\frac1M\,\Var{\gamma_q\operatorname{sgn}\!\big(Q_{\calB_q}\big)\frac1K\sum_{k=1}^K\dbra F\hat{\calE}_{c;m,k}^{-1}\dket{b_{m,k}}}
  \ \le\ \frac{\gamma_q^2}{M}\,\bbE\!\left[\Big(\frac1K\sum_{k=1}^K\dbra F\hat{\calE}_{c;m,k}^{-1}\dket{b_{m,k}}\Big)^{\!2}\right] \nonumber\\[4pt]
  &\le\frac{\gamma_q^2}{M}\left\{\frac1{K^2}\sum_{k=1}^K\bbE\!\left[\dbra F\hat{\calE}_{c;m,k}^{-1}\dket{b_{m,k}}^2\right]
  +\frac1{K^2}\sum_{k\neq k'}\bbE\!\left[\dbra F\hat{\calE}_{c;m,k}^{-1}\dket{b_{m,k}}\,\dbra F\hat{\calE}_{c;m,k'}^{-1}\dket{b_{m,k'}}\right]\right\} \nonumber\\[4pt]
  &=\frac{\gamma_q^2}{M}\,\frac1K\,\bbE\!\left[\dbra F\hat{\calE}_{c}^{-1}\dket{b}^2\right]
  +\frac{\gamma_q^2}{M}\Big(1-\frac1K\Big)\,\bbE\!\left[\dbra F\hat{\calE}_{c}^{-1}\dket{b}\,\dbra F\hat{\calE}_{c}^{\prime-1}\dket{b'}\right] ,
  \label{eq: hybrid var expand}
\end{align}
where we used $\operatorname{sgn}^2=1$ to replace the variance by the second moment, expanded the square over the $K$ shots into its $k=k'$ and $k\neq k'$ parts (which contain $K$ and $K(K-1)$ terms, hence the prefactors $1/K$ and $1-1/K$), and used that the shots and their classical instances are i.i.d.\ given the circuit specified by $\calB_q$. In the last line, $\hat{\calE}_c^{-1}$ and $\hat{\calE}_c^{\prime-1}$ are two independent CNI processes for different shots $b$ and $b^\prime$ of the same circuit.

We first focus on the {Diagonal ($k=k'$) term.} Expanding a single-shot estimator over its $L$ classical instances and separating $l=l'$ from $l\neq l'$,
\begin{align}
  \bbE\!\left[\dbra F\hat{\calE}_c^{-1}\dket b^2\right]
  &=\bbE\!\left[\Big(\frac1L\sum_{l=1}^L\gamma_c\operatorname{sgn}\!\big(Q_{\calB_c^{(l)}}\big)\dbra F\calB_c^{(l)}\dket b\Big)^{\!2}\right] \nonumber\\[4pt]
  &=\frac{\gamma_c^2}{L^2}\sum_{l=1}^L\bbE\!\left[\dbra F\calB_c^{(l)}\dket b^2\right]
  +\frac{\gamma_c^2}{L^2}\sum_{l\neq l'}\bbE\!\left[\operatorname{sgn}\!\big({\calB_c^{(l)}}\big)\operatorname{sgn}\!\big({\calB_c^{(l')}}\big)\dbra F\calB_c^{(l)}\dket b\,\dbra F\calB_c^{(l')}\dket b\right] \nonumber\\[4pt]
  &=\frac{\gamma_c^2}{L}\,\bbE\!\left[\dbra F\calB_c\dket b^2\right]
  +\Big(1-\frac1L\Big)\gamma_c^2\,\bbE\!\left[\dbra F \operatorname{sgn}\!\big({\calB_c}\big) \calB_c\dket b\,\dbra F \operatorname{sgn}\!\big({\calB_c^{^\prime}}\big)\calB_c'\dket b\right] .
\end{align}
Taking the expectation explicitly---the circuit $\calB_q$ is drawn with probability $\absLR{Q_{\calB_q}}/\gamma_q$, the outcome $b$ with the Born weight $\dbra b\calE_c\calE_q\calB_q\calU\dket\rho$, and each classical instance $\calB_c$ with probability $\absLR{Q_{\calB_c}}/\gamma_c$---and absorbing the quasi-probability signs into these $\ell_1$-normalized sampling weights defines the two semi-norms
\begin{align}
  \norm F_{\star,'}^2 &:= \sum_{\calB_q}\frac{\absLR{Q_{\calB_q}}}{\gamma_q}\sum_b\dbra b\calE_c\calE_q\calB_q\calU\dket\rho\;\sum_{\calB_c}\frac{\absLR{Q_{\calB_c}}}{\gamma_c}\dbra F\calB_c\dket b^2 , \label{eq: Fstar def}\\
  \norm F_{\circ,'}^2 &:= \sum_{\calB_q}\frac{\absLR{Q_{\calB_q}}}{\gamma_q}\sum_b\dbra b\calE_c\calE_q\calB_q\calU\dket\rho\;\Big(\sum_{\calB_c}\frac{\absLR{Q_{\calB_c}}}{\gamma_c} \operatorname{sgn}\!\big({\calB_c}\big)  \dbra F\calB_c\dket b\Big)^{\!2} , \label{eq: Fcirc def}
\end{align}
so that
\begin{equation}
  \bbE\!\left[\dbra F\hat{\calE}_{c}^{-1}\dket{b}^2\right]\ \le\ \frac{\gamma_c^2}{L}\,\norm F_{\star,'}^2+\Big(1-\frac1L\Big)\gamma_c^2\,\norm F_{\circ,'}^2 . \label{eq: hybrid diag}
\end{equation}

We now focus on the {Off-diagonal ($k\neq k'$) term.} The two shots are independent given the circuit, and---the key point---the CNI average over the classical instances reproduces the exact inverse, $\bbE_{\calB_c}[\hat{\calE}_c^{-1}]=\calE_c^{-1}$. Writing out the two independent Born outcomes $b,b'$ of the circuit $\calB_q$,
\begin{align}
  \bbE\!\left[\dbra F\hat{\calE}_{c}^{-1}\dket b\,\dbra F\hat{\calE}_{c}^{\prime-1}\dket{b'}\right]
  &=\sum_{\calB_q}\frac{\absLR{Q_{\calB_q}}}{\gamma_q}\sum_{b,b'}\dbra b\calE_c\calE_q\calB_q\calU\dket\rho\,\dbra{b'}\calE_c\calE_q\calB_q\calU\dket\rho\;\dbra F\calE_c^{-1}\dket b\,\dbra F\calE_c^{-1}\dket{b'} \nonumber\\[4pt]
  &=\sum_{\calB_q}\frac{\absLR{Q_{\calB_q}}}{\gamma_q}\dbra F\calE_c^{-1}\Big(\sum_b\dketbra b b\Big)\calE_c\calE_q\calB_q\calU\dket\rho\;\dbra F\calE_c^{-1}\Big(\sum_{b'}\dketbra{b'}{b'}\Big)\calE_c\calE_q\calB_q\calU\dket\rho \nonumber\\[4pt]
  &=\sum_{\calB_q}\frac{\absLR{Q_{\calB_q}}}{\gamma_q}\big(\dbra F\MZ\calE_q\calB_q\calU\dket\rho\big)^2\ \eqqcolon \ \norm F_{\square,'}^2 , \label{eq: hybrid offdiag}
\end{align}
where we used $\sum_b\dketbra b b=\MZ$ and the CNI cancellation $\calE_c^{-1}\MZ\calE_c=\MZ$; the CNI overhead $\gamma_c$ {drops out} of this term entirely.

We finally obtain the \textit{Refined bound.} Substituting Eqs.~\eqref{eq: hybrid diag} and \eqref{eq: hybrid offdiag} into Eq.~\eqref{eq: hybrid var expand} gives
\begin{equation}
  \Var{\hat f_{\mathrm{tot}}}\ \le\ \frac{\gamma_q^2}{M}\left[\Big(1-\frac1K\Big) \norm F_{\square,'}^2+\frac1K\Big(\frac{\gamma_c^2}{L}\norm F_{\star,'}^2+\Big(1-\frac1L\Big)\gamma_c^2\norm F_{\circ,'}^2\Big)\right] ,
  \label{eq: hybrid variance refined}
\end{equation}
where the semi-norms obey $ \norm F_{\square,'}^2 \le \norm F_{\circ,'}\le\norm F_{\star,'}\le\norm F_\infty$. 
 Bounding every semi-norm by $\norm F_\infty$, using $\gamma_c^2/L+(1-1/L)\gamma_c^2=\gamma_c^2$ and $1-1/K\le1$, yields the \textit{coarse-grained bound} of Theorem~\ref{th: coarse grained variance},
\begin{equation}
  \boxed{\ \Var{\hat f_{\mathrm{tot}}}\ \le\ \frac{\gamma_q^2}{M}\Big(1+\frac{\gamma_c^2}{K}\Big)\norm F_\infty^2.\ } \label{eq: hybrid variance}
\end{equation}

We now discuss the coarse-grained variance bound given by Eq.~\eqref{eq: hybrid variance}.
First, {when classically intractable noise persists, it introduces a circuit-sampling overhead $\gamma_q$, which is only related to the classically intractable factor $\calE_q$ of the total noise model and is much smaller than $\gamma_c$ (and $\gamma$) whenever the intractable part is relatively small.} The variance floor $\gamma_q^2\norm F_\infty^2/M$ is the irreducible cost of quantum circuit sampling required to cancel $\calE_q$, and is reduced only by drawing more random circuits (more shots do not help). The (possibly large) $\gamma_c$ enters only the cheap term $\gamma_q^2\gamma_c^2\norm F_\infty^2/(MK)$, which can be reduced using more shots.

Second, {when the classically intractable part vanishes, the circuit-sampling cost vanishes and the bound recovers pure CNI.} The floor in Eq.~\eqref{eq: hybrid variance} is an artefact of upper-bounding the variance by the second moment; noticing that $\gamma_q = 1$ when the classically intractable part vanishes and retaining the $-f^2$ terms which are dropped in Eq.~\eqref{eq: hybrid var expand} when deriving the upper bound, one can find that the circuit sampling cost collapses to zero when $\calE_q=\calI$ (the draw $\calB_q = \calI$ is deterministic). In that limit Eq.~\eqref{eq: hybrid variance} reduces exactly to the pure-CNI variance of Theorem~\ref{th: perf CNI}: the noise-induced cost is imposed only on measurement shots $K$ instead of circuit samples $M$.

Third, {the bound is not tight.} Equation~\eqref{eq: hybrid variance} is designed to expose how the classically intractable noise introduces the circuit-sampling overhead $M$, and is intentionally coarse on the classical side. In particular, although $L$ does not appear explicitly in the coarse-grained bound in Eq.~\ref{eq: hybrid variance}, the refined bound Eq.~\ref{eq: hybrid variance refined} shows that the variance can still be reduced by increasing $L$ exactly as in pure CNI. 
}

\section{{Minimum-cost partial quantum sampling via direct optimization}}
\label{app: minimum cost in-circuit}

{
Proposition~\ref{prop: Quantum-classical factorization} gives a constructive route to compensate CNI with partial quantum sampling: form the factorization $\calE=\calE_c\calE_q$ and invert the intractable factor $\calE_q$ inside of the circuit through the quasi-probability decomposition of Eq.~\eqref{eq: noisy quasi dep}. This route is expensive: $\calE$ is a $4^n\times4^n$ Pauli transfer matrix, and the factorization requires the inverse $\calE_{00}^{-1}$ of its $2^n\times2^n$ $\bbZ$-block, which is generically dense even when $\calE$ is a propagated local noise channel. We give an alternative to circumvent this factorization and, moreover, returns the decomposition with cheapest circuit-sampling overhead. 
When $\calE$ is propagated local noise, or a truncated noise model (Appendix~\ref{app: series expansion}), only a small basis set $\bbB$ is involved and the optimal split can be searched for directly.

We split the total inverse into a quantum factor $\calE_q^{\star-1}$ and a classical factor $\calE_c^{\star-1}$,
\begin{equation}
  \calE^{-1}=\calE_q^{\star-1}\,\calE_c^{\star-1},
  \label{eq: opt split}
\end{equation}
the stars distinguishing these optimized factors from the canonical $\calE_c,\calE_q$ of Proposition~\ref{prop: Quantum-classical factorization}. The quantum factor is realized by sampling the noisy basis operations $\{\calE_\calB\calB\}_{\calB\in\bbB}$ (each a basis channel $\calB$ dressed by its gate-dependent noise $\calE_\calB$). The classical factor is realized by CNI and must therefore be classically tractable, i.e., its upper-right block must vanish so that it propagates through the computational-basis measurement (Proposition~\ref{prop:valid}, Eq.~\eqref{eq: condition}). So we decompose it over a set $\bbB_c=\{\calC_k\}$ of {tractable} maps. Crucially, such tractable maps are generally {not} individual universal (Clifford+projector) basis channels: those typically have non-vanishing upper-right blocks, and a generic tractable map is instead a real linear combination of them whose upper-right blocks cancel. For instance, on a single qubit the tractable map $\dketbra{I}{I}+\dketbra{X}{Z}$ is realized as $\tfrac12(\calH+\calS_y)$, the average of the Hadamard channel $\calH$ and the $\tfrac\pi2$ rotation $\calS_y$ about $Y$, neither of which is tractable on its own.  Minimizing the circuit-sampling overhead $\gamma_q$ over all valid splits gives
}
\begin{equation}
\begin{aligned}
  \gamma_q^\star\ =\ \min_{\{Q_{\calB,q}\},\,\{Q_{k,c}\}}\quad & \sum_{\calB\in\bbB}\absLR{Q_{\calB,q}}\\[2pt]
  \mathrm{s.t.}\quad & \calE_q^{\star-1}=\sum_{\calB\in\bbB}Q_{\calB,q}\,\calE_\calB\calB,\qquad \calE_c^{\star-1}=\sum_{\calC_k\in\bbB_c}Q_{k,c}\,\calC_k,\\[2pt]
   & \calE_q^{\star-1}\,\calE_c^{\star-1}\,\calE=\calI .
\end{aligned}
\label{eq: gamma_q sdp}
\end{equation}

{
Three features make Program~\eqref{eq: gamma_q sdp} tractable. First, decomposing $\calE_c^{\star-1}$ over the tractable basis $\bbB_c$ guarantees its tractability {by construction}, so no constraint on the full $2^n\times(4^n-2^n)$ off-diagonal block of any map need ever be evaluated. Second, the coupling is imposed in the form $\calE_q^{\star-1}\calE_c^{\star-1}\calE=\calI$, equivalent to Eq.~\eqref{eq: opt split} but involving only $\calE$ and the basis operations---never the dense inverse $\calE^{-1}$ or $\calE_{00}^{-1}$---and for propagated local or truncated noise these products act on a small support. Third, the objective is the $\ell_1$ norm of the in-circuit quasi-probabilities and the problem is posed entirely over the small bases $\bbB$ and $\bbB_c$, being linear in either factor's coefficients with the other held fixed; the coupling to the fixed, invertible $\calE$ excludes the trivial all-zero solution.
}

{
In the following proposition, we show that the circuit-sampling cost $\gamma_q^\star$ is no more than the cost of the explicit factorization given by Proposition~\ref{prop: Quantum-classical factorization}, and is typically cheaper.
}
\begin{proposition}[The optimization improves on the explicit factorization]
\label{prop: gamma_q opt}
Let $\gamma_q$ be the circuit-sampling overhead of the explicit factorization $\calE=\calE_c\calE_q$ of Proposition~\ref{prop: Quantum-classical factorization}---the cost of decomposing its intractable inverse $\calE_q^{-1}$ over $\{\calE_\calB\calB\}$. Whenever this factorization is realizable in the chosen bases, the optimum of Eq.~\eqref{eq: gamma_q sdp} satisfies
\begin{equation}
  \gamma_q^\star\ \le\ \gamma_q. 
  \label{eq: gamma_q opt bound}
\end{equation}
\end{proposition}

\begin{proof}
By Proposition~\ref{prop: Quantum-classical factorization}, $\calE^{-1}=\calE_q^{-1}\calE_c^{-1}$ with $\calE_c^{-1}$ tractable, hence representable over $\bbB_c$. Setting $(\calE_q^{\star-1},\calE_c^{\star-1})=(\calE_q^{-1},\calE_c^{-1})$, with $\calE_q^{-1}$ decomposed over $\{\calE_\calB\calB\}$ at cost $\gamma_q$, is therefore feasible for Eq.~\eqref{eq: gamma_q sdp} and attains objective value $\gamma_q$. The minimum can only be smaller, giving Eq.~\eqref{eq: gamma_q opt bound}.
\end{proof}

{
The inequality is strict whenever a tractable residual cheaper than the canonical $\calE_c$ exists. Proposition~\ref{prop: Quantum-classical factorization} fixes the split by demanding that $\calE_q$ be the unit-triangular shear; relaxing this---keeping only that $\calE_c^\star$ be tractable---leaves the gauge freedom $\calE_c\mapsto\calE_c\calG,\ \calE_q\mapsto\calG^{-1}\calE_q$ for any tractable $\calG$, and Eq.~\eqref{eq: gamma_q sdp} selects the gauge of least circuit-sampling cost. The optimization thus matches or undercuts the explicit factorization while working only with the small basis $\bbB$, never with the dense inverse $\calE_{00}^{-1}$.
}
\end{document}